\documentclass{article}
\usepackage{amsmath, amssymb, theorem, color}
\usepackage{tabularx}

\newtheorem{lemma}{Lemma}
\newtheorem{theorem}[lemma]{Theorem}

\newtheorem{observation}[lemma]{Observation}
\newtheorem{definition}[lemma]{Definition}

\newcommand{\proofstart}{{\bf Proof\hspace{2em}}}
\newcommand{\proofend}{\hspace*{\fill}\mbox{$\Box$}}
\newenvironment{proof}{\proofstart}{\proofend}

\newenvironment{prevlemma}[1]{\vspace*{1.5ex}\noindent{\bf Lemma #1} \itshape}{\normalfont \vspace*{1.5ex}}

\newcommand{\pr}{\mbox{\bf Pr}}

\newcommand{\inv}[1]{\frac{1}{#1}}
\newcommand{\interval}[1]{\left[\underline{ #1}, \overline{ #1}\right]}

\renewcommand{\a}{\alpha}
\renewcommand{\b}{\beta}

\renewcommand{\d}{\delta}

\newcommand{\e}{\epsilon}

\newcommand{\g}{\gamma}
\newcommand{\G}{\Gamma}

\newcommand{\z}{\zeta}

\renewcommand{\l}{\lambda}

\newcommand{\n}{\nu}

\newcommand{\s}{\sigma}

\renewcommand{\t}{\tau}

\newcommand{\smallo}{{\rm o}}
\newcommand{\bigo}{{\rm O}}

\newcommand{\bigth}{{\rm \Theta}}
\renewcommand{\to}{\rightarrow}
\newcommand{\E}{{\bf E}}
\newcommand{\ex}{{\bf E}}
\newcommand{\rme}{{\rm e}}
\newcommand{\hf}{\frac{1}{2}}

\def\cres{{\bf C-RES}}
\def\ngres{{\bf NG-RES}}
\def\I{{\cal I}}
\def\poly{{\rm poly}}
\def\e{\epsilon}
\def\a{\alpha}

\def\B{\Upsilon}

\begin{document}

\title{The Satisfiability Threshold for a Seemingly Intractable Random Constraint Satisfaction Problem}
\author{Harold Connamacher \\ Case Western Reserve University \\
\and Michael Molloy \\
University of Toronto}

\maketitle

\begin{abstract}
We determine the exact threshold of satisfiability
for random instances of a particular NP-complete constraint satisfaction problem (CSP).
This is the first random CSP model for which
we have determined a precise linear satisfiability threshold, and for which random instances
with density near that threshold appear to be computationally difficult.
More formally, it is the first random CSP model for which the satisfiability
threshold is known and which shares the following characteristics with random $k$-SAT for
$k \geq 3$.  The problem is NP-complete,  the
satisfiability threshold occurs when there is a linear number of clauses,
and a uniformly random instance with a linear number of clauses asymptotically almost
surely has exponential resolution complexity.
\end{abstract}

\section{Introduction}

Determining the satisfiability threshold for random $k$-SAT is a fundamental
problem that has received attention from several scientific communities.
(See, eg., \cite{ach01} for a survey of the area.)
The basic question is this:  does there exist a constant $c^*_k$ such that
a uniformly random instance of $k$-SAT with $n$ variables and $cn$ clauses
will be asymptotically almost surely (a.a.s.)\footnote{A property holds {\em asymptotically almost surely}
if its probability tends to 1 as the number of variables
tends to $\infty$.}
satisfiable if $c<c^*_k$
and a.a.s.\ unsatisfiable if $c>c^*_k$?
Neither the existence nor the location of $c^*_k$ is known for any
$k \geq 3$; Friedgut~\cite{fri99} proved that the
threshold for $k$-SAT is sharp, although the location of the threshold might be not be
at the same clause density for each value of $n$.
We have a tight asymptotic bound on the conjectured $c^*_k$:
$2^k \log 2 - \bigo(k)$~\cite{ach-per03}~$\leq c^*_k \leq 2^k \log 2$~\cite{fra-pau83}.
For 3-SAT, the current state of research has
$3.52$~\cite{haj-sor03,kap-kir-lal03}~$\leq c^*_3 \leq 4.4898$~\cite{diaz-etal08}.
From experimental evidence, the threshold for 3-SAT appears to be roughly
$4.2$~\cite{kir-sel94,sel-mit-lev96,cra-aut96}.

Research on the satisfiability threshold has extended to generalizations
of $k$-SAT such as the Schaefer \cite{sch78} generalizations:
1-in-$k$ SAT~\cite{ach-cht-ist-moo01}
(only one true literal per clause),
NAE-SAT~\cite{ach-cht-ist-moo01,ach-moo02}
(at least one true and one false literal per clause),
and XOR-SAT~\cite{dub-man02,mez-ric-zec03, dgmm} (each clause is an exclusive-or
rather than a disjunction), and more generally, to random constraint
satisfaction problems.

In a constraint satisfaction problem (CSP), we can allow variables to take
on values from a domain of size larger than 2, and we have more freedom as to
the types of constraints to use.
There have been various models proposed and studied for random
constraint satisfaction problems, eg.~\cite{ach-etal01,cre-dau03,mit02,mol02,mol-un}.
For such models, a major
goal has typically been to determine whether it has a sharp satisfiability threshold,
and if so, to determine its location. In addition, there has been a large body of
experimental studies \cite{gen-mac-pro-smi-wal01}
to find the approximate location of the satisfiability
threshold and to study the difficulty of solving random instances of various
CSP models.

The primary motivation for a large number of these studies was the discovery by Selman, et al~\cite{sel-mit-lev96}
that random instances of 3-SAT with clause density near the conjectured threshold are very difficult to solve.
Researchers who develop SAT-solvers study such problems  to test and improve their solvers and to
look for insights into what can make SAT instances computationally difficult.  When these researchers
expanded to CSP-solvers, they wanted to find thresholds for other random CSP models with the expectation
that this should provide a rich source of difficult instances.

Most of this research has been on models where the constraint- and domain-sizes are
constant.  There have been studies of models where one or both of these parameters grows
with $n$ (eg.~\cite{dye-fri-mol03,xu-li00,xu-li03,fri-wor01,fla03}). As discussed below,
such models tend to have a very different nature; eg., the thresholds typically occur when the number
of constraints is superlinear. So our focus in this discussion is on the case where these parameters are both constant.

Most random CSP models for which the satisfiability threshold is known are in P (eg.\
2-SAT~\cite{chv-ree92,fer92,goe96}, 3-XOR-SAT\cite{dub-man02, dgmm}), and hence do not provide
difficult instances. We know the satisfiability threshold for a few NP-complete
problems with constant sized domain and constraints, for example
1-in-$k$-SAT~\cite{ach-cht-ist-moo01}, and a mixture of 2-SAT and 3-SAT when
the number of clauses of size 3 is kept small~\cite{mon-zec-kir-sel-tro99},
and a model from \cite{mol-sal03}.
However, in each of these cases, a random instance
whose number of clauses is just below the threshold, is in some sense equivalent
to random 2-SAT. This allows us to prove that such problems are a.a.s.\ satisfiable by analyzing
a polynomial time algorithm that solves them with uniformly positive probability (w.u.p.p.)
\footnote{A property holds {\em with uniformly positive probability}
if as the number of variables
tends to $\infty$, the $\liminf$ of its probability is at least some positive constant.}.
Furthermore, instances with density slightly
above the threshold a.a.s.\ have short resolution proofs of unsatisfiability.
(This is straightforward to show for 1-in-$k$-SAT, and is proven in \cite{ach-bea-mol03,mol-sal03} for the other two.)
It is straightforward to design polynomial time algorithms
that make use of such proofs to a.a.s.\ recognize that such instances are unsatisfiable.  Thus, while these
problems are NP-hard, they are easy to solve when their density is near the satisfiability threshold.

In this paper, we introduce a particular NP-hard problem and determine its satisfiability threshold.
The problem is different from previous NP-hard problems whose thresholds are known in that we don't know
of any algorithm that can solve random instances whose density is near the threshold.  Furthermore,
we can prove that such instances a.a.s.\ have exponential resolution complexity, which implies that
resolution based algorithms will a.a.s.\ fail on unsatisfiable instances. (Virtually all complete\footnote{A solver is {\em complete} if it is able to recognize all satisfiable and all unsatisfiable instances.} CSP-solvers
that are commonly used in practice are resolution-based.)

This is the first random constraint satisfaction problem for which: (1) We have determined a constant
$c^*$ such that if the number of constraints is at most $(c^*-\e)n$ then it is a.a.s.\ satisfiable
and if the number of constraints is at least $(c^*+\e)n$ then it is a.a.s.\ unsatisfiable. (2) We don't know
of any polynomial time algorithm that will a.a.s.\ succeed on a random instance when the number of clauses
is $cn$ for $c$ arbitrarily close to $c^*$.  In short, we know the location of a satisfiability threshold around which
the instances actually appear to be difficult.

Our problem is called $(3,4)$-UE-CSP. We will define $(k,d)$-UE-CSP in the next section.  We use  $U^{(k,d)}_{n,m}$ to denote
a random instance with $n$ variables and $m$ constraints.  In Subsection \ref{sec:2core} we will specify
a precise constant $c^*=.917935...$   Our main theorem is:

\begin{theorem}\label{mt}
If $c<c^*$ then $U^{(3,4)}_{n,m=cn}$ is a.a.s.\ satisfiable.
If $c>c^*$ then $U^{(3,4)}_{n,m=cn}$ is a.a.s.\ unsatisfiable.
\end{theorem}

In \cite{con08} the first author proves
\begin{theorem}\label{thc} For every $d\geq 4$,
$(3,d)$-UE-CSP is NP-complete.  Furthermore, it is NP-complete under the restriction
that no two constraints share more than one variable.
\end{theorem}

The latter restriction is relevant here because random instances with a linear number of constraints
a.a.s.\ have that property.

In Section \ref{sec:resolution} we prove
\begin{theorem}\label{t2}
For any constant $c>0$, and any $k\geq3$, $d\geq2$, the
resolution complexity of $U^{(k,d)}_{n,m=cn}$  is a.a.s.\ $2^{\Theta(n)}$.
\end{theorem}

The proof of Theorem \ref{mt} follows along the lines of that of
Dubois and Mandler~\cite{dub-man02, dub-manj}, who determined the satisfiability threshold for $k$-XOR-SAT.\footnote{{Unfortunately, a full version of their paper never appeared, and the short versions provide only a sketch of the argument for $k=3$ and no details for $k=4$.  Ditzfelbinger et. al.\cite{dgmm} provide a full proof for all $k\geq3$, using a somewhat different technique.} } That problem is
in P - it can be solved using Gaussian elimination.  However, they did not make use of the fact that
it is in P in establishing the threshold.  Rather, they applied the second moment method to prove
that problems with density below the threshold are a.a.s.\ satisfiable.  The problem $(3,4)$-UE-CSP
is somewhat contrived in that it is designed to have the properties of $k$-XOR-SAT that allow the approach of
Dubois and Mandler to work, while at the same time being NP-complete.  The analysis is much more complicated
than that of \cite{dub-man02, dub-manj}.  In fact, we have to resort to a computer-aided proof which makes use of interval
analysis to rigorously determine the global maximum of a particular function.

The discussion above has focused on random CSP models in which the constraint size and
domain size are both constant.
Researchers have studied models where the domain size
grows with $n$~\cite{dye-fri-mol03,fri-mol03,smi01,xu-li00,xu-li03}
or where the constraint size grows
with $n$~\cite{fri-wor01,fla03}.
Exact thresholds are known for some such models~\cite{xu-li00,xu-li03,fri-wor01,fla03}.
In \cite{xu-li03} it is proven that the models of
\cite{xu-li00} and \cite{xu-li03} contain many problems that
a.a.s.\ have exponential tree resolution complexity,
and it is noted that the model of \cite{fri-wor01} contains problems that a.a.s.\ have
exponential resolution complexity.  But in all of these models, the satisfiability threshold
occurs when the number of clauses is superlinear in $n$,
and the structure of a random constraint satisfaction
problem with a superlinear number of clauses is very different from one
with a linear number of clauses.  Thus, these problems are of a very different nature than, say,
the random 3-SAT problem.

In summary, random $(3,4)$-UE-CSP is the first model of random constraint satisfaction problems
for which we know the exact satisfiability threshold and
which shares the following characteristics with random 3-SAT:  (a) it is NP-complete;
(b) there are a linear number of constraints at the threshold; (c) random instances near
the threshold a.a.s.\ have exponential resolution complexity; (d) the domain and constraint sizes are constant.
Random 3-XOR-SAT was close - it met all of these conditions except for (a).

Ours is also the first random CSP model with a known linear satisfiability threshold for which it is seems that random instances with density near that threshold are difficult to solve.
There is computational evidence that such problems really are difficult.  Since the appearance of a short preliminary
version of this paper~\cite{con-mol04},
random instances of $(3,4)$-UE-CSP with densities close to the satisfiability threshold have
been chosen as challenging problems in experimental studies
\cite{lew-sch-bec05, heu-mar06, heu-maa07, heu-maa08, hyv-jun-nie08}.
In fact, the study of \cite{alt-mon-zam07} indicated that the problems
become difficult when the density exceeds the threshold for the appearance of a 2-core in the
underlying hypergraph (see Section \ref{sec:uecsp-intro}).  Random instances of $(3,4)$-UE-CSP have also been
used as test cases for the
SAT solver competitions of the SAT 2004, 2005, 2007, and 2009 conferences~\cite{ber-sim04, sat, sat07, sat09}.
None of the solvers in the competition were able to solve a test instance on 1200 variables whose density
was at the satisfiability threshold.

{\bf Remark:} A short preliminary
version  of this paper appeared in a conference proceedings~\cite{con-mol04}.
There is an error in that version; specifically in the proof of Lemma 16, which corresponds
to Lemma \ref{lem:fmax} in this paper.  We replaced that proof with the
much longer computer-aided proof found here.

\section{Uniquely Extendible CSPs}
\label{sec:uecsp-intro}

We define a {\em clause} to be an ordered subset of variables and a {\em constraint}
to be a list of tuples of values that we may assign to the variables of the clause.
A {\em constraint satisfaction problem}  consists of a set of $n$ variables where each
variable has a non-empty domain of possible values,
a set of $m$ clauses, and  a constraint applied to each clause.
The goal is to find an assignment
to the variables such that the constraint on each clause is satisfied.  One common
restriction is that every variable must have the same domain of values; that will
be the case throughout this paper.
In keeping with SAT notation, we will denote an instance of a constraint
satisfaction problem as a {\em formula}. The {\em underlying hypergraph} of a CSP
is the hypergraph whose vertices are the variables and whose hyperedges are the clauses.

The inspiration for uniquely extendible CSPs comes from the proof of the satisfiability
threshold for 3-XOR-SAT~\cite{dub-man02, mez-ric-zec03, dgmm}.
In XOR-SAT, each clause is an exclusive-or of the literals,
rather than a disjunction. Thus each clause can be interpreted as a linear equation
mod 2, which is why it can be solved using Gaussian elimination.

As we mentioned above, their proof of the threshold
does not rely in any way on the Gaussian elimination algorithm.
Instead, they reduce the random formula to its {\em 2-core},
the unique maximal subformula where each variable
occurs in at least two clauses.
Then, the first  and second moment
methods are used to give coinciding upper and lower bounds on the
satisfiability threshold for the 2-core.
Standard calculations translate the satisfiability threshold for the 2-core
into a satisfiability threshold for 3-XOR-SAT.

The key property that makes this proof work
is that every constraint in XOR-SAT is
{\em uniquely extendible}:
\begin{definition} A constraint on $k$-variables is {\em uniquely extendible} if for every truth assignment
to any $k-1$ of its variables there is exactly one value that can be assigned to the remaining variable
so that the constraint will be satisfied. It is {\em at-most-one-extendible} if there is always at most one
such value and {\em at-least-one-extendible} if there is always at least one such value.
\end{definition}

It is the property that the constraints of $k$-XOR-SAT are at-least-one-extendible
that permits us to consider only the 2-core.
Consider the following procedure to find the 2-core:

\begin{quote}
\noindent{\bf CORE:}
While the formula has any variable that occurs in at most one clause,
choose an arbitrary such variable and delete it along with
any clause that contains it.
\end{quote}

The order in which variables are chosen to
be deleted is easily seen to be irrelevant in that it does not affect the
final output of the procedure. This proves that the 2-core is unique.

\begin{lemma}\label{lem:2coresat}
Let $F$ be an instance of a CSP such that
every constraint is at-least-one-extendible. Then $F$ is satisfiable
iff the 2-core of $F$ is satisfiable.
\end{lemma}

\begin{proof}
Clearly, if the 2-core of $F$ is unsatisfiable
then so is $F$.  Assume that the 2-core of $F$ is satisfiable.
Consider running CORE on
$F$, and suppose that the deleted variables are $x_1,x_2,...,x_t$ in that order.
Start with any satisfying assignment of the 2-core.  Now restore the deleted variables
in reverse order, i.e.\ $x_t,x_{t-1},...,x_1$, each time adding the variable
along with the at most one clause that was deleted when the variable was
deleted.  Because the constraint applied to the clause is at-least-one-extendible,
there is a value that can be assigned to the variable that does not
violate the constraint.
This will result in a satisfying assignment for $F$.
\end{proof}

The fact that the constraints are all at-most-one-extendible is crucial
to the fact that the first moment bound yields the exact satisfiability threshold
for the 2-core of a random $k$-XOR-SAT formula.  To be more specific:  let $X$ denote
the number of satisfying assignments.  A straightforward calculation shows that if the
number of clauses is $cn$ then $\ex(X)=2^{f(c)n+o(n)}$, where $f(c)$ is a continuous monotonically
decreasing function; in fact, this is true for
virtually all models of random CSP's where the domain size is constant (of course, $f(c)$ depends on the model).
We define $c^+$ to be the solution of $f(c)=0$ and note
that if $c>c^+$ then $\ex(X)=2^{-\Theta(n)}$ and so the problem is a.a.s.\ not satisfiable.
For most problems, $c^+$ is strictly greater than the actual satisfiability threshold.
Usually this can be seen by observing a ``jackpot phenomena'' where the presence of
at least one satisfying assignment a.a.s.\ implies the presence of an exponential number of them.

This is most easily seen in the random instance before it is stripped to the 2-core. Standard arguments show
that such an instance will a.a.s.\ have at least $\e(c) n$ variables of degree zero, for some $\e(c)>0$.
Clearly, any satisfying
assignment remains satisfying after changing the values of any of those variables.  So a.a.s.\ $X>0$ implies
$X\geq 2^{\e(c) n}$.  By continuity, we can choose some $c<c^+$ such that $\e(c)>2f(c)>0$.  Markov's
Inequality implies that $\pr(X\geq 2^{\e(c) n})<\ex(X)/2^{\e(c) n}=2^{f(c)n+o(n)-\e(c) n}<2^{-f(c) n}$.
Therefore a.a.s.\ $X< 2^{\e(c) n}$, and so a.a.s.\ $X=0$; i.e. the formula is a.a.s.\ unsatisfiable.
This implies that the satisfiability threshold (if it exists) is at most $c<c^+$.

Stripping to the 1-core will eliminate this particular jackpot, but a similar jackpot will remain.
There will a.a.s.\ be a linear number of clauses with two variables that don't appear in any other clauses.
For any satisfying assignment, there will always be a way to jointly change the values of those
two variables so that the assignment is still satisfiable.  So once again we a.a.s.\ have $X>0$ implies
$X\geq 2^{\Theta(n)}$.

But if we strip to the 2-core then these, and all similar jackpots, are eliminated. To see this,
consider any satisfying assignment. If we change the assignment of
any one variable, then the constraint on every clause it is in will become violated because they are all
at-most-one-extendible.
So we must change the values of at least one other variable in each of those clauses.  Because every variable
lies in at least 2 clauses (as we are in the 2-core) each such change will cause the constraint on another
clause to be violated.
Expansion properties of random formulas ensure that this pattern spreads until many variables must
be changed - far too many to create a jackpot of the type described above.  Of course, this doesn't imply
that $c^+$ will  be the actual satisfiability threshold, but it shows that a common simple argument
fails to prove that it is not.

Observing that unique satisfiability played a crucial role in this proof, we are inspired to
define a more general class of problems which may also be amenable to similar analysis.

\begin{definition}[UE-CSP]
A UE-CSP instance is a constraint satisfaction problem
where  every constraint is uniquely extendible.
\end{definition}

\begin{definition}[$(k,d)$-UE-CSP]
An instance of $(k,d)$-UE-CSP is an instance of UE-CSP where we restrict every clause
to have size $k$ and every variable to have domain $\{1,...,d\}$.
\end{definition}

A straightforward induction on $k$ yields that there are only
two different uniquely
extendible constraints of size $k$ with $d=2$, and that they are
equivalent to the two different possible parities for a $k$-XOR-SAT constraint.
In other words, $(k,2)$-UE-CSP is exactly $k$-XOR-SAT, and thus is in P.

In addition, both $(k,3)$-UE-CSP and $(2,d)$-UE-CSP are in P~\cite{con08}.
On the other hand, Theorem \ref{thc} states that $(3,d)$-UE-CSP is NP-complete for any $d \geq 4$~\cite{con08}.  In this paper, we focus on the simplest NP-complete case: $(3,4)$-UE-CSP.

\section{The Random Model}
\label{sec:random}
For each appropriate $n,m$, we define $\Omega^{(k,d)}_{n,m}$ to
be the set of $(k,d)$-UE-CSP instances with $m$ clauses
on variables $\{v_1,...,v_n\}$ and a uniquely extendible
constraint on each clause.  We define $U^{(k,d)}_{n,m}$
to be a uniformly random member of $\Omega^{(k,d)}_{n,m}$.
When $m$ is defined to be some function $g(n)$,
we often write $U^{(k,d)}_{n,m=g(n)}$.
As is common in the study of random problems of this sort,
we will be most interested in the case where $m=c n$ for some
constant $c$.
This model is equivalent to first choosing a uniformly
random hypergraph on $n$ vertices and $m$ hyperedges to be the
underlying hypergraph of the $(k,d)$-UE-CSP instance, and then
for each hyperedge, arbitrarily ordering the vertices of the hyperedge and choosing
a uniformly random uniquely extendible
constraint of size $k$ and domain size $d$.

Alternatively, we can consider a second random model.  Define $U^{(k,d)}_{n,p}$ to
be an instance of $(k,d)$-UE-CSP on $n$ variables where each of the
$\binom{n}{k}$ clauses occurs in $U^{(k,d)}_{n,p}$ with probability
$p$ and a uniformly random constraint is applied to each clause.

From results of \cite{bol79,luc90} on random structures,
the two models are asymptotically equivalent in the sense
that
\[ \lim_{n \to \infty} \pr \left(U^{(k,d)}_{n,m} \text{ has property }
   \mathcal{A} \right) =
   \lim_{n \to \infty} \pr \left(U^{(k,d)}_{n,p} \text{ has property }
   \mathcal{A} \right) \]
if $\mathcal{A}$ is a monotone (increasing or decreasing) property
and $m$ and $\binom{n}{k} p$ are ``close'' to each other.
Formally, $m$ is ``close'' to $\binom{n}{k} p$ if
\[ m = \binom{n}{k}p + \bigo\left(\sqrt{\binom{n}{k}p(1-p)}\right), \]
and $p$ is ``close'' to $\frac{m}{\binom{n}{k}}$ if
\[ p = \frac{m}{\binom{n}{k}} + \bigo\left(\sqrt{\frac{m\left(\binom{n}{k}-m\right)}{\binom{n}{k}^3}}\right) . \]

So in particular, Theorem \ref{mt} implies that $U^{(3,4)}_{n,p=d/n^2}$ has a sharp
threshold of satisfiability at $d^*=6c^*$.

We need another random model with which to analyze the 2-core of $U^{(3,4)}_{n,m}$.
Let $\Psi_{n,m}$ denote the subset of $\Omega^{(3,4)}_{n,m}$
in which every variable lies in at least 2 clauses,
and let $U^*_{n,m}$ denote a uniformly random member of $\Psi_{n,m}$.  The following lemma
allows us to work in the  $U^*_{n,m}$ model.

\begin{lemma}\label{f2}
For any $n,m,n',m'$, if we condition on the event that
the 2-core of $U^{(3,4)}_{n,m=cn}$ has $n'$ variables and $m'$
clauses, then that 2-core is a uniformly random member
of $\Psi_{n',m'}$.
\end{lemma}

\begin{proof}
This is a straightforward
variation of the proof of
Claim 1 in the proof of Lemma 4(b) from \cite{mol04},
which is itself a very standard argument.
Consider a formula $F$ and its 2-core $F_c$.  Assume $F_c$ has $n'$
variables and $m'$ clauses.  Replace $F_c$ in $F$ by a arbitrary member of
$\Psi_{n',m'}$, $F'_c$.  Call this new formula $F'$.  Note that $F'_c$
is the 2-core of $F'$ and that $F'$ and $F$ have the same number of clauses.
Thus, the probability a random formula is equal to $F$ is the same as
the probability it is equal to $F'$, and this implies the probability that the
2-core of a random formula is $F_c$ is equal to the probability that the
2-core is $F'_c$.
\end{proof}

\section{Resolution Complexity}
\label{sec:resolution}

While the uncertainty of the P versus NP question means that we cannot
state definitively whether an efficient algorithm exists to either
solve an instance of $(k,d)$-UE-CSP or prove no solution to the instance
exists, we can state that there is no efficient resolution based algorithm,
such as DPLL, that will correctly handle unsatisfiable instances of
$(k,d)$-UE-CSP for $k \geq 3$.

The vast majority of SAT- and CSP-solving algorithms that are commonly used in practice
will recognize an unsatisfiable problem by implicitly producing a resolution proof
that it is unsatisfiable (see \cite{mit02} for the definition of resolution proof).
The {\em resolution complexity} of a boolean formula is the length of the shortest
resolution proof that it is unsatisfiable. (If the formula is satisfiable then
the resolution complexity is $\infty$.)  Thus, the resolution complexity of an
unsatisfiable formula is a lower bound on the time
that it will take an algorithm to solve the problem. In a seminal paper,
Chv\'atal and Szemer\'edi\cite{chv-sze88} proved that a random instance of 3-SAT with a linear
number of clauses a.a.s.\ has exponentially high resolution complexity. This
partially explains the aforementioned observation of Selman et al\cite{sel-mit-lev96} that SAT-solving
algorithms take a very long time on random instances whose density is near the satisfiability threshold.

Mitchell\cite{mit03} discusses two natural ways to extend the notion of resolution complexity to
the setting of a CSP. These two measures of resolution complexity are denoted
$\cres$ and $\ngres$.  The latter appears on the surface to be the most natural
extension in that it extends resolution rules to the setting of a CSP and then
carries them out.  $\cres$, on the other hand, converts a CSP to a boolean CNF-formula
and then carries out CNF-resolution on that formula.  Mitchell shows that for
every CSP instance $\I$, $\cres(\I)\leq\poly(\ngres(\I))$ whereas there are many choices
for $\I$ for which the converse it not true.  Furthermore, all commonly used resolution-type
CSP algorithms correspond nicely to the $\cres$ complexity of the input, but there
are some that do not correspond to the $\ngres$.  For that reason,  we  focus in this paper
on the $\cres$ complexity, as did Mitchell in \cite{mit02}. But note that the above inequality
implies that Theorem \ref{t2}, which we restate below, also holds for the $\ngres$ complexity.

\vspace{.3ex}

{\bf Theorem \ref{t2}} {\em
For any constant $c>0$, and any $k\geq3$, $d\geq2$, the $\cres$
resolution complexity of a uniformly random instance of $(k,d)$-UE-CSP
with $n$ variables and $cn$ clauses is a.a.s.\ $2^{\Theta(n)}$.
}

\vspace{.3ex}

\begin{proof}
From techniques developed in \cite{bs-wig02,mit02,mol-sal03,mit02},
a.a.s.\ the shortest $\cres$ resolution proof of unsatisfiability of a constraint
satisfaction problem on $n$ variables has exponential size
if there exists constants $\a,\z > 0$ such that a.a.s.\ the following
three conditions hold.
\begin{enumerate}
\item Every subformula on at most $\a n$ variables is satisfiable.
\item Every subproblem on $v$ variables, where $\hf \a n \leq v \leq \a n$,
has at least $\z n$ variables of degree at most 1, where the degree of a
variable is the number of clauses containing that variable.
\item If $x$ is a variable of degree at most 1 in a CSP $F$ then,
letting $F'$ be the subproblem obtained by removing $x$ and its clause,
every satisfying assignment of $F'$ can be extended to a satisfying assignment
of $F$ by assigning some value to $x$.
\end{enumerate}
Because our random model for UE-CSP applies one uniquely extendible constraint
to each clause,
the third condition is trivially true.
The following lemma from \cite{mol-sal03} states a useful property
of random formulae with a linear number of clauses:
a.a.s.\ every subproblem on at most $\a n$ variables
has a low clause density.
A similar lemma is proven in \cite{mit02} and several other papers.

\begin{lemma}[\cite{mol-sal03}]
Let $c > 0$ and $k \geq 2$ and let $H$ be a random $k$-uniform hypergraph
with $n$ vertices and $m=cn$ edges.  Then for any $\d > 0$, there exists
$\a = \a(c,k,\d) > 0$ such that a.a.s.\ $H$ has no subgraph with
$0 < h \leq \lfloor \a n \rfloor$ vertices and at least
$\left( \frac{1+\d}{k-1}\right) h$ edges.
\end{lemma}

Let $F'$ be a minimally unsatisfiable subformula of $F$ with at most
$h \leq \lfloor \a n \rfloor$ variables.  Then
$F'$ cannot have a variable of degree less than 2 because condition 3
above is trivially true for UE-CSP.
Thus, $F'$ must have at least $\left(\frac{2}{k}\right) h$ edges.
However, if $k > 2$, there exists $\d = \d(k) > 0$ such that
$\frac{2}{k} > \frac{1+\d}{k-1}$.  As a result, a.a.s.\ every subformula on
$h$ variables, and in particular $F'$, will be satisfiable.
By the same argument, let $\z = 1 - \frac{(1+\d)k}{2(k-1)}$, and $F'$
must have at least $\z n$ variables of degree at most 1.
\end{proof}

\section{The Satisfiability Threshold for $(3,4)$-UE-CSP}
\label{sec:3-4-thresh}

In this section we formally define $c^*$ and prove Theorem~\ref{mt}, our main theorem.

The proof will proceed as follows:  first
the variables of degree 0 and 1 are stripped away to produce the
{\em 2-core}.  Then the satisfiability threshold on the 2-core is proven
using the first and second moment methods, and this threshold on the 2-core
yields the threshold for the original problem.
The upper bound for satisfiability
is a simple use of the first moment method, but the lower bound is more
challenging.
Applying the second moment method produces a complicated summation.
Following the example of \cite{dub-man02}, the summation is
approximated
by a multiple integral, and then the Laplace Method is used to approximate
the integral.  In order to apply the Laplace Method, we must determine the
global maximum of a certain function (see Lemma \ref{lem:fmax}).  We do so
with a computer aided proof using interval analysis, which we present in the appendix.

\subsection{Determining $c^*$}
\label{sec:2core}
Since we will be working with the 2-core,
we need to determine how its density relates to that of the original formula.

\begin{lemma}\label{f1}
Let $U^{(3,4)}_{n,m=cn}$ be a uniformly random instance of
$(3,4)$-UE-CSP with $n$ variables and $m=cn$ clauses.
If $c < \min_{x > 0} \frac{x}{3(1-\rme^{-x})^2} = 0.818469\ldots$, a.a.s.\ $U^{(3,4)}_{n,m=cn}$
has no non-empty 2-core.  Otherwise,
a.a.s.\ the 2-core of $U^{(3,4)}_{n,m=cn}$ has $\Theta(n)$ variables
and $\g(c)+\smallo(1)$ times as many clauses as variables with
\[\g(c)=\frac{x(1-\rme^{-x})}{3(1-\rme^{-x}-x\rme^{-x})},\]
where $x$ is the largest solution to
\[x = 3c(1-\rme^{-x})^2.\]
\end{lemma}

\begin{proof} Recall the CORE procedure from Section~\ref{sec:uecsp-intro}.  Note that
it essentially works on the underlying hypergraph of the random instance.
Cores of random uniform hypergraphs are well-studied (see, for example,
\cite{pit-spe-wor96, maj-wor-hav-cze96, mol04, cai-wor06, kim06, rio07, dar-nor}).
From \cite{cai-wor06}, the 2-core of a $k$-uniform
random hypergraph on $n$ vertices and $cn$ edges
has a.a.s.\ $\left(1-\rme^{-x}-x\rme^{-x}\right)n + \smallo(n)$ vertices and
a.a.s.\ $\left(1-\rme^{-x}\right)^k c n + \smallo(n)$ hyperedges where
$x$ is the largest solution to
$x = ck(1-\rme^{-x})^{k-1}$, and if there is no positive solution to $x = ck(1-\rme^{-x})^{k-1}$ then
the 2-core is a.a.s.\ empty.
Setting $k = 3$ and rearranging completes the proof.
\end{proof}

Below, we will
prove that the satisfiability
threshold for the 2-core of random $(3,4)$-UE-CSP is at clause density 1.
Thus we define:

\begin{definition}\label{dc*}
$c^*=.917935...$ is defined to be the unique solution to $\g(c)=1$.
\end{definition}

\subsection{The First and Second Moment Arguments}
Here we prove that $c=1$ is the satisfiability threshold for $U^*_{n,m=cn}$; i.e.
random $(3,4)$-UE-CSP with minimum degree at least 2. Lemmas \ref{f2} and \ref{f1}
then imply Theorem \ref{mt}. We need to prove two sides of the threshold:

\begin{lemma}\label{l4}
For every $c>1$, $U^*_{n,m=cn}$ is a.a.s.\ unsatisfiable.
\end{lemma}

\begin{lemma}\label{l3}
For every $c<1$, $U^*_{n,m=cn}$ is a.a.s.\ satisfiable.
\end{lemma}

Theorem \ref{mt} then follows:

{\bf Proof of Theorem \ref{mt}}  Let $F$ be a random instance of $U^{(3,4)}_{n,m=cn}$.
Expose $n',m'$, the number of variables and constraints in the 2-core of $F$. By Lemma \ref{f1},
a.a.s.\ $m'=(\g(c)+o(1))n'$. By Lemma \ref{f2}, the 2-core of
$F$, after conditioning on the values of $n',m'$, has the same distribution as $U^*_{n',m'}$.
Therefore, the theorem follows from Lemmas \ref{l4},\ref{l3} and the definition of $c^*$.
\proofend

The proof of the first lemma is straightforward:

{\bf Proof of Lemma \ref{l4}}
We apply what is, in this field, a very standard
first moment argument.  Consider a random instance $F$ chosen
from $\Psi_{n,m=cn}$, and let $N$ denote the number of satisfying assignments
of $F$. We will show that $\ex(N)=\smallo(1)$;  this implies
that a.a.s.\ $N=0$; i.e., that a.a.s.\ $F$ is unsatisfiable.

Consider any assignment $\sigma$ of values to the variables of $F$.
Since each constraint is uniquely extendible, for each possible setting of
two variables in a clause, there is exactly one possible value for
the third variable.  Because the random model
considered includes all possible uniquely extendible constraints and
because there are $4$ possible values for the third variable, the
probability that a particular constraint is satisfied by $\sigma$
is $\inv{4}$.  As there are $4^n$ choices for $\sigma$, we have
\[\ex(N)=4^n4^{-m}=4^{(1-c)n}=\smallo(1),\]
since $c>1$.
\proofend

The next proof is much more complicated, and takes up most of the work in this paper.
As mentioned above, we follow the proof of the corresponding theorem in \cite{dub-man02},
applying a second moment argument.
Unfortunately, our larger domain size yields a larger set of constraints
to choose from and much more complicated calculations
than those in \cite{dub-man02}.

{\bf Proof of Lemma \ref{l3}}
As in the proof of Lemma \ref{l4}, consider a random instance $F$ chosen
from $\Psi_{n,m=cn}$, and let $N$ denote the number of satisfying assignments
of $F$.
Again, we have $\ex(N) = 4^{(1-c)n}$.
The main step of this proof is
to compute the second moment of $N$ obtaining:
\begin{equation}\label{lsm1}
\ex(N^2)=\ex(N)^2(1+\smallo(1)).
\end{equation}
Because $N$ is non-negative, a well known application of the Cauchy-Swartz
inequality implies
\[\pr(N>0)\geq\frac{\ex(N)^2}{\ex(N^2)},\]
and so (\ref{lsm1}) implies Lemma \ref{l3}.

Following the technique of \cite{dub-man02}, the proof will compute
$\ex(N^2)$ by putting $\frac{\ex(N^2)}{\ex(N)^2}$ into the
form
\[ \int_{a_1}^{b_1} \!\! \int_{a_2}^{b_2} \!\! \int_{a_3}^{b_3}
g(x_1,x_2,x_3)e^{nh(x_1,x_2,x_3)}dx_1\,dx_2\,dx_3 \]
where $g$ is polynomial in $n$ and $h$ has a unique maximum in the range of
the integrals.
Then we use the Laplace Method to approximate the triple integral.

Let the $4^n$ possible assignments be $\s_1,\ldots,\s_{4^n}$, and
let $N_i$ be the indicator variable that $\s_i$ is a satisfying assignment.
Then $N = N_1 + \cdots + N_{4^n}$ and $N^2 = \sum_{i,j} N_i N_j$.
Since $N_iN_j=1$ if and only if $F$ is
satisfied by both $\s_i$ and $\s_j$, this indicates that
we must focus on counting the number
of instances satisfied by two
assignments to the variables.

Similarly to \cite{dub-man02},
let $\sigma$ and $\tau$ be arbitrary assignments to the variables,
let $\#\mathcal{C}$ be the total number of instances in $\Psi_{n,m}$,
and let $\#\mathcal{C}_{\sigma , \tau}$ be
the total number of instances in $\Psi_{n,m}$
that are satisfied by both $\sigma$ and $\tau$.
Then,
\[ \ex(N^2) = \frac{1}{\#\mathcal{C}} \sum_{\sigma , \tau} \#\mathcal{C}_{\sigma , \tau}. \]

Let $q$ be the number of different uniquely extendible constraints on three ordered variables.
So for each clause, we will have $q$ choices for the constraint that we assign to
that clause.
As is done in \cite{dub-man02},
we can think of the $m$ clauses as inducing a distribution
of $3m$ ``places'' to the $n$ variables such that each variable
receives at least 2 ``places''.  In addition, we add the restriction
that for each triple of ``places'' corresponding to
a clause, no two ``places''
are assigned to the same variable.  We define:
\begin{itemize} \item $S(i,j,2)$, known as a generalized Stirling number of the second kind,
is the number of ways to partition $i$ elements into
$j$ sets such that each set has at least 2 elements.
\item $\Lambda_{3m,n}$ is the probability that, in the above partition,
no two elements from any one triple appear in the same set.
\end{itemize}
So, $\#\mathcal{C} = q^{m}S(3m,n,2)n!\Lambda_{3m,n}$.

Consider a clause and a random constraint on that clause.  We need to determine
the probability that both assignments $\s$ and $\t$ satisfy the constraint.
Place an arbitrary ordering on the variables of the clause, and let
$\a$, $\b$, and $\g$ be the values assigned to those variables by $\s$, and let
$\a'$, $\b'$, and $\g'$ be the values assigned to those variables by $\t$.
In addition, let $z_{\a\b}$ be the unique value that the constraint forces the third
variable to be
if the first variable is assigned $\a$ and the second variable is
assigned $\b$.

If the values of the three variables are unchanged between $\s$ and $\t$,
i.e.\ if $\a=\a'$, $\b=\b'$, and $\g=\g'$, then
the probability that both assignments satisfy a random constraint is
the same as the probability that a random constraint assigns
the third variable $\g$ if the first two are assigned $\a$ and $\b$
respectively.  Every constraint of size 3 will permit a tuple of
the form $(\a, \b, z_{\a\b})$,
and there are $4$ possible choices for $z_{\a \b}$.  Exactly
$\frac{1}{4}$ of the constraints will have $z_{\a \b} = \g$.  As a result,
\[\pr(z_{\a\b} = \g) = \frac{1}{4}.\]
Thus, both assignments will satisfy a proportion of
$\frac{1}{4}$ of the possible constraints.

Note that the uniquely extendible property means that if exactly one
of the clause's variables changes value between $\s$ and $\t$, for example
if $\a=\a'$, $\b = \b'$, and $\g \neq \g'$, then
the constraint cannot be satisfied by both $\s$ and $\t$.

Suppose one variable, assume w.l.o.g.\ the first variable, is assigned the
same value by $\s$ and $\t$ and each of the other two variables is assigned
a different value in $\t$ from what it is assigned in $\s$, i.e.\
$\a=\a'$, $\b \neq \b'$, and $\g \neq \g'$.
In this case, we need to determine
\[\pr(z_{\a\b} = \g \ \wedge \ z_{\a\b'} = \g')
	= \pr(z_{\a\b'} = \g' | z_{\a\b} = \g)\pr(z_{\a\b} = \g).\]
Every constraint of size 3 that permits the tuple $(\a, \b, \g)$ will
also permit a tuple of the form $(\a, \b', z_{\a \b'})$
with $z_{\a\b'} \neq \g$.
There are $3$ choices for $z_{\a\b'}$, and by symmetry, each is equally likely. So exactly
$\frac{1}{3}$ of the constraints will have $z_{\a\b'} = \g'$.
Thus,
\[\pr(z_{\a\b'} = \g' | z_{\a\b} = \g)\pr(z_{\a\b} = \g)
 =  \frac{1}{3} \cdot \frac{1}{4}. \]
As a result, both assignments will satisfy a proportion of
$\frac{1}{12}$ of the possible constraints.

Finally, suppose none of the variables of the clause receives
the same value in $\t$ as it
does in $\s$.
If $\a \neq \a'$, $\b \neq \b'$ and $\g \neq \g'$, then
\begin{eqnarray*}
\pr(z_{\a\b} = \g \ \wedge \ z_{\a'\b'} = \g') & = &
  \pr(z_{\a\b} = \g \ \wedge \ z_{\a\b'} \neq \g' \ \wedge \ z_{\a'\b'} = \g')\\
 & = & \pr(z_{\a'\b'} = \g' \ | \ z_{\a\b} = \g \ \wedge z_{\a\b'} \ \neq \g')\\
 &   & \mbox{} \times \ \pr(z_{\a\b'} \neq \g' | z_{\a\b} = \g)
			\pr(z_{\a\b} = \g)
\end{eqnarray*}
Consider only the tuples $(\a,\b,z_{\a\b})$, $(\a,\b',z_{\a\b'})$,
and $(\a',\b',z_{\a'\b'})$ permitted by the constraint.
There are $4$ choices for $z_{\a\b}$,
and exactly $\frac{1}{4}$ of the constraints will have $z_{\a\b}=\g$.
For each constraint that contains the tuple $(\a,\b,\g)$, there are
exactly $3$ choices for $z_{\a\b'}$, and $2$ of these choices are not
$\g'$.  By a straightforward symmetry argument, each is equally likely.
  As a result, $\frac{2}{3}$ of the constraints that contain
the tuple $(\a,\b,\g)$ will not contain the tuple $(\a,\b',\g')$.
Finally, for every choice of $z_{\a\b}$ and
$z_{\a\b'}$ there are $3$ equally likely
choices for
$z_{\a'\b'}$.  The reason is that
we fix the set of 3-tuples that start with $\a$, and then the number of ways
we can choose the set of tuples that starts with $\a'$ is equal to the number
of derangements on $4$ elements, and by symmetry
each derangement is equally likely.

As a result, the expression above yields:
\[ \pr(z_{\a\b} = \g \ \wedge \ z_{\a'\b'} = \g')
  =  \frac{1}{3} \cdot \frac{2}{3} \cdot \frac{1}{4} . \]
So both assignments will simultaneously
satisfy a proportion of $\frac{1}{18}$ of the
possible constraints.

Using the same notation as \cite{dub-man02},
let $I_k = \{ 0, \frac{1}{k}, \frac{2}{k}, \ldots, \frac{k-1}{k}, 1 \} $,
and let $\alpha \in I_n$ be the
proportion of variables having the same value in both assignments.
To enumerate all pairs of assignments,
we must count
the number of choices for the $\alpha n$ variables and count
the possible assignments to the variables.
This gives $\sum_{\alpha \in I_n}\binom{n}{\a n} 4^n 3^{(1-\alpha)n}$
pairs of assignments.

To enumerate all satisfied instances for one pair of assignments, let
$r \in I_{3m}$ be the proportion of $3m$ ``places'' in the second
assignment that receive one of the $\alpha n$ variables.
Note that if $\a = 0$ then $r = 0$, and if $\a = 1$ then $r = 1$.
Otherwise,
because each variable occurs in at least two places, we must have $r \geq \frac{2\a}{3c}$
and $1-r \geq \frac{2(1-\a)}{3c}$.

Let $T_k$ be the number of clauses with $3-k$ of these $\alpha n$ variables.
Recall that $T_1 = 0$.
For each choice of $T_0, T_2, T_3$, we need to \\
(a) count the ways to choose the clauses for $T_0, T_2, T_3$:
\[\binom{m}{T_0} \binom{m-T_0}{T_2};\]
(b) for each clause, count the number of ways we can choose a constraint for
the clause given the number of variables in the clause that receive the
same value in $\t$ as in $\s$:
\[\left( \frac{q}{18} \right) ^{T_3}
    \left( \frac{q}{12} \right) ^{T_2}
    \left( \frac{q}{4} \right) ^{T_0};\]
(c) for each of the $T_2$ clauses that have exactly one of the $\a n$ variables,
count the 3 positions for the  $\alpha n$ variables: \[3^{T_2};\]
(d) finally, distribute the variables amongst the ``places''. Recall that
$S(i,j,2)$ counts the number of ways to partition $i$ elements into
$j$ sets such that each set has at least 2 elements.  Also, we let $\Lambda_{3m,n}(\a,r,T_0,T_2,T_3)$
denote the probability that this distribution of ``places'' is such that each
clause contains three distinct variables. So the total number of choices for this step is:
  \[S(r3m, \alpha n, 2)(\alpha n)! S((1-r)3m, (1-\alpha)n, 2)((1-\alpha)n)!\Lambda_{3m,n}(\a,r,T_0,T_2,T_3). \]

As a result, for a given $\a$, $r$, $T_0$, $T_2$, and $T_3$, the number of instances of $\Psi_{n,m}$ satisfied by
a pair of assignments that fit the given parameters is
\begin{multline*}
\binom{n}{\alpha n} 4^n
 3^{(1-\alpha)n}\binom{m}{T_0}\binom{m-T_0}{T_2} \left( \frac{q}{18} \right) ^{T_3}
\left( \frac{q}{12} \right) ^{T_2}
\left( \frac{q}{4} \right) ^{T_0} 3^{T_2} \\
\qquad \times
 S(r3m, \alpha n, 2)(\alpha n)!  S((1-r)3m, (1-\alpha)n, 2)((1-\alpha)n)!
            \Lambda_{3m,n}(\a,r,T_0,T_2,T_3).
\end{multline*}

Let $t \in I_m = \{0, \frac{1}{m}, \frac{2}{m}, \ldots, \frac{m-1}{m}, 1\}$
be the proportion of $m$ clauses in which all 3
variables have the same assignment.  Thus,

\noindent
$\begin{array}{lclcl}
T_0 & = & tm & & \\
T_2 & = & 3rm-3T_0 & = & 3rm-3tm \\
T_3 & = & m-T_0-T_2 & = & m-3rm+2tm.
\end{array}$

\noindent

Note that $T_2 \geq 0$ implies $r \geq t$ and $T_3 \geq 0$ implies
$t \geq \frac{3r-1}{2}$. Note also that if $\a =0$ then $r = 0$, $T_0 = T_2 = 0$ and $T_3 = m$.
Likewise, if $\a =1$ then $r = 1$, $T_0 = m$ and $T_2 = T_3 = 0$.

\begin{definition} \begin{enumerate}
\item[(a)] We say that $(\a,r,t)\in I_n\times I_{3m}\times I_m$ is {\em feasible} if it satisfies $0\leq\a\leq 1$,
$\frac{2\a}{3c}\leq r\leq1-\frac{2(1-\a)}{3c}$, $\max(0,\frac{3r-1}{2})\leq t\leq r$ where if $\a=0$ then $r=t=0$,
and if $\a=1$ then $r=t=1$.
\item[(b)] We say that $(\a,r,t)$ is {\em extremal} if it is feasible and either
$\a=0,\a=1,r=\frac{2\a}{3c}, r=1-\frac{2(1-\a)}{3c},t=\max(0,\frac{3r-1}{2})$ or $t=r$.
\end{enumerate}
\end{definition}

 Setting $\Lambda_{3m,n}^{\a,r,t}= \Lambda_{3m,n}(\a,r,T_0,T_2,T_3)$,
substituting and factoring out common terms yields that $\ex(N^2)$ is the sum over all feasible $(\a,r,t)$ of:
\begin{eqnarray*}
F(\a,r,t)&=&
\frac{1}{S(3m,n,2)n!\Lambda_{3m,n}}\times
        \binom{n}{\a n} 4^n 3^{(1-\a)n} \\
& & \times \left(\frac{m!}{(m-3rm+2tm)!(3rm-3tm)!(tm)!} \right)
 	3^{3rm-3tm} \\
& & \times  \left( \frac{1}{4} \right)^{m}
	\left( \frac{1}{3} \right)^{2m+tm-3rm}
	2^{m-3rm+2tm} \\
 & & \times S(r3m,\a n, 2)(\a n)! S((1-r)3m, (1-\a)n, 2)((1-\a)n)!\Lambda_{3m,n}^{\a,r,t}.
\end{eqnarray*}

We will treat the cases where $\a,r,t$ are close to the endpoints of their ranges separately.  So we will take a constant $\z>0$
(to be specified in Section \ref{sborder}) and we define:
\begin{alignat*}{1}
& I_n^+ = I_n \cap [\z,1-\z], \qquad I_{3m}^+ =  I_{3m} \cap \left[\frac{2\a+\z}{3c} , 1 - \frac{2(1-\a)+\z}{3c}\right], \\
& I_m^+ =  I_m \cap \left[\max\left(0,\frac{3r-1}{2}\right)+\z, r-\z\right].
\end{alignat*}
In Section~\ref{sborder} we will prove:

\begin{lemma}\label{lborder}
Let $c > \frac{2}{3}(1 + 2\z)$. 
The sum of $F(\a,r,t)$ over all feasible $\a,r,t$ not satisfying $\a\in I_n^+, r\in I_{3m}^+,t\in I_m^+$, is $o(4^{2n-2m})$.
\end{lemma}
In Section~\ref{sborder} we choose a $\z$ small enough that $\frac{2}{3}(1+2\z) < 0.818469\ldots$, the threshold for a non-empty
2-core. Thus, we can assume $c>\frac{2}{3}(1 + 2\z)$ and so Lemma \ref{lborder} yields

\[\ex(N^2) =  o(\ex(N)^2) + \sum_{\alpha \in I_n^+} \sum_{r \in I_{3m}^+} \sum_{t \in I_m^+} F(\a,r,t).\]

\noindent{\bf Notation:}  We use $f\sim g$ to mean that $f=g+o(g)$.

In the next few steps, we will use similar techniques as \cite{dub-man02} to approximate $\ex(N^2)$.
First, we will apply approximations for the $S$ terms of $F$.  Given Lemma~\ref{lborder}, we can restrict the analysis to
$c > \frac{2}{3}(1+2\z)$ and
points satisfying $\a\in I_n^+, r\in I_{3m}^+,t\in I_m^+$. In this domain, each of the $S$ terms is of the form
$S(i, j, 2)$ where {$j>\z n>\frac{\z}{3c}i$ and $j<\frac{i}{2}-\frac{\z}{2}n<\frac{i}{2}(1-\frac{\z}{2})$.}  To approximate the $S$ terms,
we use Lemma~\ref{lemma:Stirling} which will be presented in
Section~\ref{sec:Stirling}.  This lemma appears in \cite{dub-man02} but with
a typographical error.
The lemma is also a specific case of the general
results in \cite{hen94}.
It gives the following approximation that holds uniformly as $i \rightarrow \infty$ for all 
$\d_1 i < j \leq \frac{i}{2}(1-\d_2)$ for any positive constants $\d_1$ and $\d_2$: 
\begin{equation}\label{esij}
S(i,j,2) \sim \frac{1}{j!} \left( \frac{i}{z_0\rme} \right) ^{i}
(\rme^{z_0}-1-z_0)^j \Phi(i,j)
\end{equation}
where $z_0$ is the positive real solution of the equation
\begin{equation} \frac{j}{i}z_0 = \frac{\rme^{z_0}-1-z_0}{\rme^{z_0}-1}\label{eq:z0} \end{equation}
and where
\[\Phi(i,j) = \sqrt{\frac{ij}{z_0j(i-j)-i(i-2j)}} .\]

If we rewrite (\ref{eq:z0}) as $\frac{i}{j} = \frac{z_0(\rme^{z_0}-1)}{\rme^{z_0}-1-z_0}$, note that the right hand side is a monotonically
increasing function of $z_0$. Also, $\lim_{z_0 \to 0} \frac{z_0(\rme^{z_0}-1)}{\rme^{z_0}-1-z_0} = 2$ and
$\lim_{z_0 \to \infty} \frac{z_0(\rme^{z_0}-1)}{\rme^{z_0}-1-z_0} = \infty$.  Given these facts, the following
observation is straightforward.

\begin{observation}
\label{obs:z0}
The function $z_0 = z_0(i,j)$, defined by (\ref{eq:z0}), is a continuous and differentiable function of $i$ and $j$ in the domain $i > 2j > 0$.  
In addition, the limit of $z_0(i, j)$ as $(i, j)$ approaches $(2j', j')$, for any $j' > 0$, exists and is 0.
\end{observation}

For $\a\in I_n^+, r\in I_{3m}^+,t\in I_m^+$, we have $j,i-2j>\z n$ in each of the following three approximations,
and so they hold uniformly as $n\rightarrow \infty$ for all such $\a,r,t$.
\[ \begin{array}{lcl}
\multicolumn{3}{l}{S(3m,n,2)n! \sim (\rme^x-1-x)^nx^{-3m}\rme^{-3m}(3m)^{3m}\Phi(3m,n)} \\
\multicolumn{3}{l}{S(r3m,\alpha n,2)(\alpha n)! \sim (\rme^z-1-z)^{\alpha n} z^{-r3m}\rme^{-r3m}(r3m)^{r3m}\Phi(r3m,\alpha n)} \\
S((1-r)3m,(1-\alpha)n,2)((1-\alpha)n)! & \sim & (\rme^y-1-y)^{(1-\alpha)n} y^{-(1-r)3m} \\
 & & \rme^{-(1-r)3m}((1-r)3m)^{(1-r)3m} \\
 & & \Phi((1-r)3m,(1-\alpha)n)
\end{array} \]
for some $x,y,z > 0$ such that
\begin{equation}
\frac{\rme^x-1-x}{\rme^x-1}-\frac{x}{3c} =
\frac{\rme^y-1-y}{\rme^y-1}-\frac{y(1-\alpha)}{3c(1-r)} =
\frac{\rme^z-1-z}{\rme^z-1}-\frac{z\alpha}{3cr} = 0.
\label{eq:implicits}
\end{equation}

To complete our approximation of $F(\a,r,t)$, we apply Stirling's formula, $i!\sim i^i\rme^{-i}\sqrt{2\pi i}$, to $m!,(m-3rm+2tm)!,(3rm-3tm)!,(tm)!$.  For $t\in I_{m}^+$, we have $m-3rm+2tm, 3rm-3tm,tm\geq \z m$ and so this
is an asymptotically tight approximation as $n\rightarrow\infty$. We then group the non-exponential and exponential terms and simplify.  Lemma \ref{lborder} along with the uniformity of our approximations for the
$S(i,j,2)$ terms and of Stirling's approximation, allow us to write $\ex(N^2)$ in the form:
\begin{equation}\label{eap1}
\ex(N^2) \sim o(E(N)^2) +
 \sum_{\alpha \in I_n^+}
\sum_{r \in I_{3m}^+}
\sum_{t \in I_m^+}
g(\alpha , r,t) \rme^{nf(\alpha ,r,t)}
\end{equation}
where
\begin{eqnarray*}
g(\alpha,r,t) & = & (\Phi(3m,n)\Lambda_{3m,n})^{-1} \Phi(r3m,\alpha n) \Phi((1-r)3m,(1-\alpha) n)
	\Lambda_{3m,n}^{\a,r,t} \\
 & & \mbox{} \times \frac{1}{\sqrt{(2 \pi)^3 nm^2 \a (1-\a)(1-3r+2t)(3r-3t)t }} \\
f(\alpha,r,t) & = & \frac{1}{n} \left[ n \ln 4 + (1-\alpha)n \ln 3 -m \ln 4
 -m(2+t-3r)\ln 3 \right. \\
 & & \mbox{} + m(1-3r+2t)\ln 2
     + m(3r-3t)\ln 3 -n\alpha \ln \alpha \\
 & & \mbox{} -n(1-\alpha) \ln (1-\alpha) - m(1-3r+2t) \ln (1-3r+2t) \\
 & & \mbox{} - m(3r-3t) \ln (3r -3t) -mt \ln t + \alpha n \ln (\rme^z-1-z) \\
 & & \mbox{} -r3m \ln z -r3m + r3m \ln (r3m) + (1 - \alpha) n \ln(\rme^y-1-y) \\
 & & \mbox{} -(1-r)3m \ln y - (1-r)3m + (1-r)3m \ln ((1-r)3m) \\
 & & \mbox{} \left.  -n \ln(\rme^x-1-x) +3m \ln x + 3m - 3m \ln (3m) \right].
\end{eqnarray*}

By considering each term of $f$ separately, and from Observation~\ref{obs:z0}, we have:
\begin{observation}\label{obs:f}
For all $c > \frac{2}{3}$,
$f(\a, r, t)$ is continuous and differentiable in the domain $0 < \a < 1$, $\frac{2\a}{3c} < r < 1 - \frac{2(1-\a)}{3c}$, and
$\max\left\{0, \frac{3r-1}{2}\right\} < t < r$.  Furthermore,
for every point $(\a_1, r_1, t_1)$ on the boundary of this domain
$\lim_{(\a,r,t) \to (\a_1, r_1, t_1)} f(\a, r, t)$ exists.
\end{observation}

We fix $c$ and replace $m$ with $cn$.
Combining common terms, and dividing $f$
through by $n$ gives
\begin{eqnarray}
g(\alpha,r,t) & = & (\Phi(3cn,n)\Lambda_{3m,n})^{-1} \Phi(r3cn,\alpha n) \Phi((1-r)3c n,(1-\alpha) n)
		\Lambda_{3m,n}^{\a,r,t} \notag \\
 & & \mbox{} \times \frac{1}{(2 \pi n)^{3/2} c\sqrt{\a (1-\a)(1-3r+2t)(3r-3t)t }} \label{eq:g} \\
f(\alpha , r,t) & = & \ln 4 -c \ln 4 + (1-\a)\ln 3 - c(2+t-3r)\ln 3 \notag \\
 & & \mbox{} +c(1-3r+2t)\ln 2 -\alpha \ln \alpha -(1-\alpha)\ln (1-\alpha) \notag \\
 & & \mbox{} -c(1-3r+2t) \ln (1-3r+2t) -c(3r-3t) \ln (r-t) -ct \ln t \notag \\
 & & \mbox{} + r3c \ln r + (1-r)3c \ln (1-r) + \alpha \ln (\rme^z-1-z)-r3c \ln z \notag \\
 & & \mbox{} + (1-\alpha)\ln (\rme^y-1-y) - (1-r)3c \ln y - \ln (\rme^x-1-x) \notag \\
 & & \mbox{} +3c \ln x, \label{eq:f}
\end{eqnarray}
and thus,
\[ \frac{E(N^2)}{E(N)^2} \sim \sum_{\alpha \in I_n^+}
\sum_{r \in I_{3cn}^+}
\sum_{t \in I_{cn}^+}
g(\alpha , r,t) e^{n(f(\alpha ,r,t) - 2(1-c)\ln 4)} . \]

To analyze this sum, we will need a few technical lemmas about $f$ and $g$.
We begin with the key - determining a global maximum for $f$:
\begin{lemma} \label{lem:fmax}
For all $\frac{2}{3} \leq c < 1$, the unique global maximum of
$f(\alpha, r, t)$ in the domain $0 < \alpha < 1$, $\frac{2\a}{3c} < r < 1-\frac{2(1-\a)}{3c}$,
and $\max \{0, \frac{3r-1}{2}\} < t < r$ occurs
at $f(\frac{1}{4}, \frac{1}{4}, \frac{1}{16})=2(1-c)\ln 4$.  Furthermore, 
for all $.67 \leq c \leq 1$, there exists a constant $b = b(c)$ such that
 at every extremal $(\a,r,t)$,
$f$ approaches a limit that is less than $2(1-c)\ln 4 - b$.
\end{lemma}
The proof of this lemma is computer-aided.  We provide an outline in Section~\ref{sd}, and then more
details in the appendix.

The analysis of $g$ is not as delicate.  The proofs of the following lemmas appear in Section \ref{sd}.
First, we show that $g$ is bounded:
\begin{lemma}\label{lgbound}
There exists a constant $\n$, independent of $n$,
such that for all $\alpha \in I_n^+$, $ r \in I_{3m}^+$,
and $t \in I_m^+$
we have $g(\a,r,t) \leq \n$. 
\end{lemma}

Next we provide an approximation of $g(\a,r,t)$ when $\a,r,t$ are near the critical point from Lemma \ref{lem:fmax}.

\begin{lemma}\label{lem:distribution} If $\a=\frac{1}{4}+o(1),r=\frac{1}{4}+o(1),t=\frac{1}{16}+o(1)$ then
\[g(\a,r,t) = \Phi(3cn,n)
\times (\pi n)^{-\frac{3}{2}}\frac{2^6}{9c}(1+o(1)).\]
\end{lemma}

By these lemmas, it is straightforward (see Section \ref{sd}) to show that the mass of the sum will
occur when $\a,r,t$ are near $\inv{4},\inv{4},\inv{16}$ and so we can replace $g(\a,r,t)$ by its approximation
from Lemma \ref{lem:distribution}, obtaining:
\begin{equation}\label{esum}
\frac{E(N^2)}{E(N)^2} \sim \sum_{\alpha \in I_n^+} \sum_{r \in I_{3cn}^+} \sum_{t \in I_{cn}^+}
\Phi(3cn,n) \times (\pi n)^{-\frac{3}{2}}\frac{2^6}{9c}
 \rme^{n(f(\alpha ,r,t) - 2(1-c)\ln 4)} .
\end{equation}
Next, we replace the summation with an integral.  The summation is
essentially a Riemann sum, and as $n$ tends to infinity, the error term
from approximating the summation with an integral tends to 0.
\begin{eqnarray*}
\frac{\ex(N^2)}{\ex(N)^2} & \sim &
\Phi(3cn,n) \times (\pi n)^{-\frac{3}{2}}\frac{2^6}{9c}
\sum_{\alpha \in I_n^+}
\sum_{r \in I_{3cn}^+}
\sum_{t \in I_{cn}^+}
 \rme^{n(f(\alpha ,r,t) - 2(1-c)\ln 4)} \\
 & \sim &
\Phi(3cn,n) \times (\pi n)^{-\frac{3}{2}}\frac{2^6}{9c} \times
n \times 3cn \times cn \\
 & & \times \int_{\z}^{1-\z} \!\! \int_{\frac{2\a+\z}{3c}}^{1-\frac{2(1-\a)+\z}{3c}} \!\! \int_{\max \{ 0, \frac{3r-1}{2} \}+\z }^{r-\z}
\rme^{n(f(\alpha ,r,t) - 2(1-c)\ln 4)} dt\, dr\, d\alpha \\
 & < & \Phi(3cn,n) \times n^{3/2}\frac{2^6c}{3\pi^{3/2}} \\
 & &\times
\int_0^1 \!\! \int_{\frac{2\a}{3c}}^{1-\frac{2(1-\a)}{3c}} \!\! \int_{\max \{ 0, \frac{3r-1}{2} \} }^{r}
\rme^{n(f(\alpha ,r,t) - 2(1-c)\ln 4)} dt\, dr\, d\alpha .
\end{eqnarray*}

Continuing with the technique of \cite{dub-man02}, we approximate this sum by using the Laplace
Method, as given in eg.\
 \cite{ben-ors78} and \cite{bru81}.
The Laplace Method for a triple integral can be stated as follows.

\begin{lemma} [\cite{bru81}]\label{lb}
Let
\[ F(n) = \int_{a_1}^{b_1} \!\! \int_{a_2}^{b_2} \!\! \int_{a_3}^{b_3}
\rme^{nh(x_1,x_2,x_3)}dx_1\,dx_2\,dx_3 \]
where
\begin{itemize}
\item[(a)] $h$ is continuous in $a_i \leq x_i \leq b_i$,
\item[(b)] $h(c_1,c_2,c_3)=0$ for some point $(c_1,c_2,c_3)$
with $a_i < c_i < b_i$ \\
and $h(x_1,x_2,x_3) < 0$ for all other points in
the range,
\item[(c)]
$\displaystyle{h(x_1,x_2,x_3) =
    -\frac{1}{2}\sum_{i=1}^{3}\sum_{j=1}^{3}a_{ij}x_ix_j
    + \smallo(x_1^2+x_2^2+x_3^2)}$ \\
with $(x_1^2+x_2^2+x_3^2 \rightarrow 0)$, and
\item[(d)] the quadratic form $\sum\sum a_{ij}x_ix_j$ is positive definite.
\end{itemize}
Then,
\[ F(n) \sim (2 \pi)^{\frac{3}{2}}D^{-\frac{1}{2}}n^{-\frac{3}{2}} \]
where $D$ is the determinant of the matrix $(a_{ij})$.
\end{lemma}

Just as is done in \cite{dub-man02},
we apply Lemma \ref{lb} by letting $h(x_1,x_2,x_3) = f(\alpha,r,t)-2(1-c)\ln 4$. By Observation~\ref{obs:f}, this satisfies point (a).
Point (b) is satisfied by Lemma \ref{lem:fmax}. Point (c) is
satisfied if we approximate $h$ by the Taylor expansion about the
point $\alpha = \frac{1}{4}$, $r=\frac{1}{4}$, $t=\frac{1}{16}$ and
take the $a_{ij}$'s from the second partial derivatives of $h$.
Part (d) is satisfied by the following lemma which
is proved in Section \ref{sd}:

\begin{lemma}\label{lpd}
The quadratic form $\sum\sum a_{ij}x_ix_j$ is positive definite, and the determinant of
the matrix $(a_{ij})$ is $D=\frac{2^{15}}{9}c^2(\Phi(3cn,n))^2$.
\end{lemma}

Now, we can apply the Laplace Method to obtain,
\begin{eqnarray*}
\frac{\ex(N^2)}{\ex(N)^2} & \sim & \Phi(3cn,n) \times n^{\frac{3}{2}}\frac{2^6c}{3\pi^{3/2}}
\times(2 \pi)^{\frac{3}{2}}
    \left(\frac{2^{15}}{9} c^2(\Phi(3cn,n))^2\right)^{-\frac{1}{2}}
    n^{-\frac{3}{2}}\\
  &\sim & 1
\end{eqnarray*}
proving (\ref{lsm1}) and thus completing the proof of Lemma \ref{l3}.
\proofend

\section{Further details}\label{sd}
Here we present most of the details that were postponed from the analysis in the previous section.
We'll start with the straightforward derivation of (\ref{esum}).  By Lemma \ref{lem:fmax},
$f(\frac{1}{4}, \frac{1}{4}, \frac{1}{16})$ is a global maximum.  Therefore, there exists a function
$\rho(n)=o(1)$ such that, defining $\Psi=\{(\a,r,t):|\a-\inv{4}|,|r-\inv{4}|,|t-\inv{16}|>\rho(n)\}$,
we have $f(\a,r,t)<f(\frac{1}{4}, \frac{1}{4}, \frac{1}{16})-10\log n/n$ for all  $(\a,r,t)\in \Psi$.
This and Lemma \ref{lgbound} imply that there exists a constant $\n$ independent of $n$
such that
\begin{eqnarray*}
\lefteqn{\sum_{\alpha,r,t \in \Psi} g(\alpha , r,t) e^{n(f(\alpha ,r,t) - 2(1-c)\ln 4)}} \\
&<&\sum_{\alpha,r,t \in \Psi} \n \, \rme^{n(f(\inv{4},\inv{4},\inv{16}) - 2(1-c)\ln 4-10\log n/n)}\\
&\leq& |I_n^+||I_{3cn}^+||I_{cn}^+|\times \n \times n^{-10}\rme^{n(f(\inv{4},\inv{4},\inv{16}) - 2(1-c))}\\
&=&o(e^{n(f(\inv{4},\inv{4},\inv{16}) - 2(1-c))}).
\end{eqnarray*}
Thus, if for $\alpha,r,t \in \Psi$, we replace $g(\alpha , r,t)$ by any constant, specifically  $g(\inv{4},\inv{4},\inv{16})$, then it will have a negligible effect on the sum.  For $\alpha,r,t \notin \Psi$,
Lemma \ref{lem:distribution} implies that we can also replace $g(\alpha , r,t)$ by
$g(\inv{4},\inv{4},\inv{16})\sim \Phi(3cn,n)\times (\pi n)^{-\frac{3}{2}}\frac{2^6}{9c}$
and so obtain
\[\frac{\ex(N^2)}{\ex(N)^2} \sim \sum_{\alpha \in I_n^+} \sum_{r \in I_{3cn}^+} \sum_{t \in I_{cn}^+}
\Phi(3cn,n) \times (\pi n)^{-\frac{3}{2}}\frac{2^6}{9c}
 \rme^{n(f(\alpha ,r,t) - 2(1-c)\ln 4)} ,\]
 which is (\ref{esum}).
 \proofend

\subsection{The global maximum}
We start with our proof that $f(\frac{1}{4}, \frac{1}{4}, \frac{1}{16})$ is a global maximum.  Recall our statement:

\begin{prevlemma}{\ref{lem:fmax}}
For all $\frac{2}{3} \leq c < 1$, the unique global maximum of
$f(\alpha, r, t)$ in the domain $0 < \alpha < 1$, $\frac{2\a}{3c} < r < 1-\frac{2(1-\a)}{3c}$,
and $\max \{0, \frac{3r-1}{2}\} < t < r$ occurs
at $f(\frac{1}{4}, \frac{1}{4}, \frac{1}{16})=2(1-c)\ln 4$.  Furthermore, 
for all $.67 \leq c \leq 1$, there exists a constant $b = b(c)$ such that
 at every extremal $(\a,r,t)$,
$f$ approaches a limit that is less than $2(1-c)\ln 4 - b$.
\end{prevlemma}

Most of the computer-aided proof of this lemma is deferred to the appendix.
We lay down some of the initial steps here.

If we differentiate $f$ with respect to $\alpha, r, t$, we get
\begin{alignat}{1}
\frac{\partial f}{\partial \alpha} &= D_\alpha [(1-\alpha) \ln 3
   - \alpha \ln \alpha - (1- \alpha ) \ln (1-\alpha) + \alpha \ln (\rme^z-1-z)
   \notag \\
 & \qquad -r3c \ln z + (1-\alpha) \ln (\rme^y-1-y) - (1-r)3c \ln y ]
   \notag \\
 &= - \ln 3 - \ln \alpha + \ln (1-\alpha) + \ln (\rme^z-1-z) -
   \ln (\rme^y-1-y) \notag \\
 & \qquad + \left[\frac{\alpha(\rme^z-1)}{\rme^z-1-z} -
     \frac{3rc}{z}\right]\frac{\partial z}{\partial \alpha}
    + \left[\frac{(1-\alpha)(\rme^y-1)}{\rme^y-1-y}
    - \frac{(1-r)3c}{y}\right]\frac{\partial y}{\partial \alpha} \notag \\
 &= - \ln 3 - \ln \alpha + \ln (1-\alpha) + \ln (\rme^z-1-z) -
   \ln (\rme^y-1-y) \tag*{by (\ref{eq:implicits})}
\end{alignat}
\begin{alignat}{1}
\frac{\partial f}{\partial r} &= D_r [ 3rc \ln 3 - 3rc \ln 2
-c(1-3r+2t)\ln(1-3r+2t) \notag \\ & \qquad - c(3r-3t)\ln (r-t)
 + r3c \ln r + (1-r)3c \ln (1-r) + \alpha \ln (\rme^z-1-z) \notag \\
  & \qquad - r3c \ln z
 + (1-\alpha) \ln (\rme^y-1-y) - (1-r)3c \ln y ] \notag \\
&= 3c \ln 3 - 3c \ln 2 + 3c\ln(1-3r+2t) - 3c \ln (r-t) + 3c \ln r \notag \\
 & \qquad - 3c \ln (1-r) - 3c \ln z + 3c \ln y
 + \left[\frac{\alpha (\rme^z-1)}{\rme^z-1-z} - \frac{r3c}{z}\right]
    \frac{\partial z}{\partial r} \notag \\
& \qquad + \left[\frac{(1-\alpha)(\rme^y-1)}{\rme^y-1-y} - \frac{(1-r)3c}{y}\right]
    \frac{\partial y}{\partial r} \notag \\
&= 3c \ln 3 - 3c \ln 2 + 3c\ln(1-3r+2t) - 3c \ln (r-t) + 3c \ln r \notag \\
 & \qquad - 3c \ln (1-r) - 3c \ln z + 3c \ln y \tag*{by (\ref{eq:implicits})}
\end{alignat}
\begin{eqnarray*}
\frac{\partial f}{\partial t} & = & D_t [ -ct \ln 3 + c2t \ln 2
-c(1-3r+2t)\ln(1-3r+2t) \\
 & & \mbox{} - c(3r - 3t) \ln (r-t) - ct \ln t ] \\
 & = & -c\ln 3 + 2c \ln 2 - 2c\ln(1-3r+2t) + 3c \ln (r-t) - c \ln t  .
\end{eqnarray*}

Setting
$\frac{\partial f}{\partial \alpha} = \frac{\partial f}{\partial r} =
  \frac{\partial f}{\partial t} =  0$ implies
\begin{alignat}{1}
\frac{1-\alpha}{\alpha} &= 3 \frac{\rme^y-1-y}{\rme^z-1-z} \label{eq:dalpha} \\
\frac{1-r}{r} &= \frac{3}{2} \cdot \frac{y}{z} \cdot \frac{(1-3r+2t)}{(r-t)}
\label{eq:dr} \\
\frac{(r-t)^3}{(1-3r+2t)^2} &= \frac{3 t}{4} .  \label{eq:dt}
\end{alignat}

We will show that the maximum for $f$ is obtained by setting $x=y=z$.
The intuition for this is that
$x$, $y$, and $z$ correspond to the parameters for the truncated Poisson random variables
used to model the degrees of all the variables
in the 2-core, the variables in the 2-core that are not in the set of
$\a n$ variables, and the variables in the 2-core that are the set of
$\a n$ variables, respectively.  Since $f$ is considering all possible sets
of size $\a n$, it is reasonable to guess that the expected degrees for the
variables in each set are equal when $f$ is maximized.

Plugging $x=y=z$ into
(\ref{eq:dalpha}), (\ref{eq:dr}) and (\ref{eq:dt}) gives
$\alpha = \frac{1}{4}, r = \frac{1}{4}, t = \frac{1}{16}$ as a point where each of the partial
derivatives are 0.  To confirm that this point is in the
domain of $f$, we note that setting $\a=r$ in (\ref{eq:implicits}) yields
$x = y = z$.  Thus $\a=\frac{1}{4}$, $r=\frac{1}{4}$, $t = \frac{1}{16}$
is a stationary point of $f$ with
\[
f\left(\frac{1}{4}, \frac{1}{4}, \frac{1}{16}\right) = 2(1-c) \ln 4 .
\]
Lemma \ref{lpd}, proved below, establishes that this is a local maximum.
In the appendix, we will present a computer-aided proof that it is a global maximum.
As part of the proof, the computer program gives a rigorous upper bound of
$f^+(\a,r,t) \geq \lim_{(\a', r', t') \to (\a, r, t)} f(\a',r',t')$ for each triple
$\a, r, t$ on the boundary, and the program proves that $2(1-c)\ln 4 - f^+(\a, r, t) \geq b$
where $b > 0$ is a constant independent of $\a$, $r$, and $t$, but it may be dependent on $c$.

\subsection{Proof of Lemma \ref{lpd}}
Recall the statement:

\begin{prevlemma}{\ref{lpd}}
 The quadratic form $\sum\sum a_{ij}x_ix_j$ is positive definite, and the determinant of
the matrix $(a_{ij})$ is $D=\frac{2^{15}}{9}c^2(\Phi(3cn,n))^2$.
\end{prevlemma}

\proofstart
The second partial derivatives of $f$ are:
\begin{eqnarray*}
f_{\alpha \alpha} & = & -\frac{1}{\alpha} - \frac{1}{1-\alpha}
    + \frac{\rme^z-1}{\rme^z-1-z} \frac{\partial z}{\partial \alpha}
    - \frac{\rme^y-1}{\rme^y-1-y} \frac{\partial y}{\partial \alpha} \\
f_{\alpha r} & = & \frac{\rme^z-1}{\rme^z-1-z} \frac{\partial z}{\partial r}
    - \frac{\rme^y-1}{\rme^y-1-y} \frac{\partial y}{\partial r} \\
f_{\alpha t} & = & 0 \\
f_{r \alpha} & = & \frac{-3c}{z}\frac{\partial z}{\partial \alpha}
        + \frac{3c}{y}\frac{\partial y}{\partial \alpha} \\
f_{rr} & = & - \frac{3c}{r-t} + \frac{3c}{r} + \frac{3c}{1-r}
        - \frac{9c}{1-3r+2t} - \frac{3c}{z}\frac{\partial z}{\partial r}
        + \frac{3c}{y}\frac{\partial y}{\partial r} \\
f_{rt} & = & \frac{3c}{r-t} + \frac{6c}{1-3r+2t} \\
f_{t\alpha} & = & 0 \\
f_{tr} & = & \frac{3c}{r-t} + \frac{6c}{1-3r+2t} \\
f_{tt} & = & - \frac{3c}{r-t} - \frac{c}{t} - \frac{4c}{1-3r+2t}
\end{eqnarray*}
where, from (\ref{eq:implicits}),
\begin{eqnarray*}
\frac{\partial z}{\partial \alpha} & = & \frac{-z(\rme^z-1)^2}{\alpha(\rme^z-1)^2
        +3rc(\rme^z(\rme^z-1-z)-(\rme^z-1)^2)} \\
\frac{\partial z}{\partial r} & = & \frac{\alpha z (\rme^z-1)^2}{r[\alpha(\rme^z-1)^2
        +3rc(\rme^z(\rme^z-1-z)-(\rme^z-1)^2)]} \\
\frac{\partial y}{\partial \alpha} & = & \frac{y(\rme^y-1)^2}{(1-\alpha)(\rme^y-1)^2
        +3(1-r)c(\rme^y(\rme^y-1-y)-(\rme^y-1)^2)} \\
\frac{\partial y}{\partial r} & = & \frac{-(1-\alpha)y(\rme^y-1)^2}{(1-r)[(1-\alpha)(\rme^y-1)^2
        +3(1-r)c(\rme^y(\rme^y-1-y)-(\rme^y-1)^2)]}.
\end{eqnarray*}

Since $x=y=z$ at the maximum and setting
$K = \frac{c(\rme^x-1)^2}{(\rme^x-1)^2+3c(\rme^x-x\rme^x-1)}$,
we have
\begin{eqnarray*}
f_{\alpha \alpha}(\frac{1}{4},\frac{1}{4},\frac{1}{16}) & = & - \frac{16}{3} - 16 K \\
f_{\alpha r}(\frac{1}{4},\frac{1}{4},\frac{1}{16}) & = & 16 K \\
f_{\alpha t}(\frac{1}{4},\frac{1}{4},\frac{1}{16}) & = & 0 \\
f_{r \alpha}(\frac{1}{4},\frac{1}{4},\frac{1}{16}) & = & 16 K \\
f_{r r}(\frac{1}{4},\frac{1}{4},\frac{1}{16}) & = & - 24 c - 16 K \\
f_{r t}(\frac{1}{4},\frac{1}{4},\frac{1}{16}) & = & 32 c \\
f_{t \alpha}(\frac{1}{4},\frac{1}{4},\frac{1}{16}) & = & 0 \\
f_{t r}(\frac{1}{4},\frac{1}{4},\frac{1}{16}) & = & 32 c \\
f_{t t}(\frac{1}{4},\frac{1}{4},\frac{1}{16}) & = & - \frac{128}{3} c.
\end{eqnarray*}
Thus,
\[ (a_{ij}) = \left[ \begin{array}{ccc}
    \frac{16}{3}+16K & -16K & 0 \\
    -16K & 24c+16K & -32c \\
    0 & -32c & \frac{128}{3}c \end{array} \right]. \]

The quadratic form is positive definite if the
following determinants are all positive (see, e.g., \cite{apo57} p. 152).
\begin{alignat*}{1}
\left| a_{11} \right| &= \frac{16}{3}+16K \\
\left| \begin{array}{cc} a_{11} & a_{12} \\
                a_{21} & a_{22} \end{array} \right|
  &=  \left(\frac{16}{3} + 16K\right)\left(24c+16K\right) - (16 K)^2 \\
  &= 16^2\left( \left(\frac{1}{3} + K\right)\left(\frac{3}{2}c + K\right) - K^2\right) \\
  &= 2^8\left(\frac{c}{2} + \frac{K}{3} +\frac{3cK}{2}\right) \\
\left| \begin{array}{ccc}
    a_{11} & a_{12} & a_{13} \\
    a_{21} & a_{22} & a_{23} \\
    a_{31} & a_{32} & a_{33} \end{array} \right|
 &=
 \left( \frac{16}{3}+16K\right) \left| \begin{array}{cc}
        a_{22} & a_{23} \\ a_{32} & a_{33} \end{array} \right|
    \mbox{} - (-16K)\left| \begin{array}{cc}
        a_{21} & a_{23} \\  a_{31} & a_{33} \end{array} \right| \\
  &= 16
	\left( \left(\frac{1}{3}+K\right)\left[2^{10}c^2 + \frac{2^{11} K c}{3} - 2^{10}c^2\right] \right. \\
 & \qquad \qquad \left. - (- K) \left[ - \frac{2^{11} K c}{3}\right]\right) \\
  &= 2^{15}
	\left( \frac{K c}{9}+\frac{K^2 c}{3} - \frac{K^2 c}{3} \right) \\
  &=  \frac{2^{15}}{9}Kc.
\end{alignat*}

From (\ref{eq:implicits}), $\rme^x = 1 + \frac{3cx}{3c-x}$.  Using this, we
can simplify $K$.
\begin{alignat}{1}
K &= \frac{c(\rme^x-1)^2}{(\rme^x-1)^2+3c(\rme^x-1-x\rme^x)} \notag \\
  &= \frac{c(3cx)^2}{(3cx)^2+3c(3cx)(3c-x)-3cx(3c-x)^2-(3cx)^2(3c-x)} \notag \\
  &= \frac{3c^2}{x(3c-1)-3c(3c-2)} \notag \\
  &= c \times \frac{3cn^2}{xn(3cn-n)-3cn(3cn-2n)} \notag \\
  &= c(\Phi(3cn,n))^2.
\end{alignat}
Thus, $K\geq0$, which yields that the determinants above are all positive and so the quadratic form is positive definite. This establishes that the determinant $D$ of $(a_{ij})$ is $\frac{2^{15}}{9} K c$, as required.
\proofend

\subsection{Approximating $g$}
First we prove our bound on $g$.
Recall that
\begin{alignat*}{1}
g(\alpha,r,t) &= (\Phi(3cn,n)\Lambda_{3m,n})^{-1} \Phi(r3cn,\alpha n) \Phi((1-r)3c n,(1-\alpha) n)
		\Lambda_{3m,n}^{\a,r,t} \\
 & \mbox{} \times \frac{1}{(2 \pi n)^{3/2} c\sqrt{\a (1-\a)(1-3r+2t)(3r-3t)t }} \\
\intertext{where}
\Phi(i,j) &= \sqrt{\frac{ij}{z_0j(i-j)-i(i-2j)}},
\end{alignat*}
and where $\Lambda_{3m,n}$ is the probability that if we distribute $3m$ objects from
$m$ triples into
$n$ sets such that set receives at least 2 objects, then we do not have two objects from the same
triple in the same set. $\Lambda_{3m,n}^{\a,r,t}$ is a similar probability but with the distribution
further restricted by $\a$, $r$, and $t$.
Our main job is to approximate the $\Lambda$ terms:
\begin{lemma} \label{lambound}
Let $\a = \frac{1}{4}+o(1)$, $r = \frac{1}{4}+o(1)$, and $t = \frac{1}{16}+o(1)$, then
\[\Lambda_{3m,n}=\Lambda^{\a,r,t}_{3m,n}+o(1).\]
\end{lemma}

\proofstart
We recast the setting slightly, focusing only on what is relevant to this lemma:
we consider a random partition of $3m$ elements into $n$ buckets, where the
elements come in $m$ triples, and the partition is uniform conditional on each bucket containing at least two elements.
$\Lambda_{3m,n}$ is the probability that no bucket contains two elements from one triple.

Now suppose that we split the buckets into two groups $A,B$ with $|A|=\a n,|B|=(1-\a)n$, and fix
\begin{itemize}
\item the total number of elements that land in the buckets of $A$ is $3rm$;
\item the number of triples, all of whose elements lie in $A$, is $tm$.
\end{itemize}
$\Lambda^{\a,r,t}_{3m,n}$ is the probability, under this conditioning, that no bucket contains two elements from one triple.

To approximate $\Lambda_{3m,n}$, we will first expose the number of elements in each bucket, and then
we will expose which elements go into the buckets.  Let $B_i$ denote the number of elements
in bucket $i$, and let
\[\B=\sum_{i=1}^n B_i(B_i-1).\]
Having exposed the bucket sizes, we take a uniformly random partition of the elements into parts
matching the bucket sizes. Let $Y$ denote the number of pairs of elements that lie in the same triple and
in the same bucket. There are $3m$ pairs that lie in the same triple, and so
\begin{eqnarray*}
\E(Y)=\mu_Y&=&3m\times\sum_{i\geq 1}B_i(B_i-1)\times\frac{1}{(3m)(3m-1)}\\
&\sim&\inv{3m}\times\B\\
&=&\frac{1}{3m}\B.
\end{eqnarray*}

For each $j\geq2$ a straightforward calculation shows that the expected number of $j$-tuples of pairs
that are counted by $Y$ is $\E(Y)^j+o(1)$; the $o(1)$ term comes from the small probability that all three
elements of a triple lie in the same bucket, and from the fact that the $\frac{1}{(3m)(3m-1)}$ term becomes
$\frac{1}{(3m)(3m-1)...(3m-2j+1)}$ rather than $\left(\frac{1}{(3m)(3m-1)}\right)^j$.  Thus, the ``method of moments'' (see, eg.\ Section 6.1 of \cite{jan-luc-ruc00}) implies that $Y$ is asymptotically distributed like a Poisson and so
\[\Lambda_{3m,n}=\pr(Y=0)=\rme^{-\mu_Y}+o(1).\]

For $\Lambda^{\a,r,t}_{3m,n}$, we take the same approach, first dealing with the buckets in $A$ and then
with those in $B$.  Let $B'_i$ denote the the number of elements in bucket $i$ (and note that since this
is a different experiment, we do not necessarily have $B_i=B'_i$).  Then let
\[\B_A=\sum_{i\in A} B'_i(B'_i-1);\qquad\qquad\B_B=\sum_{i\in B} B'_i(B'_i-1). \]

The number of triples with all three elements in $A$ is $tm$, and
as discussed in Section \ref{sec:3-4-thresh}, no triples have two elements in $A$.
So letting $Y_A$ denote the number of pairs of elements that lie in the same triple and in the same bucket of $A$, we have:
\begin{eqnarray*}
\E(Y_A)=\mu_A&=&3tm\times\sum_{i\in A} B'_i(B'_i-1)\times\frac{1}{(3rm)(3rm-1)}\\
&\sim&\frac{t}{3r^2m}\times\B_A\\
&\sim&\frac{1}{3m}\B_A.
\end{eqnarray*}
From Section \ref{sec:3-4-thresh}, the number of triples with all three elements in $B$ is $T_3=m-3rm+2tm$ and the
number with two elements in $B$ is $T_2=3rm-3tm$. So, defining $Y_B$ analogously to $Y_A$, we have:
\begin{eqnarray*}
\E(Y_B)=\mu_B&=&\left(3(m-3rm+2tm)+(3rm-3tm)\right)
\times\sum_{i\in B} B'_i(B'_i-1) \\ & & \mbox{} \times\frac{1}{(3(1-r)m)(3(1-r)m-1)}\\
&\sim&\frac{1-2r+t}{3(1-r)^2m}\times\B_B\\
&\sim&\frac{1}{3m}\B_B.
\end{eqnarray*}
Again, the method of moments yields:
\[\Lambda^{\a,r,t}_{3m,n}=\pr(Y_A=0)\pr(Y_B=0)=\rme^{-\mu_A-\mu_B}+o(1).\]
We will show below that a.a.s.\ $\B=\B_A+\B_B+o(1)$ and so $\mu_Y=\mu_A+\mu_B+o(1)$, thus yielding
the lemma.

To study $\B,\B_A,\B_B$, we will apply  Poissonization to allow us to treat $B_i,B'_i$ as independent variables.
(See Section 5.4 of \cite{mitup} for a discussion of this technique.)  Let $Z(\mu)$ denote a Poisson variable
with mean $\mu$ and let $Z_{\geq2}(\mu)$ denote the same variable but truncated as being at least two; i.e., $\pr(Z_{\geq2}(\mu)=i)$ is $0$ for $i<2$ and is
$\pr(Z(\mu)=i)/\pr(Z(\mu)\geq2)$ otherwise. Let $b_1,...,b_n$ be independent variables distributed like $Z_{\geq2}(3m/n)$ and let $b$ be the random vector $(b_1,...,b_n)$ conditional on the
event $E$ that $b_1+...+b_n=3m$.  It is straightforward to show that this vector has the same distribution
as the vector $(B_1,...,B_n)$.  Indeed, for any $x_1,...,x_n\geq2$ that sum to $3m$, the probability
that $(B_1,...,B_n)=(x_1,...,x_n)$ is
\[\binom{3m}{x_1,...,x_n}\inv{S(3m,n,2)},\]
and the probability that $(b_1,...,b_n)=(x_1,...,x_n)$ is
\begin{eqnarray*}\inv{\pr(E)}\prod_{i=1}^n\frac{\pr(Z(3m/n)=x_i)}{\pr(Z(3m/n)\geq2)}
&=&\inv{\pr(E)\pr(Z(3m/n)\geq2)^n}\prod_{i=1}^n\rme^{-3m/n}\frac{(3m/n)^x_i}{x_i!}\\
&=&\frac{\rme^{-3m}(3m/n)^{3m}}{\pr(E)\pr(Z(3m/n)\geq2)^n}\times\inv{\prod_{i=1}^n x_i!}.
\end{eqnarray*}
These probabilities are both proportional to $1/\prod_{i=1}^n x_i!$ and hence are both equal.

Let $x_1,...,x_n$ be a sequence of $n$ independent variables distributed like $Z_{\geq2}(3m/n)$.
Note that $\pr(E)=\Theta(n^{-1/2})$.  Therefore, if a property fails with probability $o(n^{-1/2})$
for $x_1,...,x_n$ then it holds a.a.s.\ for $B_1,...,B_n$.
(In fact, for monotone properties, one can show that if it holds a.a.s.\ for
$x_1,...,x_n$ then it holds a.a.s.\ for $B_1,...,B_n$, but we won't need this here.)
Let $\B^*=\sum_{i=1}^n x_i(x_i-1)$.  It is straightforward to show, eg.\ by Azuma's Inequality,
that $\B^*$ is highly concentrated around it's mean; eg.\ that
$\pr(|\B^*-\ex(\B^*)|\geq n^{2/3})\ll n^{-1/2}$.
Therefore, a.a.s.\ $|\B-\ex(\B^*)|<n^{2/3}$.

The same analysis shows that a.a.s.\ $\B_A$, $\B_B$ are concentrated around $\ex(\B_A^*)$, $\ex(\B_B^*)$, where
$\B_A^*$, $\B_B^*$ are defined in the same manner except that the means of the truncated Poisson are
$3rm/\a n, 3(1-r)m/(1-\a)n$ rather than $3m/n$.  Since $\a=r+o(1)$, we have
$\ex(\B_A^*)=\frac{|A|}{n}\ex(\B^*)+o(1) ,\ex(\B_B^*)=\frac{|B|}{n}\ex(\B^*)+o(1)$ and
so a.a.s.\ $\B=\B_A+\B_B+o(1)$ as required.
\proofend

{\bf Remark:}  It seems fortuitous that we obtained  $\mu_A+\mu_B\sim\mu_Y$ here.  Had this not occurred,
we still would have had  $\Lambda_{3m,n}=L\times\Lambda^{\a,r,t}_{3m,n}+o(1)$, for some constant $L$
which would have yielded $\frac{\ex(N^2)}{\ex(N)^2}\sim L$.  This, in turn would have showed that
the probability of satisfiability is at least a constant.  Then, a  standard application of Friedgut's Theorem~\cite{fri99} would have shown that the probability of satisfiability is indeed $1-o(1)$. (See, eg.~\cite{ach01}
for a similar argument.)

This now easily yields Lemma~\ref{lem:distribution}:

\begin{prevlemma}{\ref{lem:distribution}}
If $\a=\frac{1}{4}+o(1),r=\frac{1}{4}+o(1),t=\frac{1}{16}+o(1)$ then
\[g(\a,r,t) = \Phi(3cn,n)
\times (\pi n)^{-\frac{3}{2}}\frac{2^6}{9c}+o(1).\]
\end{prevlemma}

\begin{proof}
This follows immediately from Lemma~\ref{lambound} and a straightforward simplification.
\end{proof}

We close this section with the proof of
Lemma~\ref{lgbound}:

\begin{prevlemma}{\ref{lgbound}}
There exists a constant $\n$, independent of $n$,
such that for all $\alpha \in I_n^+$, $ r \in I_{3m}^+$,
and $t \in I_m^+$
we have $g(\a,r,t) \leq \n$. 
\end{prevlemma}

\begin{proof}
We will consider the terms of $g(\a, r, t)$ separately.
Recall that $I_n^+ = [\zeta, 1-\zeta]$ for some constant $\zeta$.  
As a result, $\frac{1}{\sqrt{2 \pi n \a (1-\a)}} < 1$ when $\a \in I_n^+$. 
Similarly, $\frac{1}{\sqrt{4 \pi^2 m^2(1-3r+2t)(3r-3t)t}} < 1$ 
when $t \in I_m^+$.

Next we consider the $\Phi$ terms of $g$.
Let $\chi(i,j) = \sqrt{\frac{j(i-2j)}{z_0j(i-j)-i(i-2j)}}$ where $\frac{i}{j} = \frac{z_0\left(\rme^{z_0}-1\right)}{\rme^{z_0}-1-z_0}$,
then $\Phi(i,j) = \chi(i,j) \times \frac{\sqrt{2 \pi i}}{\sqrt{2\pi(i-2j)}}$.
We define
$\chi(z_0) = \chi(i, j)$ where
\begin{alignat*}{1}
\chi(i,j) &=
\sqrt{\frac{j(i-2j)}{z_0 j(i-j)-i(i-2j)}} \\
  &= \sqrt{\frac{\frac{z_0(\rme^{z_0}-1)}{\rme^{z_0}-1-z_0} - 2}
      {z_0\left(\frac{z_0(\rme^{z_0}-1)}{\rme^{z_0}-1-z_0} - 1\right) - \frac{z_0(\rme^{z_0}-1)}{\rme^{z_0}-1-z_0}\left(\frac{z_0(\rme^{z_0}-1)}{\rme^{z_0}-1-z_0}-2\right)}} \\
&= \sqrt{\frac{z_0(\rme^{z_0}-1)-2(\rme^{z_0}-1) + 2z_0}{z_0(\rme^{z_0}-1)+z_0^2\left(\frac{\rme^{z_0}-z_0\rme^{z_0}-1}{\rme^{z_0}-1-z_0}\right)}} \\
  &= \sqrt{\frac{z_0(\rme^{z_0}-1)-2(\rme^{z_0}-1-z_0)}{z_0(\rme^{z_0}-1)-z_0^2\left(\frac{z_0(\rme^{z_0}-1)}{\rme^{z_0}-1-z_0}-1\right)}} \\
  &= \chi(z_0) .
\end{alignat*}
As a result, we have
\begin{alignat}{1}
\Phi(r3cn, \a n) &= \chi(z) \times \frac{\sqrt{2 \pi(3rcn)}}{\sqrt{2 \pi (3rcn-2\a n)}} \label{eq:phi1} \\
\Phi((1-r)3cn, (1-\a)n) &= \chi(y) \times \frac{\sqrt{2 \pi (3(1-r)c n)}}{\sqrt{2 \pi (3(1-r)cn-2(1-\a)n)}} \label{eq:phi2} \\
\Phi(3cn, n) &= \chi(x) \times \frac{\sqrt{2 \pi (3cn)}}{\sqrt{2 \pi (3cn-2n)}} .
\end{alignat}

It is straightforward to confirm that $\lim_{z_0 \to 0} \chi(z_0) = 1$, $\lim_{z_0 \to \infty}\chi(z_0) = 1$,
$\chi(z_0)$ has a minimum at $z_0 \approx 4.2$, and $0.84 < \chi(z_0) < 1$ for all positive $z_0$.
In addition, each of the square root terms in (\ref{eq:phi1}) and (\ref{eq:phi2}) is $\bigth(\sqrt{n})$ because $r \in I_{3m}^+$.
Also, $\Phi(3cn, n)^{-1}$ is $\bigo(1)$.  (It is $\bigth(1)$ if $c$ is $\frac{2}{3} + \bigth(1)$, and it is $\bigth(n^{-1/2})$ if $c$ is $\frac{2}{3} + O(n^{-1})$.
However, $r \in I_{3m}^+$ implies $c$ is $\frac{2}{3} + \bigth(1)$.)
As
a result, we have $\frac{\Phi(r3m,\a n) \Phi((1-r)3m,(1-\a n))}{\Phi(3m,n)} = \bigo(1)$.

Finally, we consider the $\Lambda$ terms of $g$.  Clearly, $\Lambda_{3m,n} > 0$, and in
the proof of Lemma~\ref{lambound} it is implied that $\Lambda_{3m, n}$ equals a constant
independent of $n$ (because $\B$ is a.a.s.\ $\bigth(n)$, $\ex(Y) = \bigth(1)$, and
thus $\pr(Y = 0) = \bigth(1)$).
This proves that $\frac{\Lambda_{3m,n}^{\a,r,t}}{\Lambda_{3m,n}} = \bigo(1)$,
and that completes the proof.
\end{proof}

\subsection{The border regions}\label{sborder}
In this subsection, we handle the values of $\a,r,t$ that are close to the borders of the feasible
region; i.e.  we prove Lemma \ref{lborder}:

\begin{prevlemma}{\ref{lborder}}
Let $c > \frac{2}{3}(1 + 2\z)$. 
The sum of $F(\a,r,t)$ over all feasible $\a,r,t$ not satisfying $\a\in I_n^+, r\in I_{3m}^+,t\in I_m^+$, is $o(4^{2n-2m})$.
\end{prevlemma}

\begin{proof}
Recall that $m=cn$ for a fixed constant $c$, and so we can choose $\z$ to be arbitrarily small in terms of $c$.
Note that if $c < \frac{2}{3}(1 + 2\z)$ then the lemma does not hold because the global maximum at $\a = \frac{1}{4}$,
$r = \frac{1}{4}$, $t=\frac{1}{16}$ has an $r$ coordinate smaller than $\frac{2\a + \z}{3c}$.

In this proof, we will approximate $F(\a,r,t)$ as in (\ref{eap1}), 
except that we must be careful about certain terms.
In particular, when we are too close to the border, some of the terms can be so small that our asymptotic approximations are not sufficiently close;
so instead we will use the following absolute bounds.

The first are:
\begin{alignat}{1}
 i! &\geq \left(\frac{i}{\rme}\right)^i ,
\label{elb1} \\
 i! &< \left(\frac{i}{\rme}\right)^i\sqrt{2 \pi i + 2} ,
\label{eub1} \\
\intertext{and}
\binom{n}{k} &\leq \left(\frac{n \rme}{k}\right)^k .
\label{bub1}
\end{alignat}

The next two bounds are on $S(i,j,2)j!$.  Recall that this is the number of ways to partition $i$ labeled elements into $j$ labeled buckets
such that each bucket has at least 2 elements.  The first bound is easily obtained by omitting the restriction that each bucket contain at least two elements:
\begin{equation}\label{eub2}
 S(i,j,2)j!\leq j^i.
\end{equation}

The second is obtained by first choosing 2 elements for each bucket, and then distributing the remaining elements. Clearly, different choices can yield the same partition, and so this is an overcount:

\begin{equation}\label{eub3}
 S(i,j,2)j!\leq {i\choose 2j} \frac{(2j)!}{2^j}\times j^{i-2j}.
\end{equation}

Now we turn to the border regions. We consider 8 cases, and note that some $(\a,r,t)$ will lie in more than one case.
In the first two cases, we prove that the sum of the triples that satisfy those cases is $\smallo(4^2(n-m))$, 
and this proves the lemma for these cases.

\noindent{\em Case 1:} $\a = 0$ or $\a = 1$.

If $\a = 0$, we must have $r = t = 0$, and $F(0,0,0) = 2^{(2-c)n} 3^{(1-2c)n}$.
If $\a = 1$, we must have $r = t = 1$, and $F(1,1,1) = 4^{n-m}$.

In the cases that follow, we will assume $0 < \a < 1$.

\noindent{\em Case 2:} $\a>(1-\z)$.

This case is easily dealt with, since the number of pairs of assignments which agree on exactly $\a n$
variables is $4^n{n\choose n- \a n}3^{n-\a n}$. So $F(\a,r,t)$ is bounded by the number of such pairs, 
times the probability that the first one is satisfied, which, from bound (\ref{bub1}), is at most:
\[4^{n}\left(\frac{\rme n}{\z n}\right)^{\z n} 3^{\z n} 4^{-m}=4^{n-m}\left(\frac{3 \rme }{\z }\right)^{\z n}<o(4^{2(n-m)}/n^3),\]
for $\z$ sufficiently small.  So the sum of $F(\a,r,t)$ over the $O(n^3)$ triples $(\a,r,t)$ satisfying Case 2 is
$o(4^{2(n-m)})$.

For the remaining cases, we will approximate $F(\a,r,t)$ as in (\ref{eap1}), except that we must be careful about certain terms.
In particular, when we are too close to the border, some of the terms can be so small that our asymptotic approximations are not sufficiently close;
so instead we will use the absolute bounds from above.
In each case below, 
we will note how much this increases the bound over that obtained in (\ref{eap1}).

\vspace{1ex}

\noindent{\em Case 3:} $\a<\z$.

Instead of using the approximation $\binom{n}{\a n} \sim \a^{-\a n} (1-\a)^{-(1-\a)n} \frac{1}{\sqrt{2 \pi n \a (1-\a)}}$,
we use (\ref{bub1}) to get the bound $\binom{n}{\a n} < \left(\frac{\rme}{\a}\right)^{\a n}$.
If we first consider the exponential terms,
this replacement results in an increase over the bound in (\ref{eap1}) by a factor of
\begin{eqnarray*}
\frac{\rme^{\a n} \a^{-\a n}}{\a^{-\a n} (1-\a)^{-(1-\a)n}}
 & = & \rme^{\a n} (1-\a)^{(1-\a)n} \\
 &<& \rme^{\g n}.
\end{eqnarray*}
{So the bound on the exponential terms in
(\ref{eq:f}) is increased by at most a constant $\g > 0$} which can be made arbitrarily small by taking $\z$ to be sufficiently small.
The replacement also removes the term $(2 \pi n \a (1-\a)) ^{-1/2}$ from (\ref{eq:g}). {In the proof of 
Lemma~\ref{lgbound} this term is bounded by 1, so this change does not alter the bound of Lemma~\ref{lgbound}.}

Next we replace the term $S(r3m, \a n, 2)(\a n)!$. If $r$ is small enough that Case 4 below applies, we
make the replacement described in that case.  Otherwise, 
we replace the approximation 
$S(r3m,\alpha n,2)(\alpha n)!
\sim (\rme^z-1-z)^{\alpha n} \left(\frac{3rm}{\rme z}\right)^{3rm}\Phi(r3m,\alpha n)$
with the bound obtained from (\ref{eub2}):  $S(r3m,\alpha n,2)(\alpha n)!<(\a n)^{3rm}$.  
We will consider the exponential terms first.
Note that this replacement
results in an increase over the bound {in (\ref{eap1})} by a factor of
\begin{eqnarray*}
\lefteqn{\frac{(\a n)^{3rm}}{(\rme^z-1-z)^{\alpha n}\left(\frac{3rm}{\rme z}\right)^{3rm}}} \\
&=& \rme^{3rcn} \left(\frac{z\a}{3rc}\right)^{3rcn}(\rme^z-1-z)^{-\alpha n}\\
&=& \rme^{3rcn} \left(\frac{\rme^z - 1 -z}{\rme^z-1}\right)^{3rcn}(\rme^z-1-z)^{-\alpha n}\\
&=& \frac{\rme^{3rcn}}{(\rme^z-1)^{\a n}} \left(\frac{\rme^z - 1 -z}{\rme^z-1}\right)^{3rcn - \a n}\\
&=& \left(\frac{\rme}{(\rme^z-1)^{\frac{\rme^z-1-z}{z(\rme^z-1)}}}\right)^{3rcn} \left(\frac{\rme^z - 1 -z}{\rme^z-1}\right)^{(3rc - \a)n}\\
&<&\rme^{\g n}.
\end{eqnarray*}
{So the bound on the exponential terms in
(\ref{eq:f}) is increased by at most} a constant $\g>0$ which can be made arbitrarily small by taking $\z$ to be sufficiently small.
This follows from the fact that if $r$ is reasonably large so that Case 4 below does not apply, then
as $\z$ and therefore $\a$ tends to 0,
$z$ tends to $\infty$, and $\lim_{z \to \infty}(\rme^z-1)^{\frac{\rme^z-1-z}{z(\rme^z-1)}} = \rme$.

This change also removes the term $\Phi(3rm, \a n)$ from (\ref{eq:g}), and from the proof of 
Lemma~\ref{lgbound} this change increases the bound of Lemma~\ref{lgbound} by at most a constant factor.

\vspace{1ex}

\noindent{\em Case 4:} $r<\frac{2\a +\z}{3c}$.

Rather than using the approximation
\[S(r3m,\alpha n,2)(\alpha n)! \sim (\rme^z-1-z)^{\alpha n}\left(\frac{3rm}{\rme z}\right)^{3rm}\Phi(r3m,\alpha n),\]
we will use the bound obtained from (\ref{eub3}):
\begin{eqnarray*}
S(r3m,\alpha n,2)(\alpha n)! &<& {r3m\choose 2\a n} \frac{(2\a n)!}{2^{\a n}}\times (\a n)^{r3m-2\a n} \\
&<& \frac{(3rm/\rme)^{3rm}}{((3rm-2\a n)/\rme)^{3rm-2\a n}} 2^{-\a n}(\a n)^{3rm-2\a n}\sqrt{2 \pi 3rm + 2}
\end{eqnarray*}
with the last inequality following from bounds (\ref{eub1}) and (\ref{elb1}).
Note that this replacement increases the bound {in (\ref{eap1})} by a factor of
\begin{eqnarray*}
\lefteqn{\frac{\frac{(3rm/\rme)^{3rm}}{((3rm-2\a n)/\rme)^{3rm-2\a n}} 2^{-\a n}(\a n)^{3rm-2\a n}}
{(\rme^z-1-z)^{\alpha n} \left(\frac{3rm}{\rme z}\right)^{3rm}}}\\
&=&\left(\frac{\a n \rme}{3rm-2\a n}\right)^{3rm-2\a n} \frac{z^{3rm}}{(2(\rme^z-1-z))^{\a n}} \\
&<&\rme^{\g n}.
\end{eqnarray*}
{So the bound on the exponential terms in
(\ref{eq:f}) is increased by at most}  a constant $\g>0$ which can be made arbitrarily small by taking $\z$ to be sufficiently small.
The last inequality holds because as $\z$ tends to 0, $3rc$ tends to $2\a$ and $z$ tends to 0,
and $\lim_{z \to 0} \frac{z^2}{2(\rme^z-1-z)} = 1$.

This change also replaces the term $\Phi(3rm, \a n)$ in (\ref{eq:g}) with $\sqrt{2 \pi 3rm + 2}$, and from the proof of 
Lemma~\ref{lgbound} this change increases the bound of Lemma~\ref{lgbound} by a factor of $\bigo(\sqrt{n})$.

\vspace{1ex}

\noindent{\em Case 5:} $r>1-\frac{2(1-\a) +\z}{3c}$.

Rather than using the approximation
\begin{alignat*}{1}
S((1-r)3m,(1-\alpha)n,2)((1-\alpha)n)! \sim&
(\rme^y-1-y)^{(1-\alpha)n} \left(\frac{3(1-r)m}{\rme y}\right)^{3(1-r)m}
  \\ & \times \Phi((1-r)3m,(1-\alpha)n),
\end{alignat*}
we will use the bound obtained from (\ref{eub3}):
\begin{eqnarray*}
\lefteqn{S((1-r)3m,1-\alpha n,2)((1-\alpha) n)!} \\
 &<& {(1-r)3m\choose 2(1-\a) n} \frac{(2(1-\a) n)!}{2^{(1-\a) n}}\times ((1-\a) n)^{(1-r)3m-2(1-\a) n} \\
&<& \frac{(3(1-r)m/\rme)^{3(1-r)m}}{((3(1-r)m-2(1-\a) n)/\rme)^{3(1-r)m-2(1-\a) n}} 2^{-(1-\a) n} \\
& & \times ((1-\a) n)^{3(1-r)m-2(1-\a) n}\sqrt{2 \pi 3(1-r)m + 2}
\end{eqnarray*}
with the last inequality following from bounds (\ref{eub1}) and (\ref{elb1}).
Note that, from the same argument as in Case 4, this replacement increases the bound {in (\ref{eap1})} by a factor of
\begin{eqnarray*}
\lefteqn{\frac{\frac{(3(1-r)m/\rme)^{3(1-r)m}}{((3(1-r)m-2(1-\a) n)/\rme)^{3(1-r)m-2(1-\a) n}} 2^{-(1-\a) n}((1-\a) n)^{3(1-r)m-2(1-\a) n}}
{(\rme^y-1-y)^{(1-\alpha) n} \left(\frac{3(1-r)m}{\rme y}\right)^{3(1-r)m}}}\\
&=&\left(\frac{(1-\a) n \rme}{3(1-r)m-2(1-\a) n}\right)^{3(1-r)m-2(1-\a) n} \frac{y^{3(1-r)m}}{(2(\rme^y-1-y))^{(1-\a) n}} \\
&<&\rme^{\g n}.
\end{eqnarray*}
{So the bound on the exponential terms in
(\ref{eq:f}) is increased by at most}  a constant $\g>0$ which can be made arbitrarily small by taking $\z$ to be sufficiently small.

This change also replaces the term $\Phi(3(1-r)m, (1-\a) n)$ in (\ref{eq:g}) with $\sqrt{2 \pi 3(1-r)m +2}$, and from the proof of 
Lemma~\ref{lgbound} this change increases the bound of Lemma~\ref{lgbound} by a factor of $\bigo(\sqrt{n})$.

\vspace{1ex}

\noindent{\em Case 6:} $t<\z$.

To obtain (\ref{eap1}), we applied Stirling's formula and approximated $(tm)!$ by $\left(\frac{tm}{\rme}\right)^{tm}\sqrt{2\pi tm}$.
Instead, we will apply (\ref{elb1}) and replace it with $\left(\frac{tm}{\rme}\right)^{tm}$. This yields an upper bound on $F(\a,r,t)$
as $(tm)!$ appears in the denominator.  
This replacement does not change the exponential terms of $F(\a,r,t)$, but it does remove the term $\sqrt{2 \pi tm}$ from (\ref{eq:g}),
and from the proof of Lemma~\ref{lgbound} this change increases the bound of Lemma~\ref{lgbound} by at most a constant factor.

\vspace{1ex}

\noindent{\em Case 7:} $t<\frac{3r-1}{2}+\z$.

As in Case 6, we will replace the $(m-3rm+2tm)!$ term with
$\left(\frac{m-3rm+2tm}{\rme}\right)^{m-3rm+2tm}$.  As in Case 6,
this replacement does not change the exponential terms of $F(\a,r,t)$, but it does remove the term $\sqrt{2 \pi m(1-3r+2t)}$ from (\ref{eq:g}),
and from the proof of Lemma~\ref{lgbound} this change increases the bound of Lemma~\ref{lgbound} by at most a constant factor.

\vspace{1ex}

\noindent{\em Case 8:} $t>r-\z$.

As in Cases 6 and 7, we will replace the $(3rm-3tm)!$ term with
$\left(\frac{3rm-3tm}{\rme}\right)^{3rm-3tm}$. 
This replacement does not change the exponential terms of $F(\a,r,t)$, but it does remove the term $\sqrt{2 \pi m(3r-3t)}$ from (\ref{eq:g}),
and from the proof of Lemma~\ref{lgbound} this change increases the bound of Lemma~\ref{lgbound} by at most a constant factor.

If $(\a,r,t)$ satisfy Cases 1 or 2, then the lemma holds.  If they satisfy any of Cases 3-8, then we obtain the upperbound
\[F(\a,r,t)\leq (1+o(1)) g^*(\a,r,t)\rme^{f^*(\a,r,t)n},\]
where $g^*(\a,r,t)$ and $f^*(\a,r,t)$ are the result of applying the changes of Cases 3-8 to 
(\ref{eq:g}) and (\ref{eq:f}).

Note that $g^*(\a,r,t) \leq \nu \times \bigo(n) < \nu \times \rme^{\gamma n}$ where $\nu$ is the constant from Lemma~\ref{lgbound},
and note that $f^*(\a,r,t)\leq f(\a,r,t)+3\g$ (to achieve this extreme bound, $(\a,r,t)$ would have to satisfy
all Cases 3-5).
Thus we have
\[F(\a,r,t)\leq (1+o(1)) \nu\times\rme^{\g n}\rme^{(f(\a,r,t)+3\g)n}=(1+o(1)) \rme^{(f(\a,r,t)+4\g)n}.\]

{Lemma \ref{lem:fmax} implies that if $\z$ is small enough to bound the borders of $ I_n^+\times I_{3m}^+ \times I_m^+$ away from the unique global maximum point $(\inv{4},\inv{4},\inv{16})$ } then
there is some $x>0$, independent of $\z$, such that $f(\a,r,t) < 2(1-c)\ln 4 -x$ for any
feasible $\a,r,t$ not satisfying $\a\in I_n^+, r\in I_{3m}^+,t\in I_m^+$.
We take
$\z$ small enough to yield $\g=\inv{8}x$, thus obtaining
\[F(\a,r,t)<(1+o(1)) \rme^{(2(1-c)\ln 4-\hf x)n},\]
and we insure that $\z$ small enough so that $\frac{2}{3}(1+2\z) < 0.818469\ldots$, 
the threshold for a non-empty 2-core (see Lemma~\ref{f1}), {and thus we have $c>\frac{2}{3}(1+2\z) $.}

We need to be careful about one thing:  if $(\a,r,t)$ is extremal - i.e. if
$\a=0$, $\a=1$, $r=\frac{2\a}{3c}$, $r=1-\frac{2(1-\a)}{3c}$, $t=0$, $t=\frac{3r-1}{2}$ or $t=r$ - then $f(\a,r,t)$ 
and $g(\a,r,t)$
are undefined.  In any such case, we obtain $F(\a,r,t)\leq (1+o(1)) g^*(\a,r,t)\rme^{f^*(\a,r,t)n}$
where $g^*(\a,r,t)<\lim_{(\a',r',t')\rightarrow(\a,r,t)}g(\a',r',t')\times\rme^{\g n}$ and
$f^*(\a,r,t)\leq\lim_{(\a',r',t')\rightarrow(\a,r,t)} f(\a',r',t')+3\g$.  So again, Lemmas  \ref{lem:fmax} and \ref{lgbound} imply that 
\[{F(\a,r,t)<(1+o(1)) \rme^{(2(1-c)\ln 4-\hf x)n}<4^{2n-2m-\inv{4}xn}.}\]

The asymptotics in this bound come from our approximations of factorial terms and of $S(i,j,2)$ terms.  In all
cases, we either use approximations that hold uniformly over the given domain or we replace the
approximations by bounds that hold over the domain, and so the asymptotic approximations and upper
bounds hold uniformly.  Thus we can sum them over the $O(n^3)$ triples and obtain the lemma.
\end{proof}

\subsection{An Approximation for Generalized Stirling
Numbers of the Second Kind}
\label{sec:Stirling}

\begin{lemma}
\label{lemma:Stirling}
Let
\[ \G(i,j) = \frac{1}{j!} \left( \frac{i}{z_0\rme} \right) ^{i}
(\rme^{z_0}-1-z_0)^j \sqrt{\frac{ij}{z_0j(i-j)-i(i-2j)}} \]
where $z_0$ is the positive real solution of the equation
\[\frac{j}{i}z_0 = \frac{\rme^{z_0}-1-z_0}{\rme^{z_0}-1}.\]
Then
\[S(i, j, 2) \sim \G(i,j) , \]
and the approximation holds uniformly as $i \to \infty$ for all $\d_1 i < j < \frac{i}{2}(1-\d_2)$ where $\d_1$ and $\d_2$ are {any} positive
constants.
\end{lemma}

\begin{proof}
Hennecart, in \cite{hen94}, extends the approximation by Temme \cite{tem93} for Stirling 
numbers of the second kind,
$S(i,j,1)$ in the notation of this paper,
to give the following approximation
for the generalized Stirling number of the second kind $S(i,j,r)$.
\[ S(i,j,r) \sim \frac{i!}{j!(i-jr)!}
    \left(\frac{i-jr}{\rme}\right)^{i-jr}
    \frac{B^j(z_0,r)}{z_0^{i+1}}
    \sqrt{\frac{jt_0}{\phi^{''}(z_0)}} \]
where $B(z,r) = e^z - \sum_{l=0}^{r-1} \frac{z^l}{l!}$,
$\phi(z) = -i \ln z + j \ln B(z,r)$,
$t_0 = \frac{i-jr}{j}$, and $z_0$ is the positive real solution
of the equation $z_0 \frac{B'(z_0,r)}{B(z_0,r)} = \frac{i}{j}$.

If we let $r=2$, then $B(z,2) = \rme^z-1-z$, and
\begin{equation*}
\phi^{''}(z) = \frac{i}{z^2} + j\frac{\rme^z(\rme^z-1-z) - (\rme^z-1)^2}{(\rme^z-1-z)^2}.
\end{equation*}
Since, $z_0 \frac{\rme^{z_0}-1}{\rme^{z_0}-1-z_0} = \frac{i}{j}$, we have
\begin{alignat*}{1}
\rme^{z_0} &= \frac{i+iz_0-jz_0}{i-jz_0} \\
\rme^{z_0} - 1 &= \frac{iz_0}{i-jz_0} \\
\rme^{z_0}- 1 - z_0 &= \frac{jz_0^2}{i-jz_0},
\end{alignat*}
and we can simplify $\phi^{''}(z_0)$.
\begin{alignat*}{1}
\phi^{''}(z_0) &= \frac{i}{z_0^2} +
    j\frac{(i+iz_0-jz_0)jz_0^2 - (iz_0)^2}{(jz_0^2)^2} \\
 &= \frac{1}{jz_0^2} (z_0j(i-j)-i(i-2j)) .
\end{alignat*}
{This yields
\begin{alignat}{1}
S(i,j,2) &\sim \frac{i!}{j!(i-2j)!}
    \left(\frac{i-2j}{\rme}\right)^{i-2j}
    \frac{(\rme^{z_0}-1-z_0)^j}{z_0^{i+1}} \notag \\
 & \qquad \qquad \times
    \sqrt{\frac{jz_0^2(i-2j)}{z_0j(i-j)-i(i-2j)}} \label{eq:hen} .
\end{alignat}}

Unfortunately, \cite{hen94} states but does not prove that this approximation is uniform for the desired parameter
values.  As noted above, the \cite{hen94} approximation for $S(i,j,r)$ was based on the \cite{tem93} approximation for
$S(i,j,1)$, and the uniformity of the latter approximation was proven in \cite{che-ric-tem00}.  The \cite{che-ric-tem00}
proof takes the error bounds for integer $i$ and $j$ of a slightly different approximation of $S(i,j,1)$ by 
Moser and Wyman~\cite{mos-wym58} and proves that these bounds also apply to the approximation of \cite{tem93}.
In addition, \cite{che-ric-tem00} extend the proof of
\cite{mos-wym58} to prove that the \cite{tem93} approximation is also uniform for real and complex parameters.
For the results of this paper, we require (\ref{eq:hen}) to be uniform for integer $i$ and $j$ as $i \to \infty$ 
and for all $\d_1 i < j \leq \frac{i}{2}(1-\d_2)$ 
where $\d_1$ and $d_2$ are any positive constants.
Sections 3 and 4 of \cite{mos-wym58} prove the error bounds for their approximation
of $S(i,j,1)$ for integer parameters with $0 < j < i$ and $\lim_{i \to \infty} (i-j) = \infty$.
In \cite{con}, we provide a modification of these error bounds to $S(i,j,2)$, and 
this proves the desired uniformity of (\ref{eq:hen}).

{Applying Stirling's approximation to the $i!$ and $(i-2j)$! terms in (\ref{eq:hen}),}
we obtain
\begin{alignat*}{1}
S(i,j,2) &\sim \frac{1}{j!}
        \left[ \left(\frac{i}{\rme}\right)^i\left(\frac{\rme}{i-2j}\right)^{i-2j}
        \frac{\sqrt{2 \pi i}}{\sqrt{2 \pi (i-2j)}} \right]
    \left(\frac{i-2j}{\rme}\right)^{i-2j} \\
 & \qquad \qquad \times
    \frac{(\rme^{z_0}-1-z_0)^j}{z_0^{i}}
    \sqrt{\frac{j(i-2j)}{z_0j(i-j)-i(i-2j)}} \\
 &= \frac{1}{j!}
        \left(\frac{i}{z_0\rme}\right)^i
    (\rme^{z_0}-1-z_0)^j
    \sqrt{\frac{ij}{z_0j(i-j)-i(i-2j)}} \\
 &= \G(i, j).
\end{alignat*}
{The error terms from applying Stirling's approximation to $i!$ and $(i-2j)!$ are 
$\bigo(i^{-1})$ and $\bigo((i-2j)^{-1})$, respectively, and these approximations are asymptotically
tight as $i \to \infty$ with $i-2j > \d_2 i$ for $\d_2$ a positive constant.  
Therefore, we can conclude that the approximation of
Lemma~\ref{lemma:Stirling} is uniform for all $\d_1 i < j < \frac{i}{2}(1-\d_2)$ as $i \to \infty$ with
$\d_1$ and $\d_2$ any positive constants.}
\end{proof}

\section{Conclusion}

In this paper, we present $(3,4)$-UE-CSP,
the first known constraint satisfaction problem that has a random model with
an exact satisfiability threshold as well as all of the following properties of random $k$-SAT,
$k \geq 3$:  The problem is NP-complete; the problem has constant size clauses and domain;
the satisfiability threshold for the random model occurs when there are a linear number of
clauses; a random instance a.a.s.\ has exponential resolution complexity; and there is no
known polynomial time algorithm that, w.u.p.p., will find a solution to an instance drawn
from close to the satisfiability threshold.
In forming this CSP, we define UE-CSP as a general model of universal uniquely extendible CSPs
that naturally generalize XOR-SAT.

Interestingly, the satisfiability threshold for random $(3,4)$-UE-CSP
is exactly the same as for random 3-XOR-SAT.   As noted above, 3-XOR-SAT is
the same problem as $(3,2)$-UE-CSP.  This observation leads to the question as to
whether the location of the satisfiability threshold for random $(k,d)$-UE-CSP
depends on the domain size $d$.  We conjecture that it does not; i.e. that for every $k,d\geq 2$,
the satisfiability threshold for random $(k,d)$-CSP is the same as for $k$-XOR-SAT, a
threshold that was determined in \cite{dub-man02}.
\cite{con08} presents a proof for this conjecture, subject to a hypothesis that is analogous to
Lemma~\ref{lem:fmax}.  I.e., for any $k,d$, the conjecture holds so long as a natural local maximum
of a particular function is the unique global maximum.  With sufficient labor, this hypothesis
could probably be proven for specific values of $k,d$ using computer-aided interval arithmetic
analysis along the lines of that in our appendix.  But proving that it holds for all $k,d$ would
likely require a different approach.

\section{Acknowledgment}
We would like to thank two anonymous referees for their helpful suggestions.

\appendix

\section{The Proof that $f$ has a unique global maximum.}
\label{sec:fmax}
In this appendix, we present the computer-aided proof of Lemma~\ref{lem:fmax}, which we restate below.
Let
\begin{eqnarray*}
f(\alpha , r,t) & = & \ln 4 -c \ln 4 + (1-\a)\ln 3 - c(2+t-3r)\ln 3 \\
 & & \mbox{} +c(1-3r+2t)\ln 2 -\alpha \ln \alpha -(1-\alpha)\ln (1-\alpha) \\
 & & \mbox{} -c(1-3r+2t) \ln (1-3r+2t) -c(3r-3t) \ln (r-t) -ct \ln t \\
 & & \mbox{} + r3c \ln r + (1-r)3c \ln (1-r) + \alpha \ln (\rme^z-1-z)-r3c \ln z \\
 & & \mbox{} + (1-\alpha)\ln (\rme^y-1-y) - (1-r)3c \ln y - \ln (\rme^x-1-x) \\
 & & \mbox{} +3c \ln x,
\end{eqnarray*}
where $x,y,z > 0$ are defined as
\begin{equation}
\label{eq:kimplicitsMH}
\frac{\rme^x-1-x}{\rme^x-1}-\frac{x}{3c} =
\frac{\rme^y-1-y}{\rme^y-1}-\frac{y(1-\alpha)}{3c(1-r)} =
\frac{\rme^z-1-z}{\rme^z-1}-\frac{z\alpha}{3cr} = 0 .
\end{equation}

\vspace{.3ex}

\begin{prevlemma}{\ref{lem:fmax}}
For all $\frac{2}{3} \leq c < 1$, the unique global maximum of
$f(\alpha, r, t)$ in the domain $0 < \alpha < 1$, $\frac{2\a}{3c} < r < 1-\frac{2(1-\a)}{3c}$,
and $\max \{0, \frac{3r-1}{2}\} < t < r$ occurs
at $f(\frac{1}{4}, \frac{1}{4}, \frac{1}{16})=2(1-c)\ln 4$.  Furthermore, 
for all $.67 \leq c \leq 1$, there exists a constant $b = b(c)$ such that
 at every extremal $(\a,r,t)$,
$f$ approaches a limit that is less than $2(1-c)\ln 4 - b$.
\end{prevlemma}

\begin{proof}
Recall from Observation~\ref{obs:f} that $f$ is continuous and differentiable throughout the domain
$0 < \alpha < 1$, $\frac{2\a}{3c} < r < 1-\frac{2(1-\a)}{3c}$, $\max \{0, \frac{3r-1}{2}\} < t < r$.
In Section~\ref{sec:3-4-thresh}, we proved that (i) $f$ has one local maximum when $z=y$,
(ii) the point $\alpha = \frac{1}{4}$, $r = \frac{1}{4}$,
$t = \frac{1}{16}$ is this local maximum, and (iii)
\begin{equation}
\label{eq:fmax}
f\left(\frac{1}{4}, \frac{1}{4}, \frac{1}{16}\right) = 2(1-c)\ln 4 .
\end{equation}
Lemma~\ref{lem:interior} in Section~\ref{sec:eqinterior} gives an equation in terms of $z$ and $y$ for all
points where the partial first derivatives of $f$ are 0.
Lemma~\ref{lem:nointeriormax} of Section~\ref{sec:interior} proves that at each point in the domain of $f$ which
satisfies both the equation of Lemma~\ref{lem:interior} and $z \neq y$, the value of $f$ at that point is smaller than
$2(1-c)\ln 4$.
To complete the proof,
Lemma~\ref{lem:noboundarymax} of Section~\ref{sec:boundary} proves that at every point $(\a_1,r_1,t_1)$ on the
boundary of the domain of $f$ (i.e. for every extremal $(\a_1,r_1,t_1)$), $\lim_{(\a,r,t)\rightarrow (\a_1,r_1,t_1)}f(\a,r,t)$ exists and is
less than $2(1-c)\ln 4 - b$ where $b > 0$ is a constant independent of $\a_1$, $r_1$, and $t_1$, but it may
be dependent of $c$.
\end{proof}

\subsection{An Equation for All Stationary Points of $f$}
\label{sec:eqinterior}

We start by relating $z$ and $y$ at all stationary points.

\begin{lemma}
\label{lem:interior}
Every stationary point of the function $f$ of Lemma~\ref{lem:fmax} satisfies
the equation
\begin{equation}
\label{eq:interior}
\frac{z}{y} = \frac{3\left(\rme^y-1\right)^2 + \left(\rme^z-1\right)^2}
                   {2\left(\rme^y-1\right)^2+2\left(\rme^y-1\right)\left(\rme^z-1\right)}.
\end{equation}
Furthermore, for each $y > 0$, there are at exactly two values for $z$ that satisfy
(\ref{eq:interior}), one when
$z = y$ and one when $z > y$.
\end{lemma}

\begin{proof}
To find all stationary points for $f$ in the interior of the domain,
we use the partial derivatives of $f$, calculated
in Section~\ref{sec:3-4-thresh}, and from these we know that any stationary point must satisfy
all of the following equations.
\begin{alignat}{1}
\frac{1-\a}{\a} &= 3\frac{\rme^y-1-y}{\rme^z-1-z} \label{eq:maxk-alphaMH} \\
\frac{1-r}{r} &= \frac{y}{z} \cdot \frac{3}{2} \cdot \frac{(1-3r+2t)}{(r-t)}
						\label{eq:maxk-rMH} \\
\frac{(r-t)^3}{(1-3r+2t)^2} &= \frac{3t}{4} \label{eq:maxk-tiMH}
\end{alignat}

From (\ref{eq:kimplicitsMH}), we have
\[\frac{\rme^y-1-y}{\rme^z-1-z}
                  = \frac{y}{z} \cdot \frac{(1-\a)}{\a} \cdot \frac{r}{(1-r)} \cdot \frac{(\rme^y-1)}{(\rme^z-1)} .
\]
Plugging in (\ref{eq:maxk-alphaMH}), gives
\begin{equation}
\label{eq:gamma1}
3\frac{y}{z} \cdot \frac{r}{(1-r)} = \frac{\rme^z-1}{\rme^y-1},
\end{equation}
and combining (\ref{eq:gamma1}) with (\ref{eq:maxk-rMH}) gives
\begin{equation}
\label{eq:gamma2}
\frac{2(r-t)}{1-3r+2t} = \frac{\rme^z-1}{\rme^y-1}.
\end{equation}
Combining (\ref{eq:maxk-tiMH}) with (\ref{eq:gamma2}) gives
\begin{equation}
\label{eq:start}
\left(\frac{\rme^z-1}{\rme^y-1}\right)^2 (r-t) = 3t.
\end{equation}
Solving (\ref{eq:maxk-rMH}) for $t$ yields
\[t = \frac{2zr(1-r) -3yr+9yr^2}{6yr+2z(1-r)},\]
and plugging this value for $t$ into (\ref{eq:start}) gives
\begin{equation}
\label{eq:step2}
\left(\frac{\rme^z-1}{\rme^y-1}\right)^2y(1-r) = 2z(1-r)-3y+9yr.
\end{equation}
Solving (\ref{eq:gamma1}) for $r$ gives
\begin{equation}
\label{eq:rval}
r = \frac{\left(\rme^z-1\right)z}{3\left(\rme^y-1\right)y + \left(\rme^z-1\right)z},
\end{equation}
and substituting (\ref{eq:rval}) for $r$ in (\ref{eq:step2}) gives
\begin{equation*}
\frac{z}{y} = \frac{3\left(\rme^y-1\right)^2 + \left(\rme^z-1\right)^2}
                   {2\left(\rme^y-1\right)^2+2\left(\rme^y-1\right)\left(\rme^z-1\right)}.
\end{equation*}

This establishes (\ref{eq:interior}). To prove the rest of the lemma, we will look at cases.

First, we can rule out the case when
$0< z < y$.  If $0 < z < y$ then the left hand side of (\ref{eq:interior}) is smaller than
1, but it is straightforward to see that the right hand
side is equal to $1 + \frac{(\rme^y-\rme^z)^2}{2(\rme^y-1)^2 + 2(\rme^y-1)(\rme^z-1)}$ and so is
always at least 1, with equality when $z=y$.

Now we prove that for each $y>0$, there is exactly one solution with $z > y$.
We will do this by fixing $y$ and determining how each side of equation (\ref{eq:interior})
changes as $z$ increases.
We will calculate the first and second derivatives, {with respect to $z$,} of the right hand side
of (\ref{eq:interior}), and we will note that the first derivative is 0 when $z=y$, both derivatives are positive
when $z > y$, and the first derivative grows unbounded as $z$ tends to infinity and $y$ is fixed.  
That implies {that, for any fixed $y$,} the right hand side of (\ref{eq:interior}) will cross $\frac{z}{y}$
exactly once when $z > y$.

The first derivative of the RHS of (\ref{eq:interior}) with respect to $z$ is
\begin{equation*}
\frac{2\rme^z(\rme^y-1)\left((2(\rme^y-1) + 2(\rme^z-1))(\rme^z-1) -(3(\rme^y-1)^2+(\rme^z-1)^2)\right)}
     {\left(2\left(\rme^y-1\right)^2+2\left(\rme^y-1\right)\left(\rme^z-1\right)\right)^2}
\end{equation*}
which can be rewritten as
\begin{equation}
\label{firstderivMH}
\frac{\rme^z(\rme^z-\rme^y)(\rme^z-1+3(\rme^y-1))}
     {2(\rme^y-1)(\rme^z+\rme^y-2)^2}.
\end{equation}
It is straightforward to see that (\ref{firstderivMH}) is negative when $0 < z < y$,
is positive when $z > y > 0$, and grows unbounded as $z > 0$ increases.

The second derivative  of the RHS of (\ref{eq:interior}) with respect to $z$ is
\begin{multline*}
  \frac{\rme^z}{2(\rme^y-1)(\rme^z+\rme^y-2)^3}\left((\rme^z+\rme^y-2)\left[(\rme^z-\rme^y)(\rme^z-1+3(\rme^y-1))  \right.\right. \\
     \left.\left. \mbox{} + \rme^z(\rme^z-1+3(\rme^y-1)) + \rme^z(\rme^z-\rme^y)\right]\right. \\
     \left. \mbox{} - 2\rme^z(\rme^z-\rme^y)(\rme^z-1+3(\rme^y-1))\right)
\end{multline*}
which simplifies to
\begin{equation}
  \frac{\rme^z\left((\rme^z+\rme^y-2)(\rme^z-\rme^y)(\rme^z-1+3(\rme^y-1)) + 8\rme^z(\rme^y-1)^2\right)}
       {2(\rme^y-1)(\rme^z+\rme^y-2)}.
\label{secondderivMH}
\end{equation}
It is straightforward to verify that (\ref{secondderivMH}) is positive when $z > y > 0$.
\end{proof}

\subsection{Interval Analysis}

The technique of {\em interval analysis}
\cite{moo66,han-wal03} is a method to rigorously bound the range of values that a function takes over an interval.  If the endpoints of the interval are also rigorously determined,
then we have a proven upper and lower bound on that range.
Interval analysis is very useful with numerical algorithms because we can include all
errors of floating point approximation into the interval bounds.
Following standard interval analysis notation, we denote an interval with a capital letter;
eg., $X = \left[\:\underline{x}, \overline{x}\:\right]$ is an interval where
$\underline{x}$ and $\overline{x}$ are real numbers and denote the endpoints of the interval.

If we let $A = \left[\:\underline{\a}, \overline{\a}\:\right]$,
$R=\left[\:\underline{r}, \overline{r}\:\right]$, and $T=\left[\:\underline{t},\overline{t}\:\right]$
 be intervals, we can use interval analysis
techniques to compute the interval $\Phi = \left[\:\underline{\phi}, \overline{\phi}\:\right]$ where
$\Phi \supseteq \{\phi = f(\a,r,t) \mid \a \in A, r \in R, t \in T\}$.
If we can then prove that $\overline{\phi} < 2(1-c)\ln 4$, then we have a proof
that  $f(\a,r,t) < 2(1-c)\ln 4$ for all $\a \in A$, $r \in R$, and
$t \in T$, as required.

In Section~\ref{sec:interior} we prove that there are no global maxima in the interior of the
domain, except for the point $\a=\frac{1}{4}$, $r=\frac{1}{4}$, $t=\frac{1}{16}$, and in
Section~\ref{sec:boundary} we prove that there are no global maxima on the boundary of the
domain.  In both cases, the proof relies on an interval analysis program to rule out
the various cases.  The specifications of the program are
listed in Section~\ref{sec:code}.

\subsection{The Interior of the Domain}
\label{sec:interior}

The goal of this section is to rule out all possible stationary points, as defined by (\ref{eq:interior}),
with $z > y > 0$, from being global maxima.

\begin{lemma}
\label{lem:nointeriormax}
For all $\frac{2}{3} \leq c < 1$,
consider any stationary point of the function $f$ of Lemma~\ref{lem:fmax}, in
the domain $0 < \a < 1$, $\frac{2\a}{3c} < r < 1 - \frac{2(1-\a)}{3c}$, $\max\left\{0, \frac{3r-1}{2}\right\} < t < 4r$, and with $z > y$.
The maximum value that $f$ can take at this point is
smaller than $2(1-c)\ln 4 - 0.1$.
\end{lemma}

The proof of Lemma~\ref{lem:nointeriormax} relies on an interval analysis
program, and
the analysis will be on a slightly different function.  We define
\begin{eqnarray*}
F(c,\a,r,t,x,y,z) & = & \ln 4 -c \ln 4 + (1-\a)\ln 3 - c(2+t-3r)\ln 3 \notag \\
 & & \mbox{} +c(1-3r+2t)\ln 2 -\alpha \ln \alpha -(1-\alpha)\ln (1-\alpha) \notag \\
 & & \mbox{} -c(1-3r+2t) \ln (1-3r+2t) -c(3r-3t) \ln (r-t) \notag \\
 & & \mbox{} -ct \ln t + r3c \ln r + (1-r)3c \ln (1-r) \notag \\
 & & \mbox{} + \Upsilon(x) - \alpha \Upsilon(z) - (1-\alpha) \Upsilon(y)
\end{eqnarray*}
where
\begin{equation} \label{eq:upsilon}
 \Upsilon(w) = \frac{(\rme^w-1)w}{\rme^w-1-w}\ln w - \ln \left(\rme^w-1-w\right).
\end{equation}
Note that from (\ref{eq:kimplicitsMH}),
\begin{eqnarray*}
\Upsilon(x) &=& 3c \ln x - \ln (\rme^x-1-x) \\
(1-\a)\Upsilon(y) &=& (1-r)3c\ln y - (1-\a)\ln (\rme^y-1-y) \\
\a \Upsilon(z) &=& r3c \ln z - \a \ln (\rme^z-1-z) .
\end{eqnarray*}
$F$ is defined exactly the same as $f$ except that $x$, $y$, and $z$ are parameters of $F$ instead
of being defined by $\a$ and $r$, and $c$ is an explicit parameter of $F$.  Given intervals for each of
the parameters of $F$, we will determine an interval
\begin{eqnarray*}
\interval{\phi} &\supseteq& \{\phi \mid \phi = F(c, \a, r, t, x, y, z) \mbox{ for } c \in \interval{c}, \a \in \interval{\a},
r \in \interval{r}, \\
 & & \qquad \qquad \qquad t \in \interval{t}, x \in \interval{x}, y \in \interval{y}, \mbox{ and } z \in \interval{z}\}.
\end{eqnarray*}
Note that if $\interval{x}$, $\interval{y}$, and $\interval{z}$ contain all values of $x, y, z>0$
that satisfy (\ref{eq:kimplicitsMH}) for each $c \in \interval{c}$,
$\a \in \interval{\a}$ and $r \in \interval{r}$ then
\begin{eqnarray*}
\interval{\phi} &\supseteq& \{\phi \mid \phi = f(\a, r, t) \mbox{ for } \a \in \interval{\a},
r \in \interval{r} \mbox{, and } t \in \interval{t} \\
 & & \qquad \qquad \qquad \mbox{ and where } \underline{c} \leq c \leq \overline{c}\}.
\end{eqnarray*}  
{For a certain constant $b$, we 
verify for every such interval, we have  $ 2(1-\underline{c})\ln 4 - \overline{\phi} \geq b$. This proves that}
$f(\a,r,t) < 2(1-c)\ln 4 - b$ for all $\a \in \interval{\a}$,
$r \in \interval{r}$, $t \in \interval{t}$ and where $\underline{c} \leq c \leq \overline{c}$.
Our interval analysis program proves that such a $b$ exists and its value is between 0.010 and 0.055.

The interval analysis is as follows.  First, we will place an upper bound on the value of $y$ at any
stationary point with $z > y$ and with $\frac{2}{3} \leq c < 1$. We define:
\begin{definition}
$x^*$ is the largest solution to $\frac{x(\rme^x-1)}{\rm\rme^x-1-x} = 3$.
\end{definition}
From (\ref{eq:kimplicitsMH}), it is straightforward to show that if $z > y$, then $z > x > y$, and if $c < 1$ then $x$, and by extension $y$,
cannot exceed $x^*$.

While $y$ is defined
to be larger than 0,
we note that, from Observation~\ref{obs:z0}, 
the value of $f$ at the domain boundary $r = \frac{2\a}{3c}$ is the limit
of $f$ as $y$ approaches 0, and from Observation~\ref{obs:f}, the limit of $f$ at the border exists.  As a result,
we have to rule out the case that the value of $f$ exceeds $2(1-c)\ln 4$ at any point where $y$ is arbitrarily close to 0.
To do so, we will extend the interval $(0, x^*)$ of possible values for $y$ to $[0, x^*]$.
We can rewrite (\ref{eq:interior}) as
\begin{equation*}
2z = \frac{3(\rme^y-1)^2 + (\rme^z-1)^2}{\rme^y-1+\rme^z-1} \cdot \frac{y}{\rme^y-1},
\end{equation*}
and we note that as $y$ tends
to 0, $z$ tends to the largest solution of
\begin{equation}
\label{eq:yis0}
2z = \rme^z-1.
\end{equation}

We will cover the interval $[0, x^*]$ of possible values for $y$
with overlapping subintervals.  For each subinterval $Y = \interval{y}$, we will compute
intervals that contain all possible values that $z$, $x$, $c$, $\a$, $r$, and $t$
can take at an stationary point of
$f$ with $y \in Y$ and $y < z$.
We will then compute the maximum value that $F$ can take on these
intervals and verify that this value is smaller than $2(1-c)\ln 4$.

Given $Y = \interval{y}$, we first compute an interval that contains all $z$ such that:
$Y$ contains at least one $y<z$ where
$(z,y)$ satisfies (\ref{eq:interior}) of Lemma~\ref{lem:interior}.
Let the function $z(y)$ be defined as the largest solution to (\ref{eq:interior}) when $y > 0$ and
as the largest solution to (\ref{eq:yis0}) when $y=0$.
Given the interval $Y = \interval{y}$, we find an interval $Z = \interval{z}$ such that
\begin{equation}\label{eq:zint} \interval{z} \supseteq \left\{ z \mid z = z(y) \mbox{ for } y \in \interval{y}\right\}. \end{equation}
Because $z$ is only defined implicitly in (\ref{eq:interior}) and (\ref{eq:yis0}),
we do the following to compute $\interval{z}$.

Let $z_1 = z\left(\underline{y}\right)$, and let
$z_2 = z\left(\overline{y}\right)$.
Using the bisection method, we find intervals $Z_1=\interval{z_1}$ and $Z_2=\interval{z_2}$ that contain $z_1$ and $z_2$, respectively.
We do not use $\left[\:\underline{z_1}, \overline{z_2}\:\right]$ as the interval
$Z$ because to do so, we would need to know that $\frac{dz}{dy} > 0$.  A proof of this
inequality appears to be challenging, so instead we will find a $\d$ and $\d'$ such that
$Z=\left[\:\underline{z_1}-\d', \overline{z_2}+\d\:\right]$ satisfies (\ref{eq:zint}).

From  (\ref{eq:interior}),
Let
\[\l(y,z) = y \left(
\frac{3\left(\rme^y-1\right)^2 + \left(\rme^z-1\right)^2}
                   {2\left(\rme^y-1\right)^2+2\left(\rme^y-1\right)\left(\rme^z-1\right)}
\right),
\]
and so for $y>0$, $z(y)$ is the largest solution to $z=\l(y,z)$.
Let \[\l_I(Y,Z) = \{ \l(y,z) \mid y \in Y, z \in Z\}.\]
From our computation of $Z_1$ and $Z_2$, we know that
there is a $z \in \left[\:\underline{z_1}, \overline{z_2}\:\right]$ such that $z = \l(y,z)$ for
some $y \in Y$ (eg. $z=\underline{z_1}$ and $z=\overline{z_2}$).
This implies $z \in \l_I\left(Y, \left[\:\underline{z_1}, \overline{z_2}\:\right]\right)$, and
\[\left[\:\underline{z_1}, \overline{z_2}\:\right] \;\cap\;
\lambda_I\left(Y, \left[\:\underline{z_1}, \overline{z_2}\:\right]\right) \neq \emptyset.\]
To compute bounds for $Z$, we will find $\d,\d',\e>0$ such that
\begin{eqnarray*}
\left[\:\underline{z_1} - \d'-\e, \underline{z_1}-\d'\:\right] \;\cap\;
\lambda_I\left(Y, \left[\:\underline{z_1}-\d'-\e, \underline{z_1}-\d'\:\right]\right) &=& \emptyset\\
\left[\:\overline{z_2} + \d, \overline{z_2}+\d +\e\:\right] \;\cap\;
\lambda_I\left(Y, \left[\:\overline{z_2}+\d, \overline{z_2}+\d +\e\:\right]\right) &=& \emptyset.
\end{eqnarray*}

We will prove that for every $y\in Y$, $\overline{z_2}+\d$ is an upper bound on the largest $z$ such that $z = \lambda(y,z)$.
Suppose that for some $y'\in Y$, there is a $z' > \overline{z_2}+\d$ such that $z' = \lambda(y',z')$.
From the definition of $\l$, we have $z(y') = z'$. Because $z(\cdot)$ is continuous, the intermediate value
theorem states that for every $z_1<z^*< z'$ there is some $y^*\in\left[\:\underline{y}, y'\:\right] \subset Y$
with $z(y^*)=z^*$, and so $z^*=\l(y^*,z^*)$.  Picking $z_2+\d< z^*<\min(z_2+\d+\e,z')$ yields a
contradiction to the statement
\[\left[\:\overline{z_2} + \d, \overline{z_2}+\d +\e\:\right] \;\cap\;
\lambda_I\left(Y, \left[\:\overline{z_2}+\d, \overline{z_2}+\d +\e\:\right]\right) = \emptyset. \]
A similar argument shows that for every $y\in Y$, $\underline{z_1} - \d'$ is a lower bound on
the smallest $z$ such that $z = \lambda(y,z)$.

Therefore, we can set $Z=\left[\:\underline{z_1}-\d', \overline{z_2}+\d\:\right]$.
To find appropriate $\d$ and $\d'$ values, we set
$\e = 2.22 \times 10^{-16}$,
and we use binary search to find small
$\d$ and $\d'$ that satisfy the above properties.

Next, we combine
(\ref{eq:kimplicitsMH}) with the equations
(\ref{eq:maxk-alphaMH}) and (\ref{eq:maxk-rMH}) that must hold at any stationary point of
$f$ to define the following function,
\begin{equation*}
c(y,z) = \frac{3(\rme^y-1)y + (\rme^z-1)z}{3\times(3(\rme^y-1-y) + \rme^z-1-z)},
\end{equation*}
and given the intervals $\interval{y}$ and $\interval{z}$, we use standard interval operations to compute
\[ \interval{c} \supseteq \{c \mid c = c(y, z) \mbox{ for } y \in \interval{y} \mbox{ and } z \in \interval{z} \}. \]

From (\ref{eq:kimplicitsMH}), we define the function $x(c)$ as the largest solution to
\begin{equation}
3c = \frac{x(\rme^x-1)}{\rme^x-1-x}. \label{eq:x}
\end{equation}
and we use $x(c)$ to find an interval $\interval{x}$ given the interval $\interval{c}$.  Let
\[\interval{x} \supseteq \{x \mid x = x(c) \mbox{ for } c \in \interval{c}\}. \]
Because $x$ is only defined implicitly in (\ref{eq:x}), we use the following procedure to compute $X=\interval{x}$.

It is straightforward to verify the $\frac{x(\rme^x-1)}{\rme^x-1-x}$ is an increasing function when $x > 0$.  As a result, (\ref{eq:x}) will have at most one positive solution.  Given $\interval{c}$, we use Lemma~\ref{lem:xvalue} below to give positive upper and lower bounds for both $x(\underline{c})$ and $x(\overline{c})$, and
then we can use binary search, starting from these bounds, to find the intervals
$X_1=\interval{x_1}$ and $X_2=\interval{x_2}$ that contain $x(\underline{c})$ and $x(\overline{c})$, respectively.
From (\ref{eq:x}), it is straightforward to see that
$\frac{dx}{dc} > 0$, and so we can set $X = \left[\:\underline{x_1}, \overline{x_2}\:\right]$.

\begin{lemma}
\label{lem:xvalue}
Let $3c = \frac{x(\rme^x-1)}{\rme^x-1-x}$ and $c \geq \frac{2}{3}$.  If $x > 0$ then
$3c - 2 < x < 3c$.
\end{lemma}

\begin{proof}
First note that $\lim_{x \rightarrow 0} \frac{x(\rme^x-1)}{\rme^x-1-x} = 2$.  So when $x$ approaches 0,
$c$ approaches $\frac{2}{3}$, and we have, in the limit, $x + 2 = 3c$.
Next note that $x < x \left(\frac{\rme^x-1}{\rme^x-1-x}\right) = 3c$.
To complete the proof, we show that if $x > 0$, the derivative of $\frac{x(\rme^x-1)}{\rme^x-1-x}$ is
less than 1.  This implies that for $x > 0$, $3c$ is sandwiched between $x$ and $x+2$.

The derivative of $\frac{x(\rme^x-1)}{\rme^x-1-x}$ is $\frac{(\rme^x-1)^2 - x^2 \rme^x}{(\rme^x-1-x)^2}$.
To show that this derivative is less than 1, we show
\begin{eqnarray*}
(\rme^x-1-x)^2 - (\rme^x-1)^2 + x^2 \rme^x & = & - 2x(\rme^x-1) + x^2 + x^2\rme^x \\
                                  & = & x\left( - 2(\rme^x-1) + x + x \rme^x\right) \\
                                  & \geq & 0.
\end{eqnarray*}
The last inequality follows by noting that the first and second derivatives of $-2(\rme^x-1) + x + x \rme^x$ are
$x \rme^x - \rme^x + 1$ and $x \rme^x$, respectively.
\end{proof}

It is now straightforward to use equations that must hold at any stationary point of $f$ to define intervals for the remaining parameters of $F$
given the intervals $\interval{y}$ and $\interval{z}$.
We use (\ref{eq:maxk-alphaMH}) to define the function
\begin{equation*}
\a(y,z) = \frac{\rme^z-1 - z}{3(\rme^y-1-y) + \rme^z-1-z},
\end{equation*}
and given the intervals $\interval{y}$ and $\interval{z}$, we use standard interval operations to compute the interval
\[ \interval{\a} \supseteq \{\a \mid \a = \a(y, z) \mbox{ for } y \in \interval{y} \mbox{ and } z \in \interval{z} \}. \]
We use (\ref{eq:rval}) to define the function
\begin{equation*}
r(y,z) = \frac{(\rme^z-1)z}{3(\rme^y-1)y + (\rme^z-1)z},
\end{equation*}
and given the intervals $\interval{y}$ and $\interval{z}$, we use standard interval operations to compute the interval
\[ \interval{r} \supseteq \{r \mid r = r(y, z) \mbox{ for } y \in \interval{y} \mbox{ and } z \in \interval{z} \}. \]
We use (\ref{eq:start}) to define the function
\begin{equation*}
t(r,y,z) = \frac{(\rme^z-1)^2r}{3(\rme^y-1)^2 + (\rme^z-1)^2},
\end{equation*}
and given the intervals $\interval{r}$, $\interval{y}$ and $\interval{z}$, we use standard interval operations to compute the interval
\[ \interval{t} \supseteq \{t \mid t = t(r, y, z) \mbox{ for } r \in \interval{r},  y \in \interval{y} \mbox{ and } z \in \interval{z} \}. \]

Finally, we use standard interval operations to
compute the interval
\begin{eqnarray*}
\interval{\phi} &=& \{\phi \mid \phi = F(c, \a, r, t, x, y, z) \mbox{ for } c \in \interval{c}, \a \in \interval{\a},
r \in \interval{r}, \\
 & & \qquad \qquad \qquad t \in \interval{t}, x \in \interval{x}, y \in \interval{y}, \mbox{ and } z \in \interval{z}\}.
\end{eqnarray*}
We then verify that $\overline{\phi} < 2(1-\overline{c})\ln 4$, and this verifies that at every
point in the given intervals the value of $F$, and therefore $f$, is smaller than $2(1-c)\ln 4$.
Repeating this process for every overlapping subinterval $\interval{y}\subset [0, x^*]$, the program proves that
there exists a constant $b$ with $0.010 < b < 0.055$ such that
that the value for $f$ on each interval is smaller than $2(1-c)\ln 4 - b$, and this
proves Lemma~\ref{lem:nointeriormax}.

These computations are all performed by an interval analysis program specified in Section~\ref{sec:code}.
The computation of the interval $\interval{\phi}$ uses standard interval analysis techniques plus the following
observations.

\begin{observation}
\label{obs:xbehavior}
The function $\Upsilon(w)$
has the following form.  When $w \rightarrow 0$, the function approaches $\ln 2$, the
function has one minimum at $w=1$ with value $-\ln(e-2)$, and the function grows unbounded
as $w \rightarrow \infty$.
\end{observation}

In the cases that $z$ becomes large, we can use the following approximations for $\Upsilon(z)$.
These approximations have reduced dependency errors and will, for large $z$,
result in a tighter bound than a computation of $\Upsilon(z)$.
\begin{observation}
\label{obs:upsilon_lowerbound}
\[\Upsilon(w) =  \frac{(\rme^w-1)w}{\rme^w-1-w} \ln w - \ln (e^w-1-w) \geq w(\ln w - 1). \]
\end{observation}

\begin{observation}
\label{obs:upsilon_upperbound}
For $w > 2.5$,
\[\Upsilon(w) = \frac{(\rme^w-1)w}{\rme^w-1-w} \ln w - \ln (\rme^w-1-w) \leq w(\ln w - 1) + 1. \]
\end{observation}

\subsection{The Boundary of the Domain}
\label{sec:boundary}

The final step required for the proof of Lemma~\ref{lem:fmax} is to prove that at every
point on the boundary of the domain, $f$ approaches a limit that is less than $2(1-c)\ln 4$.

\begin{lemma}
\label{lem:noboundarymax}
Let $.67 \leq c < 1$ and let $b=b(c)$ be a positive constant that depends on $c$ but is independent of $\a$, $r$, and $t$.  
For every point $(\a_1,r_1,t_1)$ on the boundary of the
domain of $f$, $\lim_{(\a,r,t)\rightarrow (\a_1,r_1,t_1)}f(\a,r,t)$ exists and is
less than $2(1-c)\ln 4 - b$.
\end{lemma}

The proof of Lemma~\ref{lem:noboundarymax} is the subject of the remainder of this appendix.
From Observation~\ref{obs:f}, the limit exists.  The boundary of $f$ consists of 7 faces,
the faces intersect at 13 line segments, and the line segments intersect at 8 points.

For each of the 8 points, we take the limit of $f$ as its parameters approach that point and prove
that the limit is smaller than $2(1-c)\ln 4$.  For each of the line segments, we determine
a function on the line segment which is either the limit of $f$ as its parameters approach
each point of that boundary or an upper bound on the limit. For the present discussion, we
use $f_B$ to denote that function. $f_B$
will not have a discontinuity, and so any maximum value of
$f_B$ will occur either at the endpoints of the line segment or at a point where the derivative
of $f_B$ is 0.  As we have already proven that the endpoints are not global maxima,
we can restrict the analysis to the places where the derivative is 0.  Similarly, for
each of the 7 faces, we define a similar function $f_B$ on the face and evaluate the points where both the
partial first derivatives of $f_B$ are 0.

The proof of Lemma~\ref{lem:noboundarymax} will also rely on an interval analysis program.
For each boundary, the interval analysis will be
performed on the function $F_B(c, \a, r, t, x, y, z)$ where $F_B$ is defined the same as $f_B$ except that
$x$, $y$, and $z$ are parameters to $F_B$ instead of being defined in terms of $\a$ and $r$, and $c$ is
an explicit parameter of $F_B$.
The interval analysis program will verify that for each
$c \in \left[.67, 1\right)$,
for each $\a$, $r$, and $t$ in the domain of $f_B$, and
for each $x$, $y$, and $z$ such that $x,y,z > 0$ and (\ref{eq:kimplicitsMH}) is satisfied,
the maximum value for $F_B(c, \a, r, t, x, y, z)$
is smaller than $2(1-c)\ln 4$, and the program will report the smallest difference between 
$2(1-c)\ln 4$ and the maximum value for an end point or stationary point on the boundary.
Note that when $r = \frac{2\a}{3c}$, we replace any term of $F_B$ involving $z$ with its limit as $z$ tends
to 0.  Likewise, when $1-r = \frac{2(1-\a)}{3c}$, we replace any term involving $y$ with its limit as $y$ tends
to 0. 
From Observation~\ref{obs:z0} and Observation~\ref{obs:f}, the limit for $f$ exists at those points.
(While we can evaluate $F_B$ when $c = \frac{2}{3}$ by replacing any term 
involving $x$ with its limit as $x$ tends to 0, this approach will run into problem with Cases 24 and 25
below.  Specifically, when $c = \frac{2}{3}$, then the entire domain including the global maximum lies
on the boundary $r = \frac{2\a}{3c} = 1 - \frac{2(1-\a)}{3c}$.)

The program will cover the interval $\left[.67, 1\right]$ with
overlapping subintervals.  Likewise, for each parameter, $\a$, $r$, and $t$, not restricted by
the boundary to a single value, the program will cover the legal values of each parameter with
overlapping subintervals.
Given the intervals $\interval{c}$, $\interval{\a}$, and $\interval{r}$,
we compute intervals containing the possible values for $x$, $y$, and $z$, if needed.
As is done in Section~\ref{sec:interior}, we define the function
$x(c)$ as the largest solution to
\begin{equation*}
3c = \frac{x(\rme^x-1)}{\rme^x-1-x}, \label{eq:x2}
\end{equation*}
we define the function $y(c, \a, r)$ as the largest solution to
\begin{equation*}
\frac{3c(1-r)}{1-\a} = \frac{y(\rme^y-1)}{\rme^y-1-y}, \label{eq:y}
\end{equation*}
and we define the function $z(c, \a, r)$ as the largest solution to
\begin{equation*}
\frac{3cr}{\a} = \frac{z(\rme^z-1)}{\rme^z-1-z}. \label{eq:z}
\end{equation*}
Let
\begin{eqnarray*}
\interval{x} & \supseteq & \{x \mid x = x(c) \mbox{ for } c \in \interval{c}\}, \\
\interval{y} & \supseteq & \{y \mid y = y(c, \a, r) \mbox{ for } c \in \interval{c}, \a \in \interval{\a}, \mbox{ and } r \in \interval{r} \}, \mbox{ and } \\
\interval{z} & \supseteq & \{z \mid z = z(c, \a, r) \mbox{ for } c \in \interval{c}, \a \in \interval{\a}, \mbox{ and } r \in \interval{r} \}.
\end{eqnarray*}
Because $x$, $y$, and $z$ are defined implicitly, we use
binary search to separately find upper and lower bounds for the intervals.
The technique used to compute $\interval{x}$ is described in Section~\ref{sec:interior},
and the same technique is used to compute $\interval{y}$ and $\interval{z}$.
Lemma~\ref{lem:xvalue} gives the initial bounds for $\underline{x}$ and $\overline{x}$,
and from the same reasoning
we get the following lemmas that give the initial bounds for $\underline{y}$, $\overline{y}$,
$\underline{z}$ and $\overline{z}$.

\begin{lemma}
Let $\frac{3c(1-r)}{1-\a} = \frac{y(\rme^y-1)}{\rme^y-1-y}$, $c \geq \frac{2}{3}$, and $\a < 1$.
If $y > 0$ then $\frac{3c(1-r)}{1-\a} - 2 < y < \frac{3c(1-r)}{1-\a}$.
\end{lemma}

\begin{lemma}
Let $\frac{3cr}{\a} = \frac{z(\rme^z-1)}{\rme^z-1-z}$, $c \geq \frac{2}{3}$, and $\a > 0$.
If $z > 0$ then $\frac{3cr}{\a} - 2 < z < \frac{3cr}{\a}$.
\end{lemma}

Given $\interval{c}$, $\interval{\a}$, $\interval{r}$, $\interval{t}$, $\interval{x}$, $\interval{y}$,
and $\interval{z}$, if this is a subinterval of a line or face boundary, the program will compute
compute all possible values for the partial first derivatives on these intervals and verify that these
possible values contain 0.
The program will then compute an interval that contains all possible values for $F_B$ on these intervals, 
it will verify that the upper bound on this interval is smaller than $2(1-\overline{c})\ln 4$, and it will
record the difference between the maximum possible value for $F_B$ and $2(1-\overline{c})\ln 4$.
As a result, we know that
$f(\a, r, t) < 2(1-c)\ln 4-b$ for all $c$, $\a$, $r$, and $t$ in these intervals.
Repeating this process for all subintervals of $\left[\frac{2}{3}, 1\right]$ completes the proof
of Lemma~\ref{lem:noboundarymax}.
The specifics of the program are listed in
Section~\ref{sec:code}.

The rest of this section gives the specific computations of $f_B$ for each boundary.

\subsubsection{The Boundary of the Domain of $f$.}

The domain for $f$ is the region bounded by $\a \in (0,1)$, $r \in \left(\frac{2\a}{3c}, 1-\frac{2(1-\a)}{3c}\right)$, and $t \in \left( \min \left\{ 0, \frac{3r-1}{2}\right\}, r\right)$.
The boundary of the domain for $f$ has 7 faces:
\[\begin{array}{llll}
\a = 0; &
\a = 1; &
r = \frac{2\a}{3c}; &
r = 1 - \frac{2(1-\a)}{3c}; \\
t = 0; &
t = \frac{3r-1}{2}; &
t = r,
\end{array}\]
the faces intersect at 13 line segments:
\[\begin{array}{lll}
\a = 0, t = 0; &
\a = 0, t = r; &
\a = 0, r = 1 - \frac{2(1-\a)}{3c}; \\
\a = 1, t = \frac{3r-1}{2}; &
\a = 1, t = r; &
\a = 1, r = \frac{2\a}{3c}; \\
r = \frac{2\a}{3c}, t = r; &
r = \frac{2\a}{3c}, t = \frac{3r-1}{2}; &
r = \frac{2\a}{3c}, t = 0; \\
r = 1-\frac{2(1-\a)}{3c}, t = r; &
r = 1-\frac{2(1-\a)}{3c}, t = \frac{3r-1}{2}; &
r = 1-\frac{2(1-\a)}{3c}, t = 0; \\
t = 0 = \frac{3r-1}{2},
\end{array}\]
and the line segments intersect at 8 points:
\[\begin{array}{ll}
\a = 0, r = 0, t = 0; &
\a = 0, r = 1-\frac{2}{3c}, t = 0; \\
\a = 0, r = 1-\frac{2}{3c}, t = 1-\frac{2}{3c}; &
\a = 1, r = 1, t = 1; \\
\a = 1, r = \frac{2}{3c}, t = \frac{2-c}{2c}; &
\a = 1, r = \frac{2}{3c}, t = \frac{2}{3c}; \\
\a = \frac{c}{2}, r = \frac{1}{3}, t = 0; &
\a = 1-c, r = \frac{1}{3}, t = 0 .
\end{array}\]

As we did in Section~\ref{sec:interior}, we will define $\Upsilon$ with (\ref{eq:upsilon}) and use it
to rewrite $f$ as
\begin{eqnarray*}
f(\a,r,t) & = & \ln 4 -c \ln 4 + (1-\a)\ln 3 - c(2+t-3r)\ln 3 \notag \\
 & & \mbox{} +c(1-3r+2t)\ln 2 -\alpha \ln \alpha -(1-\alpha)\ln (1-\alpha) \notag \\
 & & \mbox{} -c(1-3r+2t) \ln (1-3r+2t) -c(3r-3t) \ln (r-t) \notag \\
 & & \mbox{} -ct \ln t + r3c \ln r + (1-r)3c \ln (1-r) \notag \\
 & & \mbox{} + \Upsilon(x) - \alpha \Upsilon(z) - (1-\alpha) \Upsilon(y)
\end{eqnarray*}
where
$x,y,z > 0$ and (\ref{eq:kimplicitsMH}) holds.
The following additional observations are used to both compute the limit of $f$ as its parameters approach
a boundary and to give an upper bound on that limit.

From (\ref{eq:kimplicitsMH}),
the maximum value $x$ can have in $f$ is the largest solution to
\[3 = \frac{x (\rme^x-1)}{\rme^x-1-x} \]
or $x = 2.1491\ldots$.  As a result, the largest value
of $\Upsilon(x)$ with $x > 1$ is $.60355\ldots$.  The following observation follows from this fact plus
Observation~\ref{obs:xbehavior}.

\begin{observation} \label{obs:xmax}
\[ \Upsilon(x) \leq \ln 2. \]
\end{observation}

As discussed in Lemma~\ref{lem:xvalue}, we have
$\lim_{x \to 0}\frac{x(\rme^x-1)}{\rme^x-1-x} = 2$, and $\frac{x(\rme^x-1)}{\rme^x-1-x}$ is an increasing
function of $x$.  We also have
$\lim_{x \to \infty}\frac{x(\rme^x-1)}{\rme^x-1-x} = \infty$.
As a result, from (\ref{eq:kimplicitsMH}), we have that as
$\a \to 0$, $z \to \infty$, as $\a \to 1$, $y \to \infty$, as $r \to \frac{2\a}{3c}$, $z \to 0$,
and as $r \to 1 - \frac{2(1-\a)}{3c}$, $y \to 0$.  From these facts and Observation~\ref{obs:xbehavior},
we have the following observation.

\begin{observation}
\label{obs:limits}
\begin{eqnarray}
\lim_{\a \to 0}\a\Upsilon(z) & \geq & 0 \label{lim-a-min} \\
\lim_{r \to \frac{2\a}{3c}} \a\Upsilon(z) & = & \a\ln 2 \label{lim-r-min} \\
\lim_{\a \to 1, r \to 1} \a\Upsilon(z) & = & \Upsilon(x) \label{lim-a-r-max} \\
\lim_{\a \to 1}(1-\a)\Upsilon(y) & \geq & 0 \label{lim-a-max} \\
\lim_{r \to 1-\frac{2(1-\a)}{3c}} (1-\a)\Upsilon(y) & = & (1-\a)\ln 2 \label{lim-r-max} \\
\lim_{\a \to 0, r \to 0} (1-\a)\Upsilon(y) & = & \Upsilon(x) \label{lim-a-r-min}
\end{eqnarray}
\end{observation}

\subsubsection{The Boundary Cases That Are Points}

In each case, we take the limit of $f$ as its parameters approach the boundary point, and we prove that
this limit is smaller than $2(1-c)\ln 4$.  For the first two cases, the proof is straightforward.
For cases 3--6, we bound this limit by a function of a single variable $c$.
These cases can be verified without interval analysis. But because we are using an interval analysis program
for the later cases, we will shorten the presentation here by verifying these cases with the program as well.
For cases 7 and 8, the bound on the limit of $f$ includes limits of $\Upsilon(x)$, $\Upsilon(y)$ and $\Upsilon(z)$ not
covered in Observation~\ref{obs:limits}.  In these cases, the interval analysis program will compute
bounds on these limits of $\Upsilon$ and use them to verify that the limit of $f$ is always smaller than $2(1-c)\ln 4$.

\subsubsection*{\bf Case 1: the boundary with $\a = 0$, $r = 0$, $t = 0$.}

Applying (\ref{lim-a-min}) and (\ref{lim-a-r-min}), we get the following.
\begin{eqnarray*}
\lim_{\substack{\a \to 0, r \to 0 \\ t \to 0}} f(\a, r, t) & = &
\ln 4 -c \ln 4 + \ln 3 - 2c\ln 3 + c\ln 2 \\
 & & \mbox{} + \lim_{\substack{\a \to 0, r \to 0}} \left(\Upsilon(x) - \a \Upsilon(z) - (1-\a) \Upsilon(y)\right) \\
     & \leq & \ln 12 - c \ln 18 \\
     & < & 2(1-c) \ln 4 - \ln \frac{4}{3}
\end{eqnarray*}
for $c \in [2/3, 1)$.  

\subsubsection*{\bf Case 2: the boundary with $\a = 1$, $r = 1$, $t = 1$.}

Applying (\ref{lim-a-max}) and (\ref{lim-a-r-max}), we get the following.
\begin{eqnarray*}
\lim_{\substack{\a \to 1, r \to 1 \\ t \to 1}} f(\a, r, t) & = &
\ln 4 -c \ln 4
    + \lim_{\substack{\a \to 1, r \to 1}} \left(\Upsilon(x) - \a \Upsilon(z) - (1-\a) \Upsilon(y)\right) \\
 & \leq & \ln 4 - c \ln 4 \\
 & < & 2 (1-c) \ln 4 - b
\end{eqnarray*}
for $c \in [2/3, 1)$ and where $b= b(c) > 0$.

\subsubsection*{\bf Case 3: the boundary with $\a = 0$, $r = 1 - \frac{2}{3c}$, $t = 0$.}

Applying (\ref{lim-a-min}), (\ref{lim-r-max}), and Observation~\ref{obs:xmax}, we get
the following.
\begin{eqnarray}
\lim_{\substack{\a \to 0, r \to 1-\frac{2}{3c} \\ t \to 0}} f(\a, r, t)
& = & \ln 4 -c \ln 4 + \ln 3 - (-c+2)\ln 3 +(-2c+2)\ln 2 \notag \\
 & & \mbox{} -(-2c+2) \ln \left(-2+\frac{2}{c}\right) -(3c-2) \ln \left(1-\frac{2}{3c}\right) \notag \\
 & & \mbox{} + \left(3c-2\right) \ln \left(1-\frac{2}{3c}\right) + 2 \ln \left(\frac{2}{3c}\right) \notag \\
 & & \mbox{}
    + \lim_{\substack{\a \to 0 \\ r \to 1-\frac{2}{3c}}} \left(\Upsilon(x) - \a \Upsilon(z) - (1-\a) \Upsilon(y)\right) \notag \\
 & \leq & \ln \frac{16}{27} - c \ln \frac{4}{3} -2c \ln c -2(1-c)\ln(1-c) . \label{eq:case3}
\end{eqnarray}
The interval analysis program verifies that (\ref{eq:case3}) is no larger than $2(1-c)\ln 4 - 0.005$,  for all
$c \in [.67, 1)$.

\subsubsection*{Case 4: the boundary with $\a = 0$, $r = 1 - \frac{2}{3c}$, $t = 1-\frac{2}{3c}$.}
Applying (\ref{lim-a-min}), (\ref{lim-r-max}), and Observation~\ref{obs:xmax}, we get
the following.
\begin{eqnarray}
\lim_{\substack{\a \to 0, r \to 1-\frac{2}{3c} \\ t \to 1-\frac{2}{3c}}} f(\a, r, t) & = &
\ln 4 -c \ln 4 + \ln 3 - \frac{4}{3}\ln 3 +\frac{2}{3}\ln 2 \mbox{} -\frac{2}{3} \ln \frac{2}{3c} \notag \\
& & \mbox{} -\left(c-\frac{2}{3}\right) \ln \left(1-\frac{2}{3c}\right) + (3c-2) \ln \left(1-\frac{2}{3c}\right) \notag \\
 & & \mbox{}
 + 2 \ln \frac{2}{3c}
    + \lim_{\substack{\a \to 0 \\ r \to 1-\frac{2}{3c}}} \left(\Upsilon(x) - \a \Upsilon(z) - (1-\a) \Upsilon(y)\right) \notag \\
 & < & \ln 16 - \frac{1}{3}\ln 3 - c \ln 36 + \frac{2}{3}(3c-2)\ln(3c-2) \notag \\
 & & \mbox{} -2c\ln c . \label{eq:case4}
\end{eqnarray}
The interval analysis program verifies that (\ref{eq:case4}) is no larger than $2(1-c)\ln 4 - 0.046$, for all
$c \in [.67, 1)$.

\subsubsection*{Case 5: the boundary with $\a = 1$, $r = \frac{2}{3c}$, $t = \frac{2}{3c}$.}
Applying (\ref{lim-a-max}), (\ref{lim-r-min}), and Observation~\ref{obs:xmax}, we get
the following.
\begin{eqnarray}
\lim_{\substack{\a \to 1, r \to \frac{2}{3c} \\ t \to \frac{2}{3c}}} f(\a, r, t) & = &
\ln 4 -c \ln 4 - \left(2c-\frac{4}{3}\right)\ln 3 +\left(c-\frac{2}{3}\right)\ln 2 \notag \\
 & & \mbox{} -\left(c-\frac{2}{3}\right) \ln \left(1-\frac{2}{3c}\right) -\frac{2}{3} \ln \frac{2}{3c} + \frac{2}{3} \ln \frac{2}{3} \notag \\
 & & \mbox{} + (3c-2) \ln \left(1-\frac{2}{3c}\right) \notag \\
 & & \mbox{}
    + \lim_{\substack{\a \to 1 \\ r \to \frac{2}{3c}}} \left(\Upsilon(x) - \a \Upsilon(z) - (1-\a) \Upsilon(y)\right) \notag \\
 & \leq & \frac{4}{3}\ln 12 - c \ln 162 + \frac{2}{3}(3c-2)\ln(3c-2) \notag \\
 & & \mbox{} - 2c\ln c . \label{eq:case5}
\end{eqnarray}
The interval analysis program verifies that (\ref{eq:case5}) is no larger than $2(1-c)\ln 4 -  0.147$,  for all
$c \in [.67, 1)$.

\subsubsection*{Case 6: the boundary with $\a = 1$, $r = \frac{2}{3c}$, $t = \frac{2-c}{2c}$.}
Applying (\ref{lim-a-max}), (\ref{lim-r-min}), and Observation~\ref{obs:xmax}, we get
\begin{eqnarray}
\lim_{\substack{\a \to 1, r \to \frac{2}{3c} \\ t \to \frac{2-c}{2c}}} f(\a, r, t) & = &
\ln 4 -c \ln 4 - \left(2c+\frac{2-c}{2}-2\right)\ln 3 \notag \\
& & \mbox{} +(c-2+2-c)\ln 2 \notag \\
& & \mbox{} -(c-2+2-c) \ln \left(1-\frac{2}{c}+\frac{2-c}{c}\right) \notag \\
& & \mbox{} -(2-3(2-c)) \ln \left(\frac{2}{3c}-\frac{2-c}{2c}\right) \notag \\
& & \mbox{} -\frac{2-c}{2} \ln \frac{2-c}{2c} + 2 \ln \frac{2}{3c} + (3c-2) \ln \left(1-\frac{2}{3c}\right) \notag \\
 & & \mbox{}
    + \lim_{\substack{\a \to 1 \\ r \to \frac{2}{3c}}} \left(\Upsilon(x) - \a \Upsilon(z) - (1-\a) \Upsilon(y)\right) \notag \\
 & < & \ln 16 - c \ln 54 + \frac{1}{2}\left(3 c-2\right) \ln(3c-2) \notag \\
   & & \mbox{} - \frac{1}{2}\left(2-c\right) \ln (2-c) - 2c \ln c . \label{eq:case6}
\end{eqnarray}
The interval analysis program verifies that (\ref{eq:case6}) is no larger than $2(1-c)\ln 4 -  0.071$,  for all
$c \in [.67, 1)$.

\subsubsection*{Case 7: the boundary with $\a = \frac{c}{2}$, $r = \frac{1}{3}$, $t = 0$.}
\begin{eqnarray}
\lim_{\substack{\a \to \frac{c}{2}, r \to \frac{1}{3} \\ t \to 0}} f(\a, r, t) & = &
   \ln 4 -c \ln 4 + \left(1-\frac{c}{2}\right)\ln 3 - c\ln 3 \notag \\
& & \mbox{} -\frac{c}{2} \ln \frac{c}{2} -\left(1-\frac{c}{2}\right)\ln \left(1-\frac{c}{2}\right) \notag \\
& & \mbox{} -c \ln \frac{1}{3} + c \ln \frac{1}{3} + 2c \ln \frac{2}{3} \notag \\
 & & \mbox{}
    + \lim_{\substack{\a \to \frac{c}{2}, r \to \frac{1}{3}}} \left(\Upsilon(x) - \a \Upsilon(z) - (1-\a) \Upsilon(y)\right) \notag \\
 & = &
  \ln 24 - c \ln 27 - \frac{c}{2} \ln 3 - \frac{1}{2}(2-c)\ln (2-c) - \frac{c}{2} \ln c \notag \\
 & & \mbox{}
    + \lim_{\substack{\a \to \frac{c}{2}, r \to \frac{1}{3}}} \left(\Upsilon(x) - \a \Upsilon(z) - (1-\a) \Upsilon(y)\right) . \label{eq:case7}
\end{eqnarray}
The interval analysis program verifies that (\ref{eq:case7}) is no larger than $2(1-c)\ln 4 -  0.265$,  for all
$c \in [.67, 1)$.

\subsubsection*{Case 8: the boundary with $\a = 1-c$, $r = \frac{1}{3}$, $t = 0$.}
\begin{eqnarray}
\lim_{\substack{\a \to 1-c, r \to \frac{1}{3} \\ t \to 0}} f(\a, r, t) & = &
 \ln 4 -c \ln 4 + c\ln 3 - c\ln 3 - (1-c) \ln (1-c) \notag \\
& & \mbox{} - c\ln c -c \ln \frac{1}{3} + c \ln \frac{1}{3} + 2c \ln \frac{2}{3} \notag  \\
 & & \mbox{}
    + \lim_{\substack{\a \to 1-c, r \to \frac{1}{3}}} \left(\Upsilon(x) - \a \Upsilon(z) - (1-\a) \Upsilon(y)\right) \notag \\
 & = &
  \ln 4 - c \ln 9 - (1-c)\ln (1-c) - c \ln c \notag \\
 & & \mbox{}
    + \lim_{\substack{\a \to 1-c \\ r \to \frac{1}{3}}} \left(\Upsilon(x) - \a \Upsilon(z) - (1-\a) \Upsilon(y)\right) . \label{eq:case8}
\end{eqnarray}
The interval analysis program verifies that (\ref{eq:case8}) is no larger than $2(1-c)\ln 4 - 0.292$,  for all
$c \in [.67, 1)$.

\subsubsection{The Boundary Cases That Are Line Segments}

In each case, we rule out the existence of a point on the boundary
where $f$ exceeds \mbox{$2(1-c)\ln 4$}.
We will compute a function that is either the limit of $f$ as its parameters approach the boundary or
an upper bound on this limit.  This new function will be continuous, and any
maximum of this function will occur either where its first derivative is 0 or at an
endpoint of the line segment.
Cases 1--8 above evaluate the endpoints, and the analysis that follows will focus on the points where the
first derivative is 0.

\subsubsection*{Case 9: the boundary with $\a = 0$, and $t = 0$.}
Using (\ref{lim-a-min}), we get
\begin{eqnarray}
\lim_{\substack{\a \to 0 \\ t \to 0}} f(\a, r, t) & = &
 \ln 4 -c \ln 4 + \ln 3 - c(2-3r)\ln 3 \mbox{} +c(1-3r)\ln 2 \notag \\
& & \mbox{} -c(1-3r) \ln (1-3r) -c3r \ln r + r3c \ln r \notag \\
 & & \mbox{} + (1-r)3c \ln (1-r)
    + \lim_{\a \to 0} \left(\Upsilon(x) - \a \Upsilon(z) - (1-\a) \Upsilon(y)\right) \notag \\
 & \leq &
  \ln 12 - c \ln 18 + c r \ln 27 - c r \ln 8 - c (1 - 3r) \ln (1 - 3r) \notag \\
 & & \mbox{}  + 3 c (1-r) \ln (1-r)
    + \Upsilon(x) - \Upsilon(y) \label{eq:line1}
\end{eqnarray}
where $r$ is in the range $[0, 1-\frac{2}{3c}]$.

The derivative of (\ref{eq:line1}) with respect to $r$ is
 \[ 3c \ln 3 - 3c \ln 2 + 3c \ln (1-3r) - 3c \ln (1-r) + 3c \ln y, \]
and setting the derivative to 0 gives the equation
\begin{equation} \label{eq:line1deriv}
y = \frac{2(1-r)}{3(1 - 3r)}.
\end{equation}
The interval analysis program verifies that (\ref{eq:line1}) is no larger than $2(1-c)\ln 4 -  .007$, at each point
of $c \in \left[.67, 1\right)$ and $r \in \left(0, 1-\frac{2}{3c}\right)$ where (\ref{eq:line1deriv}) holds.
This result plus the results of Case 1 and Case 3 proves that for $c \in [.67, 1)$, 
each point on this boundary is smaller that $2(1-c)\ln 4 - 0.005$.

\subsubsection*{Case 10: the boundary with $\a = 0$, and $t = r$.}
Using (\ref{lim-a-min}), we get
\begin{eqnarray}
\lim_{\substack{\a \to 0 \\ t \to r}} f(\a, r, t) & = &
 \ln 4 -c \ln 4 + \ln 3 - c(2-2r)\ln 3 +c(1-r)\ln 2 \notag \\
 & & \mbox{} -c(1-r) \ln (1-r) -cr \ln r + r3c \ln r + (1-r)3c \ln (1-r) \notag \\
 & & \mbox{}
    + \lim_{\a \to 0} \left(\Upsilon(x) - \a \Upsilon(z) - (1-\a) \Upsilon(y)\right) \notag \\
 & \leq &
  \ln 12 - c \ln 18 + c r \ln 9 - c r \ln 2  + 2 c (1-r) \ln (1-r) \notag \\
 & & \mbox{}
 + 2 c r \ln r
    + \Upsilon(x) - \Upsilon(y) \label{eq:line2}
\end{eqnarray}
where $r$ is in the range $[0, 1 - \frac{2}{3c}]$.

The derivative of (\ref{eq:line2}) with respect to $r$ is
 \[ c \ln 9 - c \ln 2 - 2c \ln (1-r) + 2c \ln r + 3c \ln y, \]
and setting the derivative to 0 gives the equation
\begin{equation} \label{eq:line2deriv}
 y^3 = \frac{2(1 - r)^2}{9r^2}.
\end{equation}
The interval analysis program verifies that (\ref{eq:line2}) is no larger than $2(1-c)\ln 4 - 0.423$, at each point
of $c \in \left[.67, 1\right)$ and $r \in \left(0, 1-\frac{2}{3c}\right)$ where (\ref{eq:line2deriv}) holds.
This result plus the results of Case 1 and Case 4 proves that for $c \in [.67, 1)$, 
each point on this boundary is smaller that $2(1-c)\ln 4 - 0.046$.

\subsubsection*{Case 11: the boundary with $\a = 0$, and $r = 1 - \frac{2(1-\a)}{3c}$.}
Applying (\ref{lim-a-min}), (\ref{lim-r-max}), and Observation~\ref{obs:xmax}, we get
the following.
\begin{eqnarray}
\lim_{\substack{\a \to 0 \\ r \to 1 - \frac{2}{3c}}} f(\a, r, t) & = &
 \ln 4 -c \ln 4 + \ln 3 - (-c+tc+2)\ln 3 \notag \\
 & & \mbox{} +(-2c+2+2tc)\ln 2 \notag \\
 & & \mbox{} -(-2c+2+2tc) \ln \left(-2+\frac{2}{c}+2t\right) \notag \\
 & & \mbox{} -(3c-2-3ct) \ln \left(1-\frac{2}{3c}-t\right) -ct \ln t \notag \\
 & & \mbox{} + (3c-2) \ln \left(1-\frac{2}{3c}\right) \notag \\
 & & \mbox{}
   + 2 \ln \frac{2}{3c}
    + \lim_{\substack{\a \to 0 \\ r \to 1 - \frac{2}{3c}}} \left(\Upsilon(x) - \a \Upsilon(z) - (1-\a) \Upsilon(y)\right) \notag \\
 & \leq &
 2 \ln 4 - 3 \ln 3 - c \ln 4 + c \ln 3 -4 c t \ln 3 \notag \\
 & & \mbox{} + (3c - 2) \ln (3c - 2)
  - 2(1-c + tc) \ln (1 - c + tc) \notag \\
 & & \mbox{} - (3c - 2 - 3 t c) \ln (3 c - 2 - 3 t c) \notag \\
 & & \mbox{}
   - c t \ln t - (2+t)c\ln c, \label{eq:line3}
\end{eqnarray}
where $t$ is in the range $[0, 1-\frac{2}{3c}]$.

The derivative of (\ref{eq:line3}) with respect to $t$ is
\[-4c \ln 3 -2c \ln (1-c + tc) + 3c \ln(3c-2-3tc)-c \ln t - c \ln c,  \]
and setting the derivative to 0 gives the equation
\begin{equation} \label{eq:line3deriv}
81 t c (1-c+tc)^2 = (3c-2-3tc)^3.
\end{equation}
The interval analysis program verifies that (\ref{eq:line3}) is no larger than $2(1-c)\ln 4 - 0.015$,  at each point
of $c \in \left[.67, 1\right)$ and $t \in \left(0, 1-\frac{2}{3c}\right)$ where (\ref{eq:line3deriv}) holds.
This result plus the results of Case 3 and Case 4 proves that for $c \in [.67, 1)$, 
each point on this boundary is smaller that $2(1-c)\ln 4 - 0.005$.

\subsubsection*{Case 12: the boundary with $\a = 1$, and $t = \frac{3r-1}{2}$.}
From (\ref{lim-a-max}),
\begin{eqnarray}
\lim_{\substack{\a \to 1 \\ t \to \frac{3r-1}{2}}} f(\a, r, t) & = &
 \ln 4 -c \ln 4 - c\left(2+\frac{3r-1}{2}-3r\right)\ln 3 \notag \\
 & & \mbox{} -3c\left(r-\frac{3r-1}{2}\right) \ln \left(r-\frac{3r-1}{2}\right) -c\frac{3r-1}{2} \ln \frac{3r-1}{2} \notag \\
 & & \mbox{} + r3c \ln r + (1-r)3c \ln (1-r) \notag \\
 & & \mbox{}
    + \lim_{\a \to 1} \left(\Upsilon(x) - \a \Upsilon(z) - (1-\a) \Upsilon(y)\right) \notag \\
 & \leq &
  \ln 4 - c \ln 2 - \frac{3 c}{2}(1-r)\ln 3 + \frac{3c}{2}(1-r)\ln(1-r) \notag \\
 & & \mbox{} +3cr\ln r - \frac{c}{2}(3r-1)\ln (3r-1)
    + \Upsilon(x) - \Upsilon(z) \label{eq:line4}
\end{eqnarray}
where $r$ is in the range $[\frac{2}{3c}, 1]$.

The derivative of (\ref{eq:line4}) with respect to $r$ is
 \[\frac{3c}{2}\ln 3 - \frac{3c}{2}\ln(1-r)+3c\ln r-\frac{3c}{2}\ln (3r-1) - 3c \ln z,  \]
and setting the derivative to 0 gives the equation
\begin{equation} \label{eq:line4deriv}
  3r^2 = z^2 (1 - r)(3r-1).
\end{equation}
The interval analysis program verifies that (\ref{eq:line4}) is no larger than $2(1-c)\ln 4 -  0.015$,  at each point
of $c \in \left[.67, 1\right)$ and $r \in \left(\frac{2}{3c}, 1\right)$ where (\ref{eq:line4deriv}) holds.
This result plus the results of Case 2 and Case 6 proves that for each $c \in [.67, 1)$ there exists a 
positive constant $b = b(c)$ such that 
each point on this boundary is smaller that $2(1-c)\ln 4 - b(c)$.

\subsubsection*{Case 13: the boundary with $\a = 1$, and $t = r$.}
From (\ref{lim-a-max}),
\begin{eqnarray}
\lim_{\substack{\a \to 1 \\ t \to r}} f(\a, r, t) & = &
 \ln 4 -c \ln 4 + - c(2-2r)\ln 3 +c(1-r)\ln 2 \notag \\
 & & \mbox{} -c(1-r) \ln (1-r) -cr \ln r + r3c \ln r + (1-r)3c \ln (1-r) \notag \\
 & & \mbox{}
    + \lim_{\a \to 1} \left(\Upsilon(x) - \a \Upsilon(z) - (1-\a) \Upsilon(y)\right) \notag \\
 & \leq &
 \ln 4 - c \ln 2 - c r \ln 2 - c (1-r) \ln 9 + 2 c (1-r) \ln (1-r) \notag \\
 & & \mbox{}
 + 2 c r \ln r
    + \Upsilon(x) - \Upsilon(z) \label{eq:line5}
\end{eqnarray}
where $r$ is in the range $[\frac{2}{3c}, 1]$.

The derivative of (\ref{eq:line5}) with respect to $r$ is
 \[-c \ln 2 + c \ln 9 - 2 c \ln (1-r) + 2 c \ln r - 3c \ln z,  \]
and setting the derivative to 0 gives the equation
\begin{equation} \label{eq:line5deriv}
  9 r^2 = 2 z^3(1 - r)^2.
\end{equation}
The interval analysis program verifies that (\ref{eq:line5}) is no larger than $2(1-c)\ln 4 - 0.015$,  at each point
of $c \in \left[.67, 1\right)$ and $r \in \left(\frac{2}{3c}, 1\right)$ where (\ref{eq:line5deriv}) holds.
This result plus the results of Case 2 and Case 5 proves that for each $c \in [.67, 1)$ there exists a 
positive constant $b = b(c)$ such that 
each point on this boundary is smaller that $2(1-c)\ln 4 - b(c)$.

\subsubsection*{Case 14: the boundary with $\a = 1$, and $r = \frac{2\a}{3c}$.}
Applying (\ref{lim-a-max}), (\ref{lim-r-min}), and Observation~\ref{obs:xmax}, we get
the following.
\begin{eqnarray}
\lim_{\substack{\a \to 0 \\ r \to \frac{2}{3c}}} f(\a, r, t) & = &
 \ln 4 -c \ln 4 - (2c+tc-2)\ln 3 +(c-2+2tc)\ln 2 \notag \\
 & & \mbox{} -(c-2+2tc) \ln \left(1-\frac{2}{c}+2t\right) -(2-3tc) \ln \left(\frac{2}{3c}-t\right) \notag \\
 & & \mbox{} -ct \ln t + 2 \ln \frac{2}{3c} + (3c-2) \ln \left(1-\frac{2}{3c}\right) \notag \\
 & & \mbox{}
    + \lim_{\substack{\a \to 1 \\ r \to \frac{2}{3c}}} \left(\Upsilon(x) - \a \Upsilon(z) - (1-\a) \Upsilon(y)\right) \notag \\
 & \leq &
 \ln 4 + 4 \ln 3 - c \ln 2 - 5 c \ln 3 + c t \ln 4 - 2 c t \ln 9 \notag \\
 & & \mbox{}
  + (3c-2) \ln(3c-2)
 - (c-2 + 2tc) \ln (c - 2 + 2tc) \notag \\
 & & \mbox{}
 - (2 - 3 t c) \ln (2 - 3 t c) - c t \ln t
  - (2+t)c\ln c \label{eq:line6}
\end{eqnarray}
where $t$ is in the range $[\frac{2-c}{2c}, \frac{2}{3c}]$.

The derivative of (\ref{eq:line6}) with respect to $t$ is
\[c \ln 4 -2c \ln 9 -2c \ln (c - 2 + 2tc) + 3c \ln(2-3tc)-c \ln t - c \ln c,  \]
and setting the derivative to 0 gives the equation
\begin{equation} \label{eq:line6deriv}
4(2-3tc)^3 = 81 t c (c-2+2tc)^2.
\end{equation}
The interval analysis program verifies that (\ref{eq:line6}) is no larger than $2(1-c)\ln 4 -  0.220$,  at each point
of $c \in \left[.67, 1\right)$ and $t \in \left(\frac{2-c}{2c}, \frac{2}{3c}\right)$ where (\ref{eq:line6deriv}) holds.
This result plus the results of Case 5 and Case 6 proves that for $c \in [.67, 1)$, 
each point on this boundary is smaller that $2(1-c)\ln 4 - 0.071$.

\subsubsection*{Case 15: the boundary with $r = \frac{2\a}{3c}$, and $t = r$.}
Applying (\ref{lim-r-min}), we get the following.
\begin{eqnarray}
\lim_{\substack{r \to \frac{2\a}{3c} \\ t \to r}} f(\a, r, t) & = &
 \ln 4 -c \ln 4 + (1-\a)\ln 3 - \left(2c-\frac{4\a}{3}\right)\ln 3 \notag \\
 & & \mbox{} +\left(c-\frac{2\a}{3}\right)\ln 2 -\alpha \ln \alpha -(1-\alpha)\ln (1-\alpha) \notag \\
 & & \mbox{} -\left(c-\frac{2\a}{3}\right) \ln \left(1-\frac{2\a}{3c}\right) -\frac{2\a}{3} \ln \frac{2\a}{3c} + 2\a \ln \frac{2\a}{3c} \notag \\
 & & \mbox{} + \left(3c-2\a\right) \ln \left(1-\frac{2\a}{3c}\right) \notag \\
 & & \mbox{}
    + \lim_{r \to \frac{2\a}{3c}} \left(\Upsilon(x) - \a \Upsilon(z) - (1-\a) \Upsilon(y)\right) \notag \\
 & = &
  \ln 12 - c \ln 162 + \frac{\a}{3}\ln 12 + \frac{a}{3}\ln \a - (1-\a) \ln (1-\a) \notag \\
 & & \mbox{}
   + \frac{2}{3}(3c-2\a)\ln(3c-2\a)-2c \ln c \notag \\
 & & \mbox{}
    + \Upsilon(x) - (1-\a)\Upsilon(y) - \a \ln 2 \label{eq:line7}
\end{eqnarray}
where $\a$ is in the range $[0,1]$.

The derivative of (\ref{eq:line7}) with respect to $\a$ is
 \[\frac{1}{3}\ln 12 + \frac{1}{3}\ln \a + \ln (1-\a)-\frac{4}{3}\ln(3c-2\a) + 2 \ln y - \ln(\rme^y-1-y) - \ln 2,  \]
and setting the derivative to 0 gives the equation
\begin{equation} \label{eq:line7deriv}
  3\a(1-\a)^3 = 2(3c-2\a)^4\left(\frac{\rme^y-1-y}{y^2}\right)^3.
\end{equation}
The interval analysis program verifies that (\ref{eq:line7}) is no larger than $2(1-c)\ln 4 -  7\times10^{-5}$, at each point
of $c \in \left[.67, 1\right)$ and $\a \in \left(0, 1\right)$ where (\ref{eq:line7deriv}) holds.
This result plus the results of Case 1 and Case 5 proves that for $c \in [.67, 1)$, 
each point on this boundary is smaller that $2(1-c)\ln 4 - 7\times10^{-5}$.

\subsubsection*{Case 16: the boundary with $r = \frac{2\a}{3c}$, and $t = \frac{3r-1}{2}$.}
Applying (\ref{lim-r-min}), we get
\begin{eqnarray}
\lim_{\substack{r \to \frac{2\a}{3c} \\ t \to \frac{3r-1}{2}}} f(\a, r, t) & = &
 \ln 4 -c \ln 4 + (1-\a)\ln 3 - \left(2c+\frac{2\a-c}{2}-2\a\right)\ln 3 \notag \\
 & & \mbox{} -\alpha \ln \alpha -(1-\alpha)\ln (1-\alpha) \notag \\
 & & \mbox{} -3\left(\frac{2\a}{3}-\frac{2\a-c}{2}\right) \ln \left(\frac{2\a}{3c}-\frac{2\a-c}{2c}\right) \notag \\
 & & \mbox{} -\frac{2\a-c}{2} \ln \frac{2\a-c}{2c}
             + 2\a \ln \frac{2\a}{3c} + (3c-2\a) \ln \left(1-\frac{2\a}{3c}\right) \notag \\
 & & \mbox{}
    + \lim_{r \to \frac{2\a}{3c}} \left(\Upsilon(x) - \a \Upsilon(z) - (1-\a) \Upsilon(y)\right) \notag \\
 & = &
  \ln 12 - c \ln 2 -3c\ln 3 - \a \ln 3+2\a\ln 2
  +\a\ln \a \notag \\
 & & \mbox{}
 -(1-\a)\ln(1-\a)
   + \frac{3c-2\a}{2}\ln(3c-2\a) \notag \\
 & & \mbox{} - \frac{2\a-c}{2}\ln(2\a-c) \notag \\
 & & \mbox{}
  -2c \ln c
    + \Upsilon(x) - (1-\a)\Upsilon(y) - \a \ln 2 \label{eq:line8}
\end{eqnarray}
where $\a$ is in the range $[\frac{c}{2}, 1]$.

The derivative of (\ref{eq:line8}) with respect to $\a$ is
 \[-\ln 3 + \ln 4 + \ln (1-\a) + \ln \a - \ln (3c-2\a)-\ln(2\a-c) + 2 \ln y -\ln(\rme^y-1-y) - \ln 2,  \]
and setting the derivative to 0 gives the equation
\begin{equation} \label{eq:line8deriv}
  2\a(1-\a) = 3(3c-2\a)(2\a-c)\left(\frac{\rme^y-1-y}{y^2}\right).
\end{equation}
The interval analysis program verifies that (\ref{eq:line8}) is no larger than $2(1-c)\ln 4 -  0.012$,  at each point
of $c \in \left[.67, 1\right)$ and $\a \in \left(\frac{c}{2}, 1\right)$ where (\ref{eq:line8deriv}) holds.
This result plus the results of Case 6 and Case 7 proves that for $c \in [.67, 1)$, 
each point on this boundary is smaller that $2(1-c)\ln 4 - 0.012$.

\subsubsection*{Case 17: the boundary with $r = \frac{2\a}{3c}$, and $t = 0$.}
Applying (\ref{lim-r-min}), we get
\begin{eqnarray}
\lim_{\substack{r \to \frac{2\a}{3c} \\ t \to 0}} f(\a, r, t) & = &
 \ln 4 -c \ln 4 + (1-\a)\ln 3 - (2c-2\a)\ln 3 +(c-2\a)\ln 2 \notag \\
 & & \mbox{} -\alpha \ln \alpha -(1-\alpha)\ln (1-\alpha) -(c-2\a) \ln \left(1-\frac{2\a}{c}\right) \notag \\
 & & \mbox{} - 2\a \ln \frac{2\a}{3c} + 2\a \ln \frac{2\a}{3c} + (3c-2\a) \ln \left(1-\frac{2\a}{3c}\right) \notag \\
 & & \mbox{}
    + \lim_{r \to \frac{2\a}{3c}} \left(\Upsilon(x) - \a \Upsilon(z) - (1-\a) \Upsilon(y)\right) \notag \\
 & = &
  \ln 12 - c \ln 2 -5c\ln 3 + 3\a \ln 3-\a\ln 4-(1-\a)\ln(1-\a) \notag \\
 & & \mbox{}
 -\a\ln \a
   - (c-2\a)\ln(c-2\a) +(3c-2\a)\ln(3c-2\a) \notag \\
 & & \mbox{}
   -2c \ln c
    + \Upsilon(x) - (1-\a) \Upsilon(y) - \a \ln 2 \label{eq:line9}
\end{eqnarray}
where $\a$ is in the range $[0, \frac{c}{2}]$.

The derivative of (\ref{eq:line9}) with respect to $\a$ is
 \[3 \ln 3 - \ln 4 + \ln(1-\a)-\ln \a +2\ln(c-2\a)-2\ln(3c-2\a)+ 2 \ln y -\ln(\rme^y-1-y) - \ln 2,  \]
and setting the derivative to 0 gives the equation
\begin{equation} \label{eq:line9deriv}
 27(1-\a)(c-2\a)^2 = 8\a(3c-2\a)^2 \left(\frac{\rme^y-1-y}{y^2}\right) .
\end{equation}
The interval analysis program verifies that (\ref{eq:line9}) is no larger than $2(1-c)\ln 4 -  2\times10^{-4}$,  at each point
of $c \in \left[.67, 1\right)$ and $\a \in \left(0, \frac{c}{2}\right)$ where (\ref{eq:line9deriv}) holds.
This result plus the results of Case 1 and Case 7 proves that for $c \in [.67, 1)$, 
each point on this boundary is smaller that $2(1-c)\ln 4 - 2\times10^{-4}$.

\subsubsection*{Case 18: the boundary with $r = 1-\frac{2(1-\a)}{3c}$, and $t = r$.}
From (\ref{lim-r-max}), we have
\begin{eqnarray}
\lim_{\substack{r \to 1-\frac{2\a}{3c} \\ t \to r}} f(\a, r, t) & = &
 \ln 4 -c \ln 4 + (1-\a)\ln 3 - \frac{4(1-\a)}{3}\ln 3 +\frac{2(1-\a)}{3}\ln 2 \notag \\
 & & \mbox{} -\alpha \ln \alpha -(1-\alpha)\ln (1-\alpha) \notag -\frac{2(1-\a)}{3} \ln \frac{2(1-\a)}{3c} \notag \\
 & & \mbox{} -\left(c-\frac{2(1-\a)}{3}\right) \ln \left(1-\frac{2(1-\a)}{3c}\right) \notag \\
 & & \mbox{} + (3c-2(1-\a))3c \ln \left(1-\frac{2(1-\a)}{3c}\right) \notag \\
 & & \mbox{} + 2(1-\a) \ln \frac{2(1-\a)}{3c} \notag \\
 & & \mbox{}
    + \lim_{r \to 1-\frac{2\a}{3c}} \left(\Upsilon(x) - \a \Upsilon(z) - (1-\a) \Upsilon(y)\right) \notag \\
 & = &
 \ln 4 - c \ln 36 + 2(1- \a) \ln 2 - \frac{1-\a}{3}\ln 3 + \frac{1-\a}{3}\ln(1-\a) \notag \\
 & & \mbox{} - \a \ln \a
  + \frac{2}{3}(2\a + 3c - 2) \ln (2\a + 3c -2) - 2 c \ln c \notag \\
 & & \mbox{}
    + \Upsilon(x) - \a \Upsilon(z) - (1-\a)\ln 2 \label{eq:line10}
\end{eqnarray}
where $\a$ is in the range $[0,1]$.

The derivative of (\ref{eq:line10}) with respect to $\a$ is
 \[-\ln 4 + \frac{1}{3}\ln 3 - \frac{1}{3} \ln (1-\a) - \ln \a + \frac{4}{3}\ln (2\a + 3c -2) - 2 \ln z + \ln(\rme^z-1-z) + \ln 2,  \]
and setting the derivative to 0 gives the equation
\begin{equation} \label{eq:line10deriv}
 8(1-\a)\a^3 = 3(2\a+3c-2)^4 \left(\frac{\rme^z-1-z}{z^2}\right)^3 .
\end{equation}
The interval analysis program verifies that (\ref{eq:line10}) is no larger than $2(1-c)\ln 4 - 3\times10^{-6}$, at each point
of $c \in \left[.67, 1\right)$ and $\a \in \left(0, 1\right)$ where (\ref{eq:line10deriv}) holds.
This result plus the results of Case 2 and Case 4 proves that for each $c \in [.67, 1)$ there exists a 
positive constant $b = b(c)$ such that 
each point on this boundary is smaller that $2(1-c)\ln 4 - b(c)$.

\subsubsection*{Case 19: the boundary with $r = 1-\frac{2(1-\a)}{3c}$, and $t = \frac{3r-1}{2}$.}
From (\ref{lim-r-max}), we have
\begin{eqnarray}
\lim_{\substack{r \to 1-\frac{2\a}{3c} \\ t \to \frac{3r-1}{2}}} f(\a, r, t) & = &
 \ln 4 -c \ln 4 + (1-\a)\ln 3 - (1-\a)\ln 3 \notag -\alpha \ln \alpha \notag \\
 & & \mbox{} -(1-\alpha)\ln (1-\alpha) -(1-\a) \ln \frac{1-\a}{3c} \notag \\
 & & \mbox{} -(c-(1-\a)) \ln \left(1-\frac{(1-\a)}{c}\right) \notag \\
 & & \mbox{} + (3c-2(1-\a)) \ln \left(1-\frac{2(1-\a)}{3c}\right) \notag \\
 & & \mbox{} + 2(1-\a) \ln \frac{2(1-\a)}{3c} \notag \\
 & & \mbox{}
    + \lim_{r \to 1-\frac{2\a}{3c}} \left(\Upsilon(x) - \a \Upsilon(z) - (1-\a) \Upsilon(y)\right) \notag \\
 & = &
 \ln 4 - 2(\a + c - 1)\ln 2 - (\a + 3c - 1)\ln 3 - \a \ln \a \notag \\
 & & \mbox{}
  -(c + \a - 1) \ln (c + \a - 1) \notag \\
 & & \mbox{} + (3c -2 + 2\a) \ln(3c-2+2\a) - 2 c \ln c \\
 & & \mbox{} + \Upsilon(x) - \a \Upsilon(z) - (1-\a)\ln 2 \label{eq:line11}
\end{eqnarray}
where $\a$ is in the range $[1-c, 1]$.

The derivative of (\ref{eq:line11}) with respect to $\a$ is
 \[-\ln 3 - 2 \ln 2 - \ln \a  - \ln (c+\a-1) + 2 \ln (3c-2+2\a) - 2 \ln z +\ln(\rme^z-1-z)+\ln 2,  \]
and setting the derivative to 0 gives the equation
\begin{equation} \label{eq:line11deriv}
 6\a (c + \a - 1) = (3c-2+2\a)^2 \left(\frac{\rme^z-1-z}{z^2}\right) .
\end{equation}
The interval analysis program verifies that (\ref{eq:line11}) is no larger than $2(1-c)\ln 4 -  0.034$,  at each point
of $c \in \left[.67, 1\right)$ and $\a \in \left(1-c, 1\right)$ where (\ref{eq:line11deriv}) holds.
This result plus the results of Case 2 and Case 8 proves that for each $c \in [.67, 1)$ there exists a 
positive constant $b = b(c)$ such that 
each point on this boundary is smaller that $2(1-c)\ln 4 - b(c)$.

\subsubsection*{Case 20: the boundary with $r = 1-\frac{2(1-\a)}{3c}$, and $t = 0$.}
From (\ref{lim-r-max}), we have
\begin{eqnarray}
\lim_{\substack{r \to 1-\frac{2\a}{3c} \\ t \to 0}} f(\a, r, t) & = &
 \ln 4 -c \ln 4 + (1-\a)\ln 3 - (-c+2(1-\a))\ln 3 \notag \\
 & & \mbox{} +(-2c+2(1-\a))\ln 2 -\alpha \ln \alpha -(1-\alpha)\ln (1-\alpha) \notag \\
 & & \mbox{} -(-2c+2(1-\a)) \ln \left(-2+\frac{2(1-\a)}{c}\right) \notag \\
 & & \mbox{} -(3c-2(1-\a)) \ln \left(1-\frac{2(1-\a)}{3c}\right) \notag \\
 & & \mbox{} + (3c-2(1-\a)) \ln \left(1-\frac{2(1-\a)}{3c}\right) \notag \\
 & & \mbox{} + 2(1-\a) \ln \frac{2(1-\a)}{3c} \notag \\
 & & \mbox{}
    + \lim_{r \to 1-\frac{2\a}{3c}} \left(\Upsilon(x) - \a \Upsilon(z) - (1-\a) \Upsilon(y)\right) \notag \\
 & = &
 \ln 16 - \ln 27 - c \ln 4 + c \ln 3 - \a \ln 4 + \a \ln 27 - \a \ln \a \notag \\
 & & \mbox{}
  +(1-\a)\ln (1-\a) -2(1-\a-c)\ln(1-\a-c) -2 c \ln c \notag \\
 & & \mbox{}
    + \Upsilon(x) - \a \Upsilon(z) - (1-\a)\ln 2 \label{eq:line12}
\end{eqnarray}
where $\a$ is in the range $[0, 1-c]$.

The derivative of (\ref{eq:line12}) with respect to $\a$ is
 \[-\ln 4 + \ln 27 - \ln \a - \ln (1-\a) + 2 \ln (1-\a-c) - 2 \ln z + \ln (\rme^z-1-z) + \ln 2,  \]
and setting the derivative to 0 gives the equation
\begin{equation} \label{eq:line12deriv}
 2 \a(1-\a) = 27(1-\a-c)^2 \left(\frac{\rme^z-1-z}{z^2}\right) .
\end{equation}
The interval analysis program verifies that (\ref{eq:line12}) is no larger than $2(1-c)\ln 4 -  9\times10^{-4}$,  at each point
of $c \in \left[.67, 1\right)$ and $\a \in \left(0, 1-c\right)$ where (\ref{eq:line12deriv}) holds.
This result plus the results of Case 3 and Case 8 proves that for $c \in [.67, 1)$, 
each point on this boundary is smaller that $2(1-c)\ln 4 - 9\times10^{-4}$.

\subsubsection*{Case 21: the boundary with $t = 0 = \frac{3r-1}{2}$.}
\begin{eqnarray}
\lim_{\substack{r \to \frac{1}{3} \\ t \to 0}} f(\a, r, t) & = &
 \ln 4 -c \ln 4 + (1-\a)\ln 3 - c\ln 3 \notag -\alpha \ln \alpha \notag \\
 & & \mbox{} -(1-\alpha)\ln (1-\alpha) - c \ln \frac{1}{3} + c \ln \frac{1}{3} + 2c \ln \frac{2}{3} \notag \\
 & & \mbox{}
    + \lim_{r \to \frac{1}{3}} \left(\Upsilon(x) - \a \Upsilon(z) - (1-\a) \Upsilon(y)\right) \notag \\
 & = &
  \ln 12 - c \ln 27 - \a \ln 3 - \a \ln \a - (1-\a) \ln (1-\a) \notag \\
 & & \mbox{}
    + \lim_{r \to \frac{1}{3}} \left(\Upsilon(x) - \a \Upsilon(z) - (1-\a) \Upsilon(y)\right) \label{eq:line13}
\end{eqnarray}
where $\a$ is in the range $[1-c, \frac{c}{2}]$.

The derivative of (\ref{eq:line13}) with respect to $\a$ is
 \[-\ln 3 - \ln \a + \ln (1-\a) + \ln (\rme^z-1-z) - \ln (\rme^y-1-y),  \]
and setting the derivative to 0 gives the equation
\begin{equation} \label{eq:line13deriv}
  \frac{1-\a}{\a} = 3 \frac{\rme^y-1-y}{\rme^z-1-z}.
\end{equation}
The interval analysis program verifies that (\ref{eq:line13}) is no larger than $2(1-c)\ln 4 - 0.130$,  at each point
of $c \in \left[.67, 1\right)$ and $\a \in \left(1-c, \frac{c}{2}\right)$ where (\ref{eq:line13deriv}) holds.
This result plus the results of Case 7 and Case 8 proves that for $c \in [.67, 1)$, 
each point on this boundary is smaller that $2(1-c)\ln 4 - 0.130$.

\subsubsection{The Boundary Cases That Are Faces}

In each case, we rule out the existence of a point on the boundary where $f$ exceeds \mbox{$2(1-c)\ln 4$}.
We will compute a function that is either the limit of $f$ as some parameter approaches the boundary or
an upper bound on this limit.  This new function will be continuous, and any
maximum of this function will occur either where both its partial first derivatives are 0 or
at an edge of the face.
Cases 9--21 above evaluate the line segments that form the boundaries of each face,
and the analysis that follows will focus on the points on each face where both
partial first derivatives are 0.

\subsubsection*{Case 22: the boundary with $\a = 0$.}
Applying (\ref{lim-a-min}), we get
\begin{eqnarray}
\lim_{\alpha \rightarrow 0} f(\alpha , r,t) & = & \ln 4 -c \ln 4 + \ln 3 - c(2+t-3r)\ln 3 +c(1-3r+2t)\ln 2 \notag \\
 & & \mbox{} -c(1-3r+2t) \ln (1-3r+2t) -c(3r-3t) \ln (r-t) \notag \\
 & & \mbox{} -ct \ln t + r3c \ln r + (1-r)3c \ln (1-r)   \notag\\
 & & \mbox{} + \lim_{\alpha \rightarrow 0} \left(\Upsilon(x) - \alpha \Upsilon(z) - (1-\alpha) \Upsilon(y) \right) \notag \\
 & \leq &
\ln 4 -c \ln 4 + \ln 3 - c(2+t-3r)\ln 3 +c(1-3r+2t)\ln 2 \notag \\
 & & \mbox{} -c(1-3r+2t) \ln (1-3r+2t) -c(3r-3t) \ln (r-t) \notag \\
 & & \mbox{} -ct \ln t + r3c \ln r + (1-r)3c \ln (1-r)   \notag\\
 & & \mbox{} + \Upsilon(x) - \Upsilon(y)
\label{eq:face1}
\end{eqnarray}
where $r$ is in the range $[0, 1-\frac{2}{3c}]$, and $t$ is in the range
$[0, r]$.  The partial derivative of (\ref{eq:face1}) with respect to $r$ is
\[ 3c \ln 3 - 3c \ln 2 + 3c \ln(1-3r+2t) - 3c\ln(r-t) + 3c\ln r -3c \ln (1-r) + 3c \ln y, \]
and the partial derivative with respect to $t$ is
\[-c \ln 3 + 2c \ln 2 - 2c \ln (1-3r+2t) + 3c \ln (r-t) - c \ln t .\]
Setting the derivatives to 0 give the equations
\begin{eqnarray}
2(1-r)(r-t) &=& 3 y r (1-3r + 2t)  \label{eq:face1eq1}\\
3t(1-3r+2t)^2 &=& 4(r-t)^3         \label{eq:face1eq2}
\end{eqnarray}
which must hold for any assignment that maximizes (\ref{eq:face1}).
In addition, solving (\ref{eq:face1eq1}) for $t$ gives
\[ t = \frac{2(1-r)r - 3yr(1-3r)}{2(1-r) + 6yr}, \]
and plugging this value into (\ref{eq:face1eq2}) and simplifying gives the equation
\begin{equation} \label{eq:face1eq3}
 9r^2y^3 = (2(1-r) - 3y(1-3r))(1-r)
\end{equation}
which must also hold at any maximum of (\ref{eq:face1}).

The interval analysis program verifies that (\ref{eq:face1}) is no larger than $2(1-c)\ln 4 -  5\times10^{-4}$,  at each point
of $c \in \left[.67, 1\right)$, $r \in \left(0, 1-\frac{2}{3c}\right)$, and $t \in \left(0, r\right)$
where (\ref{eq:face1eq1}), (\ref{eq:face1eq2}), and (\ref{eq:face1eq3}) all hold.
This result plus the results of Cases 1, 3, 4, 9, 10, and 11 proves that for $c \in [.67, 1)$, 
each point on this boundary is smaller that $2(1-c)\ln 4 - 5\times10^{-4}$.

\subsubsection*{Case 23: the boundary with $\a = 1$.}
Applying (\ref{lim-a-max}), we get
\begin{eqnarray}
\lim_{\alpha \rightarrow 1} f(\alpha , r,t) & = & \ln 4 -c \ln 4 - c(2+t-3r)\ln 3 +c(1-3r+2t)\ln 2 \notag \\
 & & \mbox{} -c(1-3r+2t) \ln (1-3r+2t) -c(3r-3t) \ln (r-t) \notag \\
 & & \mbox{} -ct \ln t  + r3c \ln r + (1-r)3c \ln (1-r)   \notag\\
 & & \mbox{} + \lim_{\alpha \rightarrow 1} \left(\Upsilon(x) - \alpha \Upsilon(z) - (1-\alpha) \Upsilon(y) \right) \notag \\
 & \leq &
\ln 4 -c \ln 4 - c(2+t-3r)\ln 3 +c(1-3r+2t)\ln 2 \notag \\
 & & \mbox{} -c(1-3r+2t) \ln (1-3r+2t) -c(3r-3t) \ln (r-t) \notag \\
 & & \mbox{} -ct \ln t + r3c \ln r + (1-r)3c \ln (1-r)   \notag\\
 & & \mbox{} + \Upsilon(x) - \Upsilon(z)
\label{eq:face2}
\end{eqnarray}
where $r$ is in the range $[\frac{2}{3c}, 1]$, and $t$ is in the range
$\left[\frac{3r-1}{2}, r\right]$.  The partial derivative of (\ref{eq:face2}) with respect to $r$ is
\[ 3c \ln 3 - 3c \ln 2 + 3c \ln(1-3r+2t) - 3c\ln(r-t) + 3c\ln r -3c \ln (1-r) - 3c \ln z, \]
and the partial derivative with respect to $t$ is
\[-c \ln 3 + 2c \ln 2 - 2c \ln (1-3r+2t) + 3c \ln (r-t) - c \ln t .\]
Setting the derivatives to 0 give the equations
\begin{eqnarray}
(1-r)2z(r-t) &=& 3 r (1-3r + 2t) \label{eq:face2eq1} \\
3t(1-3r+2t)^2 &=& 4(r-t)^3 . \label{eq:face2eq2}
\end{eqnarray}
which must hold for any assignment that maximizes (\ref{eq:face2}).
In addition, solving (\ref{eq:face2eq1}) for $t$ gives
\[ t = \frac{2(1-r)rz - 3r(1-3r)}{2(1-r)z + 6r}, \]
and plugging this value into (\ref{eq:face2eq2}) and simplifying gives the equation
\begin{equation}\label{eq:face2eq3}
 9r^2 = (2(1-r)z - 3(1-3r))(1-r)z^2
\end{equation}
which must also hold at any maximum of (\ref{eq:face2}).

The interval analysis program verifies that (\ref{eq:face2}) is no larger than $2(1-c)\ln 4 - 5\times10^{-4}$,  at each point
of $c \in \left[.67, 1\right)$, $r \in \left(\frac{2}{3c}, 1\right)$, and $t \in \left(\frac{3r-1}{2}, r\right)$
where (\ref{eq:face2eq1}), (\ref{eq:face2eq2}), and (\ref{eq:face2eq3}) all hold.
This result plus the results of Cases 2, 5, 6, 12, 13, and 14 proves that for each $c \in [.67, 1)$ there exists a 
positive constant $b = b(c)$ such that 
each point on this boundary is smaller that $2(1-c)\ln 4 - b(c)$.

\subsubsection*{Case 24: the boundary with $r = \frac{2 \a}{3 c}$.}
Applying (\ref{lim-r-min}), we get
\begin{eqnarray}
\lim_{r \to \frac{2\a}{3c} } f(\alpha , r,t) & = &
 \ln 4 -c \ln 4 + (1-\a)\ln 3 - (2c+tc-2\a)\ln 3 \notag \\
 & & \mbox{} +(c-2\a+2tc)\ln 2 -\alpha \ln \alpha -(1-\alpha)\ln (1-\alpha) \notag \\
 & & \mbox{} -(c-2\a+2tc) \ln \left(1-\frac{2\a}{c}+2t\right) \notag \\
 & & \mbox{} -(2\a-3tc) \ln \left(\frac{2\a}{3c}-t\right) \notag \\
 & & \mbox{} -ct \ln t + 2\a \ln \frac{2\a}{3c}  + (3c-2\a) \ln \left(1-\frac{2\a}{3c}\right) \notag \\
 & & \mbox{} + \lim_{r \to \frac{2\a}{3c}} \left(\Upsilon(x) - \alpha \Upsilon(z) - (1-\alpha) \Upsilon(y) \right) \notag \\
 & = &
\ln 12 - c \ln 18 - c \ln 27 + \a \ln 27 + c t \ln 4 - 2 c t \ln 9 \notag \\
 & & \mbox{} + \a \ln \a - (1-\a)\ln(1-\a) +(3c-2\a) \ln (3c-2\a) \notag \\
 & & \mbox{} -(2\a-3 t c)\ln (2\a - 3 t c) - t c \ln t \notag\\
 & & \mbox{} - (c -2\a+2 t c)\ln(c - 2\a+2 t c) - 2 c \ln c - c t \ln c \notag\\
 & & \mbox{} + \Upsilon(x) - \a \ln 2 - (1-\a)\Upsilon(y)
\label{eq:face3}
\end{eqnarray}
where $\a$ is in the range $[0,1]$, and $t$ is in the range $\left[\max\left\{0, \frac{2\a-c}{2c}\right\}, \frac{2\a}{3c}\right]$.
The partial derivative of (\ref{eq:face3}) with respect to $\a$ is
\begin{multline*}
\ln 27 + \ln \a + \ln (1-\a) - 2 \ln (3c-2\a) -2 \ln(2\a-3tc) \\
\mbox{} +2 \ln (c-2\a+2tc) + 2\ln y - \ln(\rme^y-1-y) - \ln 2,
\end{multline*}
and the partial derivative with respect to $t$ is
\[c \ln 4 - 2 c \ln 9 + 3c \ln (2\a-3tc) -c \ln t -2c \ln(c-2\a+2tc) - c \ln c. \]
Setting the derivatives to 0 gives the equations
\begin{eqnarray}
27 \a (1-\a)(c-2\a+2tc)^2 &=& 2(3c-2\a)^2(2\a-3tc)^2\left(\frac{\rme^y-1-y}{y^2}\right) \label{eq:face3eq1} \\
4(2\a-3tc)^3 &=& 81 t c (c-2\a+2tc)^2 \label{eq:face3eq2}
\end{eqnarray}
which must hold at any maximum of (\ref{eq:face3}).

The interval analysis program verifies that (\ref{eq:face3}) is no larger than $2(1-c)\ln 4 - 1\times10^{-9}$,  at each point
of $c \in \left[.67, 1\right)$, $\a \in (0, 1)$, and $t \in \left(\max\left\{0, \frac{2\a-c}{2c}\right\}, \frac{2\a}{3c}\right)$
where (\ref{eq:face3eq1}) and (\ref{eq:face3eq2}) both hold.
This result plus the results of Cases 1, 5, 6, 7, 14, 15, 16, and 17 proves that for $c \in [.67, 1)$, 
each point on this boundary is smaller that $2(1-c)\ln 4 - 1\times10^{-9}$.

\subsubsection*{Case 25: the boundary with $r = 1-\frac{2 (1-\a)}{3 c}$.}
Applying (\ref{lim-r-max}), we get
\begin{eqnarray}
\lim_{r \to 1 - \frac{2(1-\a)}{3c}} f(\alpha , r,t)
& = & \ln 4 -c \ln 4 + (1-\a)\ln 3 - \left(-c+tc+2(1-\a)\right)\ln 3 \notag \\
 & & \mbox{} +\left(-2c+2(1-\a)+2tc\right)\ln 2 -\alpha \ln \alpha \notag \\
 & & \mbox{} -(1-\alpha)\ln (1-\alpha) \notag \\
 & & \mbox{} -\left(-2c+2(1-\a)+2tc\right) \ln \left(-2+\frac{2(1-\a)}{c}+2t\right) \notag \\
 & & \mbox{} -\left(3c-2(1-\a)-3tc\right) \ln \left(1-\frac{2(1-\a)}{3c}-t\right) \notag \\
 & & \mbox{} -ct \ln t + \left(3c-2(1-\a)\right) \ln \left(1-\frac{2(1-\a)}{3c}\right) \notag \\
 & & \mbox{} + 2(1-\a) \ln \frac{2(1-\a)}{3c} \notag \\
 & & \mbox{} + \lim_{r \to 1 - \frac{2(1-\a)}{3c}}(\Upsilon(x) - \alpha \Upsilon(z) - (1-\alpha) \Upsilon(y))  \notag \\
 & = &
\ln 4 - c \ln 108 + (1-\a)\ln \frac{4}{27} +2 c (1-t) \ln 9 -\a \ln \a \notag \\
 & & \mbox{} + (1-\a)\ln (1-\a) + (3c-2+2\a)\ln(3c-2+2\a) \notag \\
 & & \mbox{} - (3c(1-t)-2+2\a)\ln(3c(1-t)-2+2\a)- c t \ln t \notag \\
 & & \mbox{} -2(1-\a -c (1-t))\ln(1-\a-c(1-t)) \notag \\
 & & \mbox{} - (2 + t) c \ln c + \Upsilon(x) - \a \Upsilon(z) - (1-\a)\ln 2 \label{eq:face4}
\end{eqnarray}
where $\a$ is in the range $[0,1]$, and $t$ is in the range $\left[\min \left\{0, \frac{\a+c-1}{c}\right\}, 1 - \frac{2(1-\a)}{3c}\right]$.
The partial derivative of (\ref{eq:face4}) with respect to $\a$ is
\begin{multline*}
-\ln\frac{4}{27} -\ln \a - \ln(1-\a) +2\ln(3c-2+2\a)-2\ln(3c(1-t)-2+2\a) \\
 \mbox{}+2\ln(1-\a-c(1-t)) - 2 \ln z + \ln(\rme^z-1-z) + \ln 2,
\end{multline*}
and the partial derivative with respect to $t$ is
\[-2c\ln 9 +3c\ln(3c(1-t)-2+2\a) -c \ln t -2c\ln(1-\a-c(1-t)) - c \ln c. \]
Setting the derivatives to 0 gives the equations
\begin{alignat}{1}
27(3c-2+2\a)^2 (1-\a-c(1-t))^2 & \left(\frac{\rme^z-1-z}{z^2}\right) \\
  &= 2\a(1-\a)(3c(1-t)-2+2\a)^2 \label{eq:face4eq1} \\
(3c(1-t)-2+2\a)^3 &= 81 t c (1-\a-c(1-t))^2 \label{eq:face4eq2}
\end{alignat}
which must hold at any maximum of (\ref{eq:face4}).

The interval analysis program verifies that (\ref{eq:face4}) is no larger than $2(1-c)\ln 4 - 7\times10^{-6}$,  at each point
of $c \in \left[.67, 1\right)$, $\a \in (0, 1)$, and $t \in \left(\min\left\{0, \frac{\a+c-1}{c}\right\}, 1-\frac{2(1-\a)}{3c}\right)$
where (\ref{eq:face4eq1}) and (\ref{eq:face4eq2}) both hold.
This result plus the results of Cases 2, 3, 4, 8, 11, 18, 19, and 20 proves that for each $c \in [.67, 1)$ there exists a 
positive constant $b = b(c)$ such that 
each point on this boundary is smaller that $2(1-c)\ln 4 - b(c)$.

\subsubsection*{Case 26: the boundary with $t = 0$.}

\begin{eqnarray}
\lim_{t \to 0} f(\alpha , r,t) & = &
 \ln 4 -c \ln 4 + (1-\a)\ln 3 - c(2-3r)\ln 3 +c(1-3r)\ln 2 \notag \\
 & & \mbox{} -\alpha \ln \alpha -(1-\alpha)\ln (1-\alpha) -c(1-3r) \ln (1-3r) \notag \\
 & & \mbox{} -3cr \ln r + r3c \ln r + (1-r)3c \ln (1-r) \notag \\
 & & \mbox{} + \Upsilon(x) - \alpha \Upsilon(z) - (1-\alpha) \Upsilon(y)  \notag \\
 & = &
 \ln 4 -c \ln 18 + (1-\a)\ln 3 +3rc\ln \frac{3}{2} -\alpha \ln \alpha \notag \\
 & & \mbox{} -(1-\alpha)\ln (1-\alpha) -c(1-3r) \ln (1-3r) \notag \\
 & & \mbox{} + (1-r)3c \ln (1-r)
     + \Upsilon(x) - \alpha \Upsilon(z) - (1-\alpha) \Upsilon(y) \label{eq:face5}
\end{eqnarray}
where $\a$ is in the range $\left[0, \frac{c}{2}\right]$ and $r$ is in the range
$\left[\frac{2\a}{3c}, \min\left\{1-\frac{2(1-\a)}{3c}, \frac{1}{3}\right\}\right]$.
The partial derivative of (\ref{eq:face5}) with respect to $\a$ is
\[-\ln 3 - \ln \a + \ln(1-\a) + \ln(\rme^z-1-z) - \ln(\rme^y-1-y),\]
and with respect to $r$ is
\[ 3 c \ln \frac{3}{2} + 3c\ln(1-3r)-3c\ln(1-r)-3c\ln z + 3c \ln y . \]
Setting the partial derivatives to 0 gives the equations
\begin{eqnarray}
3\a(\rme^y-1-y) &=& (1-\a)(\rme^z-1-z) \label{case26:eq1} \\
2(1-r)z &=& 3 (1-3r)y  \label{case26:eq2}
\end{eqnarray}
which must hold at any maximum of (\ref{eq:face5}).

From (\ref{eq:kimplicitsMH}),
\[ \frac{\rme^y-1-y}{\rme^z-1-z} \cdot \frac{1-r}{r} = \frac{1-\a}{a} \cdot \frac{y}{z} \cdot \frac{\rme^y-1}{\rme^z-1}, \]
and plugging in (\ref{case26:eq1}) gives
\begin{equation}
\frac{1-r}{r} = 3 \frac{y}{z} \cdot \frac{\rme^y-1}{\rme^z-1}. \label{case26:eq3}
\end{equation}
Solving for $r$ in (\ref{case26:eq2}) gives
\[ r = \frac{3y-2z}{9y-3z}, \]
and thus
\begin{equation}
\frac{1-r}{r} = \frac{6y}{3y-2z}. \label{case26:eq4}
\end{equation}
Combining (\ref{case26:eq4}) with (\ref{case26:eq3}) gives the equation
\begin{equation} \label{case26:eq5}
2z(\rme^z-1) = (3y-2z)(\rme^z-1)
\end{equation}
that we can use to further restrict the possible values $z$ and $y$ can have at a maximum
for (\ref{eq:face5}).

The interval analysis program verifies that (\ref{eq:face5}) is no larger than $2(1-c)\ln 4 -  7\times10^{-10}$,  at each point
of $c \in \left[.67, 1\right)$, $\a \in \left(0, \frac{c}{2}\right)$, and $r \in \left(\frac{2\a}{3c}, \min\left\{1-\frac{2(1-\a)}{3c}, \frac{1}{3}\right\}\right)$
where (\ref{case26:eq1}), (\ref{case26:eq2}), and (\ref{case26:eq5}) all hold.
This result plus the results of Cases 1, 3, 7, 8, 9, 17, 20, and 21 proves that for $c \in [.67, 1)$, 
each point on this boundary is smaller that $2(1-c)\ln 4 - 7\times10^{-10}$.

\subsubsection*{Case 27: the boundary with $t = \frac{3r-1}{2}$.}

\begin{eqnarray}
\lim_{t \to \frac{3r-1}{2}} f(\alpha , r,t) & = &
 \ln 4 -c \ln 4 + (1-\a)\ln 3 - c\left(2+\frac{3r-1}{2}-3r\right)\ln 3 \notag \\
 & & \mbox{} -\alpha \ln \alpha -(1-\alpha)\ln (1-\alpha) \notag \\
 & & \mbox{} -3c\left(r-\frac{3r-1}{2}\right) \ln \left(r-\frac{3r-1}{2}\right) -c\frac{3r-1}{2} \ln \frac{3r-1}{2} \notag \\
 & & \mbox{} + r3c \ln r + (1-r)3c \ln (1-r) \notag \\
 & & \mbox{} + \Upsilon(x) - \alpha \Upsilon(z) - (1-\alpha) \Upsilon(y)  \notag \\
 & = &
 \ln 4 -c \ln 2 + (1-\a)\ln 3 - \frac{3}{2}(1-r)c\ln 3 \notag \\
 & & \mbox{} -\alpha \ln \alpha -(1-\alpha)\ln (1-\alpha) -\frac{c}{2}(3r-1) \ln (3r-1) \notag \\
 & & \mbox{} + 3cr \ln r + \frac{3}{2}c(1-r)\ln (1-r) \notag \\
 & & \mbox{} + \Upsilon(x) - \alpha \Upsilon(z) - (1-\alpha) \Upsilon(y) \label{eq:face6}
\end{eqnarray}
where $\a$ is in the range $[1-c, 1]$ and $r$ is in the range
$\left[\max\left\{\frac{2\a}{3c}, \frac{1}{3}\right\}, 1-\frac{2(1-\a)}{3c}\right]$.
The partial derivative of (\ref{eq:face6}) with respect to $\a$ is
\[-\ln 3 - \ln \a + \ln(1-\a) + \ln(\rme^z-1-z) - \ln(\rme^y-1-y),\]
and with respect to $r$ is
\[\frac{3}{2}c \ln 3 - \frac{3c}{2}\ln(3r-1)+3c\ln r - \frac{3c}{2}\ln(1-r) - 3c \ln z + 3c \ln y. \]
Setting the partial derivatives to 0 gives the equations
\begin{eqnarray}
3\a(\rme^y-1-y) &=& (1-\a)(\rme^z-1-z) \label{case27:eq1} \\
3 r^2 y^2 &=& (3r-1)(1-r)z^2 \label{case27:eq2}
\end{eqnarray}
which must hold at any maximum of (\ref{eq:face6}).

From (\ref{eq:kimplicitsMH}),
\[ \frac{\rme^y-1-y}{\rme^z-1-z} \cdot \frac{1-r}{r} = \frac{1-\a}{a} \cdot \frac{y}{z} \cdot \frac{\rme^y-1}{\rme^z-1}, \]
and plugging in (\ref{case27:eq1}) gives
\begin{equation}
\frac{1-r}{r} = 3 \frac{y}{z} \cdot \frac{\rme^y-1}{\rme^z-1}. \label{case27:eq3}
\end{equation}
Solving for $r$ in (\ref{case27:eq2}) gives
\[ r = \frac{2z^2 \pm z\sqrt{z^2-3y^2}}{3z^2+3y^2}, \]
and plugging this value into (\ref{case27:eq3}) gives two equations for the stationary
points of (\ref{eq:face6}):
\begin{eqnarray}
3 \frac{y}{z} \cdot \frac{\rme^y-1}{\rme^z-1} &=& \frac{3y^2 + z^2 - z\sqrt{z^2-3y^2}}{2z^2 + z\sqrt{z^2-3y^2}} \label{case27:eq4} \\
3 \frac{y}{z} \cdot \frac{\rme^y-1}{\rme^z-1} &=& \frac{3y^2 + z^2 + z\sqrt{z^2-3y^2}}{2z^2 - z\sqrt{z^2-3y^2}} \label{case27:eq5}.
\end{eqnarray}
We will use (\ref{case27:eq4}) and (\ref{case27:eq5}) in addition to the fact that $z^2 \geq 3 y^2$ to restrict the values $y$ and $z$
can have at a maximum for (\ref{eq:face6}).

However, note that if $y = 0$, then any value of $z$
will satisfy (\ref{case27:eq4}).  As a result, we need additional analysis to rule out
cases where $y \in (0, y_2]$ and $z \in [z_1, z_2]$.
Consider the two sides of (\ref{case27:eq4}).
\[ \frac{3y(\rme^y-1)}{z(\rme^z-1)} \]
\[ \frac{3y^2 + z^2 - z\sqrt{z^2-3y^2}}{2z^2 + z\sqrt{z^2-3y^2}}. \]
The first derivatives of the two sides with respect to $y$ are
\[ \frac{3(\rme^y-1+y\rme^y)}{z(\rme^z-1)} \]
\[ \frac{3y}{z\sqrt{z^2-3y^2}}, \]
and the second derivatives are
\begin{equation}\label{case27:2nd_deriv_1}
 \frac{3\rme^y(y+2)}{z(\rme^z-1)}
\end{equation}
\begin{equation}\label{case27:2nd_deriv_2}
 \frac{3z}{(z^2-3y^2)^{3/2}}.
\end{equation}
Note that both sides of (\ref{case27:eq4}) and their first derivatives are 0 when $y=0$, and
both second derivatives are positive.  As a result, if there is no $y \in (0, y_1]$ and
$z \in [z_1, z_2]$ such that the two second derivatives are equal, then there is no
$y \in (0, y_1]$ and $z \in [z_1, z_3]$ that satisfy (\ref{case27:eq4}).

The interval analysis program verifies that (\ref{eq:face6}) is no larger than $2(1-c)\ln 4 - 1\times10^{-7}$, at each point
of $c \in \left[.67, 1\right)$, $\a \in \left(1-c, 1\right)$, and $r \in \left( \max\left\{\frac{2\a}{3c}, \frac{1}{3}\right\}, 1-\frac{2(1-\a)}{3c}\right)$
where (\ref{case27:eq1}), (\ref{case27:eq2}), and $z^2 \geq 3 y^2$ all hold and either
(\ref{case27:eq4}) or (\ref{case27:eq5}) holds.
In addition, if we are considering an interval for $y$ that has the lower bound equal to 0, we use
the additional test that (\ref{case27:2nd_deriv_1}) must equal (\ref{case27:2nd_deriv_2}) for some $y$ in
that interval.
This result plus the results of Cases 2, 6, 7, 8, 12, 16, 19, and 21 proves that for each $c \in [.67, 1)$ there exists a 
positive constant $b = b(c)$ such that 
each point on this boundary is smaller that $2(1-c)\ln 4 - b(c)$.

\subsubsection*{Case 28: the boundary with $t = r$.}

\begin{eqnarray}
\lim_{t \to r} f(\alpha , r,t) & = &
 \ln 4 -c \ln 4 + (1-\a)\ln 3 - c(2-2r)\ln 3 \notag \\
 & & \mbox{} +c(1-r)\ln 2 -\alpha \ln \alpha -(1-\alpha)\ln (1-\alpha) \notag \\
 & & \mbox{} -c(1-r) \ln (1-r) -cr \ln r + r3c \ln r \notag \\
 & & \mbox{} + (1-r)3c \ln (1-r) \notag \\
 & & \mbox{} + \Upsilon(x) - \alpha \Upsilon(z) - (1-\alpha) \Upsilon(y)  \notag \\
 & = &
 \ln 4 -c \ln 4 + (1-\a)\ln 3 - c(1-r)\ln \frac{9}{2} \notag \\
 & & \mbox{} -\alpha \ln \alpha -(1-\alpha)\ln (1-\alpha) \notag \\
 & & \mbox{} + 2cr \ln r + 2c(1-r) \ln (1-r) \notag \\
 & & \mbox{} + \Upsilon(x) - \alpha \Upsilon(z) - (1-\alpha) \Upsilon(y) \label{eq:face7}
\end{eqnarray}
where $\a$ is in the range $[0,1]$ and $r$ is in the range $\left[\frac{2\a}{3c}, 1-\frac{2(1-\a)}{3c}\right]$.
The partial derivative of (\ref{eq:face6}) with respect to $\a$ is
\[-\ln 3 - \ln \a + \ln(1-\a) + \ln(\rme^z-1-z) - \ln(\rme^y-1-y),\]
and with respect to $r$ is
\[ c \ln \frac{9}{2} + 2c \ln r - 2c \ln(1-r) - 3c\ln z + 3c \ln y.\]
Setting the partial derivatives to 0 gives the equations
\begin{eqnarray}
3\a(\rme^y-1-y) &=& (1-\a)(\rme^z-1-z) \label{case28:eq1} \\
9 r^2 y^3 &=& 2 (1-r)^2 z^3 \label{case28:eq2}
\end{eqnarray}
which must hold at any maximum of (\ref{eq:face7}).

From (\ref{eq:kimplicitsMH}),
\[ \frac{\rme^y-1-y}{\rme^z-1-z} \cdot \frac{1-r}{r} = \frac{1-\a}{a} \cdot \frac{y}{z} \cdot \frac{\rme^y-1}{\rme^z-1}, \]
and plugging in (\ref{case28:eq1}) gives
\begin{equation}
\frac{1-r}{r} = 3 \frac{y}{z} \cdot \frac{\rme^y-1}{\rme^z-1}. \label{case28:eq3}
\end{equation}
Rearranging (\ref{case28:eq2}) gives
\begin{equation}
\frac{1-r}{r} = \sqrt{\frac{9y^3}{2z^3}}. \label{case28:eq4}
\end{equation}
Combining (\ref{case28:eq4}) with (\ref{case28:eq3}) gives the equation
\begin{equation}
\frac{\rme^z-1}{\rme^y-1} = \sqrt{\frac{2z}{y}} \label{case28:eq5}
\end{equation}
that we can use to further restrict the possible values $z$ and $y$ can have at a maximum
for (\ref{eq:face7}).

In addition, we can see that (\ref{case28:eq5}) has exactly two non-negative solutions, one
when $z = 0$, and one when $z$ is between $\frac{y}{2}$ and $y$.  As a result, we know that
the positive solution has $z < y$.  Applying $z < y$ to (\ref{eq:kimplicitsMH}) gives $r < \a$, and
applying $z < y$ to (\ref{case28:eq1}) gives $\a < \frac{1}{4}$.  We can use this fact to further
reduce the possible values that must hold at a maximum of (\ref{eq:face7}).

The interval analysis program verifies that (\ref{eq:face7}) is no larger than $2(1-c)\ln 4 - 1\times10^{-4}$,  at each point
of $c \in \left[.67, 1\right)$, $\a \in \left(0, 1\right)$, and $r \in \left(\frac{2\a}{3c}, 1 - \frac{2(1-\a)}{3c}\right)$
where (\ref{case28:eq1}), (\ref{case28:eq2}), (\ref{case28:eq5}), and $\a < \frac{1}{4}$ all hold.
This result plus the results of Cases 1, 2, 4, 5, 10, 13, 15, and 18 proves that for each $c \in [.67, 1)$ there exists a 
positive constant $b = b(c)$ such that 
each point on this boundary is smaller that $2(1-c)\ln 4 - b(c)$.

\subsection{The Interval Analysis Program}
\label{sec:code}

The source code for the program is included with this paper on the arXiv.org site.
The program uses and should be compiled with the Profil/BIAS \cite{bias} libraries for
interval algorithms.  The Profil/BIAS libraries are available at
\begin{center}{\tt www.ti3.tu-harburg.de/Software/PROFILEnglisch.html}. \end{center}

For the program execution used to verify the cases of Lemma~\ref{lem:fmax}, the program was compiled with
version 2.0.8 of the Profil/BIAS libraries.
The program was compiled and
executed on an AMD Opteron 2350 processor running the Linux operating system, kernel version
2.6.27.56, and the program was compiled with the gcc compiler, version 4.3.

\end{document}